%% file: Geo_linear_DAEs.tex
\documentclass[review,onefignum,onetabnum]{siamart190516}

\input{format}

\ifpdf
\hypersetup{
  pdftitle={Geometric Analysis of Differential-Algebraic Equations via Linear Control Theory},
  pdfauthor={Y. Chen  and W. Respondek}
}
\fi

\begin{document}
\nolinenumbers
\maketitle

\begin{abstract}
We consider  linear differential-algebraic equations DAEs of the form $E\dot x=Hx$ and the Kronecker canonical form \textbf{KCF} \red{\cite{kronecker1890algebraische}} of the corresponding matrix pencils $sE-H$. We also consider linear control systems and \red{their} Morse canonical form \textbf{MCF} \cite{morse1973structural},\cite{molinari1978structural}.  For a linear DAE, a procedure named explicitation is proposed, \red{which attaches to any linear DAE a linear control system }defined up to \red{a coordinates change,} a feedback transformation and an output injection. \red{Then we compare subspaces associated to a DAE in a geometric way with those associated (also in a geometric way) to a control system, namely, we compare the Wong sequences of DAEs and invariant subspaces of control systems}. We prove that the \textbf{KCF} of linear DAEs and the \textbf{MCF} of control systems have a \red{perfect} correspondence and \red{that} their invariants are related. In this way, we connect the geometric analysis of linear DAEs with the classical \red{geometric} linear control theory. \red{Finally}, we propose a concept named internal equivalence for DAEs  and discuss its relation with internal regularity, i.e., the existence and uniqueness of solutions.
\end{abstract}



\section{Introduction}\label{Chap1sec1}
Consider a linear {differential-algebraic equation} DAE of the form
\begin{align}\label{DAE1}
\Delta: E\dot{x}=Hx,
\end{align}
where \red{$x\in \mathscr X\cong\mathbb{R}^n$} is called the ``generalized'' state, $E\in \mathbb{R}^{l\times n}$ and $H\in \mathbb{R}^{l\times n}$. \red{Throughout}, a linear DAE of form (\ref{DAE1}) \red{will be} denoted by $\Delta_{l,n}=(E,H)$ \red{or, shortly,} $\Delta$ and the corresponding matrix pencil of $\Delta$  by $sE-H$, which is a polynomial matrix of degree one. \red{A DAE $\Delta$ or a matrix pencil $sE-H$ is called regular if $l=n$ and $|sE-H|\not\equiv 0$.}

Terminologies as ``singular'', ``implicit'', ``generalized'' are frequently used to describe a DAE due to its difference from an ordinary differential equation ODE. Since the structure of DAE $\Delta$ is totally determined by the corresponding matrix pencil $sE-H$, it is useful to find a simplified form (a normal form or canonical form) for $sE-H$. Under predefined equivalence (see ex-equivalence of Definition \ref{ex-equivalence}), canonical forms as the Weierstrass  form \textbf{WF} \red{\cite{Weierstrass1868}} for regular matrix pencils and the Kronecker canonical form \cite{kronecker1890algebraische} (for details see \textbf{KCF} in Appendix \red{and \cite{gantmacher1959matrix}}) for more general matrix pencils have been proposed.
Note that in the paper, we will not distinguish the difference between the \textbf{KCF} of a matrix pencil $sE-H$ and the \textbf{KCF} of a DAE $\Delta$, since although \textbf{KCF} is introduced \red{for} matrix pencils, it is immediate to put the \textbf{KCF} of $sE-H$ into \red{the corresponding form for the DAE $\Delta$}.   

Geometric analysis of linear and nonlinear DAEs can be \red{found} in  \cite{lebret1994proportional,lewis1986survey,lewis1992tutorial,malabre1987more,malabre1989generalized,rabier1994geometric,reich1990geometrical,reich1991existence}. We highlight an important concept named the Wong sequences ($\mathscr V_i$ \red{and} $\mathscr W_i$ of Definition \ref{VWrealized}) for linear DAEs, which were first introduced in \cite{wong1974eigenvalue}. Connections between the Wong sequences with the \textbf{WCF} and the \textbf{KCF} have been recently \red{established in, respectively, 
	\cite{BERGER20124052} and \cite{Berger2012,berger2013addition}}.
In particular, invariant properties for the limits of the Wong sequences ($\mathscr V^*$ and $\mathscr W^*$ in Definition \ref{DAE invariances}) were used to obtain a triangular quasi-Kronecker form in \cite{Berger2012,berger2013addition}. Moreover, the authors of \cite{Berger2012,berger2013addition} show that some of the Kronecker indices can be calculated via the Wong sequences and the remaining ones can be derived from a modified version of the Wong sequences.   

On the other hand, consider a linear time-invariant control system of the following form
\begin{equation}\label{Ltis1}
\Lambda:\left\lbrace  \begin{array}{l}
\dot z = Az + Bu\\
y = Cz+Du,
\end{array} \right. 
\end{equation}
where $z\in \mathscr{Z} =\mathbb{R}^q$ is the system state, $u\in {\mathscr U}=\mathbb{R}^m$ represents the input and $y\in{\mathscr{Y}}= \mathbb{R}^p$ is the output. System matrices $A,B,C,D$ above are \red{constant and} of appropriate sizes.
We also consider \red{the prolongation of $\Lambda$} of the following form
\begin{align}\label{pro sys}
\mathbf{\Lambda}:\left\lbrace  \begin{array}{l}
\dot z = Az + Bu\\
\dot u = v\\
y = Cz + Du
\end{array} \right. \Leftrightarrow    \left\lbrace  \begin{array}{l}
\dot {\mathbf z} = \mathbf{A}\mathbf{z} + \mathbf{B}v\\
y = \mathbf{C}\mathbf{z},
\end{array} \right.
\end{align}
where  
$$
\mathbf{z}= \left[  \begin{matrix}
z\\
u
\end{matrix} \right] ,\ \ \mathbf{A} = \left[ \begin{matrix}
A&B\\
0&0
\end{matrix} \right],\ \ {\mathbf{B}} = \left[ \begin{matrix}
0\\
{I_m}
\end{matrix} \right], \ \ {\mathbf{C}} = \left[ \begin{matrix}
C&D
\end{matrix} \right].
$$
Denote a control system of form (\ref{Ltis1}) by ${\Lambda}_{q,m,p}=(A,B,C,D)$ or, simply, $\Lambda$ and denote the prolonged system (\ref{pro sys}) by $\mathbf{\Lambda}_{n,m,p}=(\mathbf{A},\mathbf{B},\mathbf{C})$, or shortly $\mathbf{\Lambda}$, where $n=q+m$. \red{Notice that there is a one-to-one correspondence between $\mathcal C^{\infty}$-solutions of (\ref{Ltis1}) and (\ref{pro sys}) (or a \red{one-to-one} correspondence between $\mathcal C^1$-solutions $(z(t),u(t))$ of (\ref{Ltis1}) and $\mathcal C^1$-solutions  $\mathbf{z}(t)$, given by  $\mathcal C^0$-controls \red{$v(t)$}, of (\ref{pro sys})).}

Two kinds of invariant subspaces have been studied for analyzing the structure of linear control systems, \red{see e.g. \cite{wonham1974linear,basile1992controlled}}. More specifically, the largest $(\mathbf{A},\mathbf{B})$-invariant subspace contained in $\ker \mathbf{C}$ (\red{denoted} $\mathbfcal{V}^*$ in Definition \ref{ControlsysinvariancewithoutD}), which is related \red{with} disturbance decoupling problems, and the smallest $(\mathbf{C},\mathbf{A})$-conditioned invariant subspace containing ${\rm Im\,} \mathbf{B}$ (\red{denoted} $\mathbfcal{W}^*$ in Definition \ref{ControlsysinvariancewithoutD}) which is related to controllability subspaces. With the help of these invariant subspaces, any control system can be brought \red{(see \cite{morse1973structural},\cite{molinari1978structural})} into its Morse canonical form (for details, see \textbf{MCF} in Appendix) under \red{the action of a} group of transformations consisting of coordinates changes, feedback transformations, and output injections. The \textbf{MCF} consists of four decoupled subsystems $MCF^1$, $MCF^2$, $MCF^3$, $MCF^4$, to which there correspond four sets of structure invariants (the Morse indices $\varepsilon'_i$, $\rho'_i$, $\sigma'_i$, $\eta'_i$ in the \textbf{MCF}) and these structure invariants are computable \red{with the help of} $\mathbfcal{V}^*$ and $\mathbfcal{W}^*$. Note that in \cite{morse1973structural}, only the triple $(A,B,C)$ \red{is} considered while in \cite{molinari1978structural}, the general case of 4-tuple $(A,B,C,D)$, with \red{a} nonzero matrix $D$, \red{is} studied. 

The first aim of the paper is to find a way to \red{relate} linear DAEs with linear control systems and find their \red{geometric} connections. In fact, we will show in the next section that \red{to} any linear DAE, we can attach a class of linear control systems defined up to a coordinates change, a feedback transformation and an output injection.  We call this attachment the explicitation of a DAE.
The second purpose of the paper is to distinguish two kinds of equivalences in linear DAEs theory, namely, internal equivalence and external equivalence.  \red{We will give the formal definition of external equivalence in Definition \ref{ex-equivalence}. Note that our notion of \red{extermal equivalence} of DAEs is different from the one introduced in \cite{willems1983input,kuijper1991minimality}, where ``systems are defined to be externally equivalent if their behaviors are the same''.} Actually, the external equivalence (also named strict equivalence \red{in} \cite{gantmacher1959matrix}) is widely considered in \red{the} linear DAEs literature. For example, the \textbf{KCF} of a DAE is actually a canonical form under  external equivalence, which is simply defined by all linear nonsingular transformations in the whole ``generalized'' state space of the DAE. However, since  solutions of a DAE exist only on a constrained (invariant) subspace, sometimes we only need to perform the analysis on \red{that} constrained subspace. This point of view motivates to introduce the notion of internal equivalence \red{and} to find normal forms not on the whole space but only on \red{that} constrained subspace.

The paper is organized as follows. In Section \ref{Chap1sec2}, we introduce the notations, define the external equivalence of two DAEs, and also the Morse equivalence of two control systems. In Section \ref{Chap1sec3}, we \red{explain how} to associate \red{to any} DAE \red{a class of control systems}. In Section \ref{Chap1sec4}, we describe \red{geometric}  relations of  DAEs and the attached control systems. In Section \ref{Chap1sec5}, we show that there exists a perfect \red{correspondence} between the \textbf{KCF} and the \textbf{MCF}, and \red{that} their invariants have direct relations. In Section \ref{Chap1sec6}, we introduce the \red{notion of} internal equivalence for DAEs \red{and then discuss} the internal regularity.  Section \ref{Chap1sec7} contains the proofs of \red{our} results and Section \ref{sec:conclusions} contains the conclusions of this paper. \red{Finally, in the Appendix we recall two basic canonical forms: the Kronecker canonical form \textbf{KCF} for DAEs and the Morse canonical form \textbf{MCF} for control systems.}   
\section{Preliminaries}\label{Chap1sec2}
We use the following notations in the present paper.
\begin{longtable}[l]{lll}
	$\mathbb{N}$ & & the set of natural numbers with zero and $\mathbb{N}^+=\mathbb{N}\backslash \{0 \}$\\
	$\mathbb{C}$ & & the set of complex numbers\\
	${\mathbb{R}^{n \times m}}$ & &the set of real valued matrices with $n$ rows and $m$ columns \\
	$\mathbb R[s]$ && the polynomial ring over $\mathbb R$ with indeterminate $s$ \\
	$Gl\left( {n,\mathbb{R}} \right)$ & & the group of nonsigular matrices of $\mathbb{R}^{n \times n}$\\
	${\rm rank\,} A$ & &the rank of \red{a linear map} $A$\\
	${\rm rank\,}_{\mathbb R[s]} (sE-H)$ && the rank of a polynomial \red{matrix} $sE-H$ over $\mathbb R[s]$	\\
	$\ker A$ & &the kernal of \red{a linear map} $A$\\	
	$\dim\, \mathscr A$ && the dimension of \red{a linear} space $\mathscr A$\\
	${\mathop{\rm Im\,}\nolimits} A$ && the image of \red{a linear map} $A$\\
	$\mathscr A/{\mathscr B}$ & & the  quotient of a vector space $\mathscr A$ by a subspace $\mathscr B\subseteq \mathscr A$ \\
	$I_n$ & & the identity matrix of size $n\times n$ for $n\in \mathbb{N}^+$\\
	$0_{n\times m}$ & & the zero matrix of size $n\times m$ for $n,m\in \mathbb{N}^+$\\
	${A^T}$ & &the transpose of a matrix $A$\\		
	${A^{-1}}$ && the inverse of a matrix $A$\\		
	${A\mathscr B}$ &&$\{Ax\,|\,x\in \mathscr B\} $, the image of $\mathscr B$ under  \red{a linear map} $A$\\
	${A^{-1}\mathscr B}$ &&$\{x\,|\,Ax\in \mathscr B\} $, the preimage of $\mathscr B$ under  \red{a linear map} $A$\\	
	${A^{-T}}\mathscr B$ & & $(A^T)^{-1}\mathscr B$\\
	$\mathscr A^{\bot}$	&& $\{ x\,|\,\forall a\in \mathscr A: x^Ta\!=\!0\}$
	
\end{longtable}
\red{Consider} a DAE $\Delta_{l,n}=(E,H)$, given by (\ref{DAE1}), \red{denoted shortly} by $\Delta$, and the corresponding matrix pencil $sE-H$. A \emph{solution}, or \emph{trajectory}, $x(t)$ of $\Delta$ is any \red{$\mathcal C^1$-}differentiable map $x:\mathbb{R}\rightarrow\mathscr X$  satisfying $E\dot x(t)=Hx(t)$. \blue{A trajectory starting from a point $x(0)=x^0$ is denoted by $x(t,x^0)$.}
\begin{defn}[external equivalence]\label{ex-equivalence}
	Two DAEs $\Delta_{l,n}=(E,H)$ and  $\tilde \Delta_{l,n}=(\tilde E,\tilde H)$ are called externally equivalent, shortly ex-equivalent, if there exist $Q\in Gl(l,\mathbb{R})$ and $P\in Gl(n,\mathbb{R})$ such that 
	\begin{align*}
	\begin{array}{l}
	\tilde{E}=QEP^{-1} \ \ {\rm and} \ \ \tilde{H}=QHP^{-1}.
	\end{array}
	\end{align*}
	We denote ex-equivalence of two DAEs as $ \Delta\mathop  \sim \limits^{ex} \tilde \Delta$, and ex-equivalence of the two corresponding matrix pencils as $sE-H \mathop  \sim \limits^{ex} s\tilde E-\tilde H$.
\end{defn}
If the ``generalized'' \blue{states} of $\Delta$ and $\tilde \Delta$ are $x$ and $\tilde x$, respectively, then $\tilde x=Px$ is, clearly, just a coordinate transformation. The following remark points out the relation of the ex-equivalence and \red{solutions of DAEs.}
\begin{rem}\label{ex-quitra}
	Ex-equivalence preserves trajectories, more precisely, if $ \Delta\mathop  \sim \limits^{ex} \tilde \Delta$ via $(Q,P)$, then any trajectory $x(t)$ of $\Delta$ satisfying $x(0)=x^0$, is mapped via $P$ into a trajectory $\tilde x(t)$ of $\tilde \Delta$ passing through $\tilde x^0=Px^0$. Moreover, if $x(t)$ is  a trajectory of $\Delta$, then $E\dot x(t)-Hx(t)=0$ and obviously $Q(E\dot x(t)-Hx(t))=0$ implying that $x(t)$ is \red{also} a trajectory of $QE\dot x=QHx$. The converse, however, is not true\red{:} even if  two DAEs have the same trajectories, they are not necessarily ex-equivalent, since the trajectories of DAEs are contained in a subspace $\mathscr M^*\subseteq \mathbb R^n$ (see Definition \ref{Max invariant subspace} of Section \ref{Chap1sec6}).
\end{rem}
\begin{defn} [Morse equivalence and Morse transformation]\label{Morse-eq}
	Two linear control systems $\Lambda_{q,m,p}=(A,B,C,D)$ and $\tilde\Lambda_{q,m,p}=(\tilde A,\tilde B,\tilde C,\tilde D)$ are called Morse equivalent, shortly M-equivalent, denoted by $\Lambda \mathop  \sim \limits^{M} \tilde \Lambda $, if there exist $T_s\in Gl(q,\mathbb{R})$, $T_i\in Gl(m,\mathbb{R})$, $T_o\in Gl(p,\mathbb{R})$, $F\in \mathbb{R}^{m\times q} $, $K\in \mathbb{R}^{q\times p}$ such that 
	\begin{align}\label{morse-trans}
	\left[ {\begin{matrix}
		{\tilde A}&{\tilde B}\\
		{\tilde C}&{\tilde D}
		\end{matrix}} \right] = \left[ \begin{matrix}
	{{T_s}}&{{T_s}K}\\
	0&{{T_o}}
	\end{matrix} \right]\left[ \begin{matrix}
	A&B\\
	C&D
	\end{matrix} \right]\left[ {\begin{matrix*}
		{T_s^{ - 1}}&0\\
		{FT_s^{ - 1}}&T_i^{ - 1}
		\end{matrix*}} \right].
	\end{align}
	\red{Any 5-tuple ${M}_{tran}=(T_s, T_i, T_o, F, K)$, is called a Morse transformation.}
\end{defn}
\begin{rem}\label{Rem:M-equi}
	(i) Apparently, in the above definition of \red{a} Morse transformation, $T_s$, $T_i$, $T_o$ are coordinates transformations in the, respectively, state space \red{$\mathscr{Z}$}, input space $\mathscr U$, and output space $\mathscr{Y}$, and  $F$ defines a state feedback and $K$ defines an output injection. Moreover, if we consider two control systems without outputs, denoted by $\Lambda_{q,m}=(A,B)$ and $\tilde\Lambda_{q,m}=(\tilde A,\tilde B)$, then the Morse equivalence reduces to the feedback equivalence, i.e., \red{the corresponding} system matrices satisfy $\tilde A=T_s(A+BF){T_s^{-1}}$ and $\tilde B=T_sB{T_i^{-1}}$.
	
	(ii) The feedback transformation $A\mapsto A+BF$ preserves all trajectories (although changes their parametrization with respect to controls). On the other hand, the output injection $A\mapsto A+KC$, \red{$B\mapsto B+KD$} preserves only those trajectories $x(t)$ \red{that satisfy $y(t)=Cx(t)+Du(t)=0$}. Finally, $A\mapsto T_sAT^{-1}_s$ maps trajectories into trajectories while $B\mapsto BT^{-1}_i$ re-parametrizes controls and $C\mapsto T_oC$ and $D\mapsto T_oD$ re-parametrize  outputs.
\end{rem}
\section{Implicitation of linear control systems and explicitation of linear DAEs}\label{Chap1sec3}
It is easy to \red{see} that, if for a  linear control system $\Lambda$, given by (\ref{Ltis1}), we require the output $y=Cz+Du$ to be identically zero, then $\Lambda$ can be seen as a DAE.
We call such an output zeroing procedure the \emph{implicitation} of a control system, which \red{can} be formalized \red{as follows}.
\begin{defn}[implicitation]\label{implicitation}
	For \red{a} linear control system $\Lambda_{q,m,p}\!=\!(A,B,C,D)$ on $\mathscr{Z}=\mathbb{R}^q$ with inputs in $\mathscr U=\mathbb{R}^m$ and outputs in $\mathscr{Y}=\mathbb{R}^p$,  by setting the output $y$ of $\Lambda$ to be zero, that is
	$$
	{\rm Impl}(\Lambda):\left\{ \begin{array}{l}
	\dot z= Az + Bu\\
	0 = Cz + Du,
	\end{array} \right.
	$$
	we define \red{the following} DAE \red{$\Delta^{Impl}$} with ``generalized'' states  \red{$(z,u)\in\mathbb{R}^{q+m}$}:
	\begin{align}\label{Impl}
	\Delta^{Impl}:\left[ \begin{matrix}
	I_q&0\\
	0&0
	\end{matrix}\right]\left[ \begin{matrix}
	\dot z\\
	\dot u
	\end{matrix}\right]= \left[ \begin{matrix}
	A&B\\
	C&D
	\end{matrix}\right]\left[ \begin{matrix}
	z\\
	u
	\end{matrix}\right].
	\end{align}
	\red{We call the procedure of output zeroing above the implicitation procedure, and the DAE given by (\ref{Impl}) will be called the implicitation of $\Lambda$ and denoted by $\Delta^{Impl}_{q+p,q+m}={\rm Impl}(\Lambda)$ or, shortly, $\Delta^{Impl}={\rm Impl}(\Lambda)$.}
\end{defn}
The converse procedure,  of associating a control system to a given DAE, is less \red{straightforward}, since the variables are expressed implicitly in DAEs. In order to understand the different roles of the variables in a DAE, take, for example, the nilpotent pencil $N_{\sigma}(s)$ of the \textbf{KCF} of  DAEs (see Appendix),  denote the corresponding variables by $x_1,...,x_{\sigma}$ and then the DAE is
$$
	\left[ \begin{matrix}
	{0}&1&{\dots}&{0}\\
	{0}& \ddots & \ddots &\vdots\\
	{\vdots}&{\ddots}& \ddots &1\\
	{0}&{\dots}&{0}&{0}
	\end{matrix} \right]\left[ \begin{matrix}
	{\dot x_1}\\
	\vdots \\
	{\dot x_{\sigma  - 1}}\\
	{\dot x_\sigma }
	\end{matrix}\right]=\left[ \begin{matrix}
	{x_1}\\
	\vdots \\
	{x_{\sigma  - 1}}\\
	{x_\sigma }
	\end{matrix}\right].  $$
It is easy to see that the last equation $ {x_\sigma }=0$ is an algebraic constraint which can be seen as \red{the} zero output of a control system. The variable $x_1$ is different from the others because it is free to be given any value and thus it performs like an input. The  variables $x_2,...,x_{\sigma-1}$ are constrained by a differential chain forming an ODE, so they can be seen as  states of a control system. Notice that in this case, replacing $\dot x_i=x_{i-1}$ by $\dot x_i=x_{i-1}+k_ix_{\sigma}$, for $2\le i\le \sigma$ and for any $k_i\in \mathbb R$ does not change \red{the solution of} the system because $x_{\sigma}=0$, which means \red{that} if we want to associate \red{to our DAE} a control system, the association is not unique. \mage{Below we generalize the above observations and show a way to attach a class of control systems to any given DAE.}
\begin{itemize}
	\item Consider a DAE $\Delta_{l,n}=(E,H)$, \red{given by} (\ref{DAE1}). Denote  ${\rm rank\,} E=q$, \red{define $p=l-q$ and $m=n-q$}. Choose a map
	$$
	P=\left[ {\begin{matrix}
		P_1\\
		P_2
		\end{matrix}} \right]\in Gl(n,\mathbb{R}),
	$$
	where $P_1\in \mathbb{R}^{q\times n}$, $P_2\in \mathbb{R}^{m\times n}$ such that $\ker P_1=\ker E$. 
	\item Define coordinates transformation
	$$
	\left[ {\begin{matrix}
		z\\
		\red{u}
		\end{matrix}} \right]
	=
	\left[ {\begin{matrix}
		P_1x\\
		P_2x
		\end{matrix}} \right]
	=
	\left[ {\begin{matrix}
		P_1\\
		P_2
		\end{matrix}} \right]x
	=
	Px.
	$$
	Then from $\ker P_1=\ker E$, we have $EP^{ - 1}= \left[ {\begin{matrix}
		{{E_0}}&0
		\end{matrix}} \right]$, where $E_0\in \mathbb{R}^{l\times q}$. Moreover, \red{since} $P$ is invertible,  it follows that ${\rm rank\,} E_0={\rm rank\,} E=q$. Thus via $P$, $\Delta$ is ex-equivalent to
	\begin{align*}
	\left[ {\begin{matrix}
		{{E_0}}&0
		\end{matrix}} \right]\left[  {\begin{matrix}
		{\dot z}\\
		{\dot u}
		\end{matrix}} \right]  = {H_0}\left[  {\begin{matrix}
		z\\
		u
		\end{matrix}} \right],
	\end{align*}
	where \red{$H_0=HP^{-1}$}. The variables $z$ are states (dynamical variables, their derivatives $\dot z$ are
	present) and $u$ are controls (enter statically into the system).
	\item Since ${\rm rank\,}E_0=q$, there exists $Q_0\in Gl(l,\mathbb{R})$ such that $Q_0E_0=\left[ {\begin{matrix}
		{E_0^1}\\
		0
		\end{matrix}} \right]$, where $E_0^1\in Gl(q,\mathbb{R})$. Thus via $(Q_0,P)$, $\Delta$ is ex-equivalent to
	$$
	\left[ {\begin{matrix}
		{E_0^1}&0\\
		0&0
		\end{matrix}} \right]\left[ {\begin{matrix}
		\dot z\\
		\dot u
		\end{matrix}} \right]=\left[ {\begin{matrix}
		A_0&B_0\\
		C_0&D_0
		\end{matrix}} \right]\left[ {\begin{matrix}
		z\\
		u
		\end{matrix}} \right],
	$$
	where \red{$Q_0{H_0}=\left[ {\begin{matrix}
		{{A_0}}&{{B_0}}\\
		{{C_0}}&{{D_0}}
		\end{matrix}} \right] $}, $A_0\in \mathbb{R}^{q\times q}, B_0\in\mathbb{R}^{q\times m}, C_0\in\mathbb{R}^{p\times q}, D_0\in\mathbb{R}^{p\times m} $.
	\item  
	Finally, via $Q_1=\left[ {\begin{matrix}
		{({E_0^1})^{-1}}&0\\
		0&I_p
		\end{matrix}} \right]$, we bring the above DAE into
	\begin{equation}\label{Ltis2}
	\left[ {\begin{matrix}
		I_q&0\\
		0&0
		\end{matrix}} \right]\left[ {\begin{matrix}
		\dot z\\
		\dot u
		\end{matrix}} \right]=\left[ {\begin{matrix}
		A&B\\
		C&D
		\end{matrix}} \right]\left[ {\begin{matrix}
		z\\
		u
		\end{matrix}} \right],
	\end{equation}
	where $A=(E_0^1)^{-1}A_0$, $B=(E_0^1)^{-1}B_0$, $C=C_0$, $D=D_0$.
	\item Therefore, the DAE $\Delta$ is ex-equivalent (via $P$ and $Q=Q_1Q_0$) to (\ref{Ltis2}) and the latter is \red{the} control system
	$$
	\Lambda:\left\{ \begin{array}{c@{\,}l}
	\dot z &= Az + Bu\\
	y &= Cz+Du,
	\end{array}\right.
	$$
	together with the constraint $y=0$, that is, \blue{$\Delta\mathop \sim \limits^{ex}\Delta^{Impl}={\rm Impl}(\Lambda) $}.
\end{itemize}
\red{Let us give a few comments on the above construction:}

(i) The map $P=\left[ {\begin{matrix}
		P_1\\
		P_2
		\end{matrix}} \right]$ defines state variables $z=P_1x$
	\red{as} coordinates on the state space $\red{\mathscr Z=\mathbb R^{n}/\ker E}$ isomorphic to $\mathbb R^q$
	and control variables $u=P_2x$ \red{as} coordinates on $\mathscr U\cong \ker E\cong  \mathbb{R}^m$. The output variables $y$ are coordinates on $\mathscr Y\cong\mathbb{R}^{l}/{{\rm Im\,} E}\cong \mathbb{R}^{p}$ and define the output map via $y=Cz+Du$.
	
(ii) Choose other coordinates $(z',u')$ \red{given} by $z'=P'_1x$ and $u'=P'_2x$ such that $\ker P'_1=\ker E=\ker P_1$, then
	\begin{align}
	\left\{ {\begin{array}{*{20}{l}}
		{ z'=T_sz}\\
		{u'=F'z+T_iu},
		\end{array}} \right.
	\end{align}
	where $T_s\in Gl(n,\mathbb{R})$ and $F'\in \mathbb{R}^{m\times n}$, $ T_i\in Gl(m,\mathbb{R})$. Clearly, $z'=T_sz$ is another set of coordinates on the state space $\mathbb R^{n}/\ker E$ and $u'=F'z+T_iu$ is a \emph{state feedback transformation}.
	
(iii) The output $y$ \red{takes values} in the quotient space $\mathbb{R}^{l}/{{\rm Im\,} E}$.
	Since $y=Cz+Du=0$, we can add $y$ to the dynamics  \red{without changing} solutions of the system on \red{the subspace} $\{y=0\}$. \red{Together with a state transformation $z'=T_sz$ \red{and an output transformation $y'=T_oy$}, it} results in a triangular transformation (\emph{output injection}) of the system
	\begin{align}
	\left[\begin{matrix}
	\dot z'\\
	y'
	\end{matrix} \right]=\left[ {\begin{matrix}
		{{T_s}}&K'\\
		0&{{T_o}}
		\end{matrix}} \right]\left[\begin{matrix}
	\dot z\\
	y
	\end{matrix} \right]=\left[ {\begin{matrix}
		{{T_s}}&K'\\
		0&{{T_o}}
		\end{matrix}} \right]\left[\begin{matrix}
	A&B\\
	C&D
	\end{matrix} \right]\left[\begin{matrix}
	z\\
	u
	\end{matrix} \right],
	\end{align}
	where $K'\in \mathbb{R}^{n\times p}$, $ T_o\in Gl(p,\mathbb{R})$.

In view of the above analysis, the \red{non-uniqueness} of the construction leads to a control system defined up to a coordinates change, a feedback transformation and an output injection,  which is actually, a class of control systems. 
\begin{defn}[explicitation]\label{Def:QPexpl}
	Given a DAE $\Delta_{l,n}=(E,H)$, \red{there always exist $Q\in Gl(l,\mathbb R)$ and $P\in Gl(n,\mathbb R)$ such that 
		\begin{align}\label{Eq:ex-equi}
		QEP^{-1}=\left[ \begin{matrix}
		I_q&0\\0&0
		\end{matrix}\right].
		\end{align}
		The control system $\Lambda$, given by $\Lambda_{q,m,p}=(A,B,C,D)$, where \red{$ QHP^{-1}=\left[ \begin{matrix}
			A&B\\C&D
		\end{matrix}\right]$}, is called the $(Q,P)$-explicitation of $\Delta$. The class of all $(Q,P)$-explicitations, corresponding to all $Q\in Gl(l,\mathbb R)$ and $P\in Gl(n,\mathbb R)$, will be called the explicitation class of $\Delta$ and denoted by ${\rm Expl}(\Delta)$.} If a particular control system $\Lambda$ \red{belongs} to the explicitation class ${\rm Expl}(\Delta)$ of $\Delta$, \red{we will write} $\Lambda\in {\rm Expl}(\Delta)$.
\end{defn}
\begin{rem}\label{explicitation}
	The implicitation \red{${\rm Impl}(\Lambda)$} of a given control system $\Lambda$ is a unique DAE $\Delta^{Impl}$, given by $(\ref{Impl})$. \red{The explicitation ${\rm Expl}(\Delta)$ of a given DAE $\Delta$} is\red{, however,} a control system defined up to a coordinates change, a feedback transformation, and an output injection, \red{that is,} a class of control systems.
\end{rem}
\begin{thm}\label{main theorem} 
	\begin{itemize}
\item[(i)]  Consider  a  DAE $\Delta=(E,H)$ and a control system $\Lambda=(A,B,C,D)$. Then $\Lambda\in {\rm Expl}(\Delta)$ if and only if \blue{$\Delta \mathop  \sim \limits^{ex} \Delta^{Impl}$, where $\Delta^{Impl}={\rm Impl}(\Lambda)$}. \blue{More specifically, $\Lambda$ is the $(Q,P)$-explicitation of $\Delta$ if and only if $\Delta \mathop  \sim \limits^{ex} \Delta^{Impl}$ via $(Q,P)$.}
\item[(ii)]  Given two DAEs $\Delta=(E,H)$ and $\tilde \Delta=(\tilde E,\tilde H)$, \red{choose} two control systems $\Lambda\in {\rm Expl}(\Delta)$ and ${\tilde {\Lambda}}\in  {\rm Expl}(\tilde \Delta) $. Then  $\Delta\mathop  \sim \limits^{ex} \tilde \Delta $ if and only if $\Lambda\mathop  \sim \limits^{M}\tilde {\Lambda}$.

\item[(iii)]  \red{Consider} two control systems $\Lambda=(A,B,C,D)$ and $\tilde \Lambda=(\tilde A, \tilde B, \tilde C, \tilde D)$. Then $\Lambda \mathop  \sim \limits^{M} \tilde \Lambda $ if and only if $\Delta^{Impl} \mathop  \sim \limits^{ex} \tilde \Delta^{Impl} $, where $\Delta^{Impl}={\rm Impl}(\Lambda)$ and $\tilde \Delta^{Impl}={\rm Impl}(\tilde \Lambda)$.  	 
	\end{itemize}
\end{thm}
The proof is given in Section \ref{Pf:main theorem}.
\begin{rem}
	Theorem \ref{main theorem} \red{describes} relations of DAEs and control systems, \red{which we illustrate} in  Figure \ref{Fig:ex and im}. \red{We} conclude \red{that} Morse equivalent control systems (and only such) give, via \emph{implicitation,} ex-equivalent DAEs. Furthermore, \emph{explicitation} is a universal procedure of producing control systems from a DAE and ex-equivalent DAEs produce Morse equivalent control systems.
	\tikzstyle{process} = [rectangle, minimum width=2cm, minimum height=0.8cm, text centered, draw=black]
	\tikzstyle{process1} = [rectangle, minimum width=2cm, minimum height=0.4cm, text centered]
	\tikzstyle{arrow} = [thick,double, distance=1pt,<->,>=stealth]
	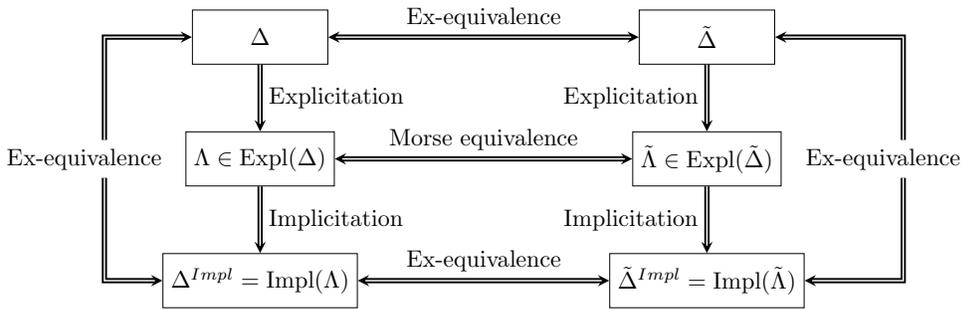
\begin{figure}[h]
		\setlength{\belowcaptionskip}{-0.3 cm}  
		\begin{center}
		\scalebox{.9}{	\begin{tikzpicture}[node distance=1.6cm]
			\node[process](A11){$\Delta$};
			\node[process, below of = A11, yshift = -0.2cm](A21){$\Lambda\in {\rm Expl}(\Delta)$};
			\node[process, right of = A11, xshift = 5cm](A12){$\tilde \Delta$};
			\node[process, below of = A21, yshift = -0.2cm](A31){$\Delta^{Impl}={\rm Impl}(\Lambda)$};
			\node[process1, left of = A21, xshift = -1cm](A20){Ex-equivalence};
			\node[process, right of = A21, xshift =5cm](A22){$\tilde \Lambda\in {\rm Expl}(\tilde \Delta)$};
			\node[process1, right of = A22, xshift = 1cm](A23){Ex-equivalence};
			\node[process, right of = A31, xshift =5cm](A32){$\tilde\Delta^{Impl}={\rm Impl}(\tilde\Lambda)$};
			\coordinate (point1) at (-3cm, -6cm);
			\draw [arrow][->] (A11) -- node[right]{Explicitation}(A21);
			\draw [arrow] (A11) -- node[above]{Ex-equivalence}(A12);
			\draw [arrow][->] (A12) -- node[left]{Explicitation}(A22);
			\draw [arrow] (A21) -- node[above]{Morse equivalence}(A22);
			\draw [arrow][->] (A21) -- node[right]{Implicitation}(A31);
			\draw [arrow] (A31) -- node[above]{Ex-equivalence}(A32);
			\draw [arrow][->] (A22) -- node[left]{Implicitation}(A32);
			\draw [arrow][->](A20) --($(A20.north) + (0.3,0)$) |-(A11);
			\draw [arrow][->](A20) --($(A20.south) + (0.3,0)$) |-(A31);
			\draw [arrow][->](A23) --($(A23.north) + (0.3,0)$) |-(A12);
			\draw [arrow][->](A23) --($(A23.south) + (0.3,0)$) |-(A32);
			\end{tikzpicture}}
			\caption{Explicitation of DAEs and implicitation of control systems}
			\label{Fig:ex and im} 
		\end{center}
	\end{figure}
\end{rem}
\section{Geometric connections between DAEs and control systems}\label{Chap1sec4}
\red{The Wong sequences  \cite{wong1974eigenvalue} of a DAE are defined as follows.}
\begin{defn}\label{VWrealized}
	For a DAE $\Delta_{l,n}=(E,H)$, \red{its} Wong sequences \red{are defined} by	
	\begin{align}
	\mathscr V_0&=\mathbb{R}^n, \ \ \ \mathscr V_{i+1}=H^{-1}E\mathscr V_{i},\ \ i\in \mathbb{N},\label{Vrealized}\\
	\mathscr W_0&=\{0\}, \ \  \mathscr W_{i+1}=E^{-1}H\mathscr W_{i},\ \ i\in \mathbb{N}.\label{Wrealized}
	\end{align}	
\end{defn}
\begin{rem}\label{rem:wong}
	The Wong sequences $\mathcal V_i$ and $\mathcal W_i$ satisfy
	\begin{align}\label{prop of V}
	\begin{array}{l}
	{\mathscr V_0} \supsetneq {\mathscr V_1} \supsetneq  \cdots  \supsetneq {\mathscr V_{{k^*}}} = {\mathscr V_{{k^*} + j}} = {\mathscr V^*} = {H^{ - 1}} {E{\mathscr V^*}}  \supseteq \ker H, \ j\in \mathbb{N},\\
	{\mathscr W_0} \subseteq \ker E = {\mathscr W_1} \subsetneq  \cdots  \subsetneq {\mathscr W_{{l^*}}} = {\mathscr W_{{l^*} + j}} = {\mathscr W^*} = {E^{ - 1}} {H{\mathscr W^*}}, \ j\in \mathbb{N}.
	\end{array}
	\end{align}
\end{rem}
\red{We now give a definition of invariant subspaces for linear DAEs.}
\begin{defn}\label{DAE invariances}
	For \red{a} DAE $\Delta_{l,n}=(E,H)$, a subspace $\mathscr V \subseteq \mathbb{R}^n$ is called  $(H^{-1},E)$  -invariant  if $\mathscr V$ satisfies $\mathscr V=H^{-1} E\mathscr V$; a subspace $\mathscr W\subseteq \mathbb{R}^n$ is called  $(E^{-1},H)$-invariant if $\mathscr W$ satisfies $\mathscr W= E^{-1}H\mathscr W$. 
\end{defn}
\red{Denote by $\mathscr V^*$ the largest $(H^{-1},E)$-invariant subspace of $\mathbb{R}^n$ and by $\mathscr W^*$ the smallest $(E^{-1},H)$-invariant subspace of $\mathbb{R}^n$. Using the \mage{same} symbols $\mathscr V^*$ and $\mathscr W^*$ \mage{as those} for the limits of Wong sequences (see Remark \ref{rem:wong}) is justified by the following.}
\begin{pro}\label{Pro:VW}
	(i) For a \red{DAE} $\Delta_{l,n}=(E,H)$, the largest $(H^{-1},E)$-invariant subspace $\mathscr V^*$ and the smallest $(E^{-1},H)$-invariant subspace $\mathscr W^*$ exist and are given, respectively, by 
	\begin{align*}
	{\mathscr V^*}={\mathscr V_{{k^*}}} \ \ \ and \ \ \ {\mathscr W^*}={\mathscr W_{{l^*}}},
	\end{align*}
	where $k^*$ is the smallest integer such that  ${\mathscr V_{{k^*}}}={\mathscr V_{{k^*}+1}}$ and $l^*$ is the smallest interger  such that  ${\mathscr W_{{l^*}}}={\mathscr W_{{l^*}+1}}$;
	
	(ii) $\mathscr V^*$ is also the largest subspace such that $H\mathscr V^*\subseteq E\mathscr V^*$, however, $\mathscr W^*$ is not necessarily the smallest subspace such that $E\mathscr W^*\subseteq H\mathscr W^*$.
\end{pro}
The proof is given in Section \ref{Pf:Pro:VW}. We now review the notions of invariant subspaces in linear control theory. We consider two cases depending on whether the control system is strictly proper ($D$ is zero or not). \red{We will use the \red{bold-notations} for the strictly proper case $D=0$, since throughout it applies to the prolongation system (\ref{pro sys}), which we denote by bold symbols.}
\begin{defn}\label{ControlsysinvariancewithoutD}
	For a control system $\mathbf{\Lambda}_{n,m,p}=(\mathbf{A},\mathbf{B},\mathbf{C})$, a subspace $\mathbfcal{V} \subseteq \mathbb{R}^n$ is called \red{an} $(\mathbf{A},\mathbf{B})$-controlled invariant subspace if $\mathbfcal{V}$ satisfies 
	\[\mathbf{A}\mathbfcal{V}\subseteq \mathbfcal{V}+{\rm Im\,} \mathbf{B}\]
	and a subspace $\mathbfcal{W}\subseteq \mathbb{R}^n$ is called \red{a} $(\mathbf{C},\mathbf{A})$-conditioned invariant subspace if $\mathbfcal{W}$ satisfies 
	\[ \mathbf{A}(\mathbfcal{W}\cap \ker \mathbf{C})\subseteq \mathbfcal{W}.\]
	Denote by $\mathbfcal{V^*}$ the largest $(\mathbf{A},\mathbf{B})$-controlled invariant subspace contained in $\ker \mathbf{C}$ and by $\mathbfcal{W^*}$  the smallest $(\mathbf{C},\mathbf{A})$-conditioned invariant subspace containing ${\rm Im\,} \mathbf{B}$.
\end{defn}
\red{The following fundamental lemma shows that $\mathbfcal{V^*}$, $\mathbfcal{W^*}$  exist and they can be calculated via the sequences of subspaces $\mathbfcal{V}_i$, $\mathbfcal{W}_i$ given below.}
\begin{lem}[\cite{wonham1974linear},\cite{basile1992controlled}]\label{tildeVWreali}
	Initialize $\mathbfcal{V}_0=\mathbb{R}^n$ \red{and,} for $i\in 
	\mathbb{N}$, define inductively
	\begin{align}\label{tildeVreali}
	{\kern 30pt}\mathbfcal{V}_{i+1}= \ker \mathbf{C}\cap \mathbf{A}^{-1}(\mathbfcal{V}_{i}+{\rm Im\,} \mathbf{B}).
	\end{align}
	Initialize $\mathbfcal{W}_0={0}$ and, for $i\in 
	\mathbb{N}$, define inductively
	\begin{align}\label{tildeWreali}
	\mathbfcal{W}_{i+1}= \mathbf{A}(\mathbfcal{W}_{i}\cap \ker \mathbf{C} )+{\rm Im\,} \mathbf{B}.
	\end{align}	
	Then there exist $\mathbf{k}^*\leq n$ and $\mathbf{l}^*\leq n$ such that
	\begin{align*}
	\begin{array}{l}
	\mathbfcal{V}_0\supseteq \ker\mathbf{C} =\mathbfcal{V}_1 \supsetneq  \cdots  \supsetneq \mathbfcal{V}_{\mathbf{k}^*} = \mathbfcal{V}_{{\mathbf{k}^*} + j} =\mathbfcal{V^*}=\ker \mathbf{C}\cap \mathbf{A}^{-1}( \mathbfcal{V^*}+{\rm Im\,}\mathbf{B}), \ \ j\in \mathbb{N},\\
	\mathbfcal{W}_0 \subseteq  {\rm Im\,} \mathbf{B} \!=\! \mathbfcal{W}_1 \subsetneq  \cdots  \subsetneq\mathbfcal{W}_{\mathbf{l}^*} = \mathbfcal{W}_{{\mathbf{k}^*} + j} =\mathbfcal{W^*}= \mathbf{A}(\mathbfcal{W}^*\cap \mathscr \ker \mathbf{C} )+{\rm Im\,}\mathbf{B}, \  j\in \mathbb{N}.
	\end{array}
	\end{align*}
\end{lem}
\mage {Note that  $\mathbf{k}^*$ and $\mathbf{l}^*$ of Lemma \ref{tildeVWreali} and  $k^*$ and $l^*$ of Remark \ref{rem:wong} are, in general,  not the same (except for some cases described later, see Proposition \ref{relation invariant subs}, in which they coincide).}  It is well-known (see e.g., \cite{wonham1970decoupling},\cite{wonham1974linear},\cite{basile1992controlled}) that $\mathbfcal{V}$ is an $(\mathbf{A},\mathbf{B})$-controlled invariant subspace if and only if there exists $\mathbf{F}\in \mathbb{R}^{m\times n}$ such that $(\mathbf{A}+\mathbf{B}\mathbf{F})\mathbfcal{V}\subseteq\mathbfcal{V}$ and $\mathbfcal{W}$ is a $(\mathbf{C},\mathbf{A})$-conditioned invariant subspace if and only if there exists $\mathbf{K}\in \mathbb{R}^{n\times p}$ such that $(\mathbf{A}+\mathbf{K}\mathbf{C})\mathbfcal{W}\subseteq \mathbfcal{W}$. \red{For a control system which is not strictly proper ($D$ is not zero),} following \red{Definitions 1--4} of \cite{molinari1978structural}, we use a generalization of that characterization of invariant subspaces.
\begin{defn}\label{ControlsysinvariancewithD}
	For $\Lambda_{q,m,p}=(A,B,C,D)$, a subspace $\mathcal{V} \subseteq \mathbb{R}^q$ is called \red{a} null-output $(A,B)$-controlled invariant subspace if there exists $F\in \mathbb{R}^{m\times q}$ such that 
	\[(A+BF)\mathcal{V}\subseteq\mathcal{V} \ \ \ {\rm and} \ \ \ (C+DF)\mathcal{V}=0,\]
	and \red{for any such $\mathcal V$, the subspace $\mathcal{U} \subseteq \mathbb{R}^m$ given by
		$$\mathcal{U}=(B^{-1}\mathcal{V})\cap \ker D,$$ 
		is called a null-output $(A,B)$-controlled invariant input subspace.}
	Denote by $\mathcal{V^*}$ (\red{resp.} $\mathcal{U^*}$)  the largest null-output $(A,B)$ controlled invariant subspace (\red{resp.} input subspace).	 \\ \indent A subspace $\mathcal{W} \subseteq \mathbb{R}^q$ is called \red{an} unknown-input $(C,A)$-conditioned invariant subspace if there exists $K\in \mathbb{R}^{q\times p}$ such that 
	$$(A+KC)\mathcal{W}+(B+KD) {\mathscr U}=\mathcal{W},$$ 
	and \red{ for any such $\mathcal W$, the subspace $\mathcal{Y} \subseteq \mathbb{R}^p$ given by 
		$$ \mathcal{Y}=C\mathcal{W}+ D{\mathscr U},$$ 
		is called \red{an} unknown-input $(C,A)$-conditioned invariant output subspace.} Denote by $\mathcal{W^*}$ (\red{resp.} $\mathcal{Y^*}$) the smallest unknown-input $(C,A)$-conditioned invariant subspace (\red{resp.} output subspace).
\end{defn}
\red{The following lemma  shows that $\mathcal{V^*}$, $\mathcal{U^*}$, $\mathcal{W^*}$, $\mathcal{Y^*}$ exist and provides a calculable algorithm to find them.}
\begin{lem}[\cite{molinari1976strong}]\label{barVWrealized}
	Initialize $\mathcal{V}_0=\mathbb{R}^q$, and for $i\in 
	\mathbb{N}$, define inductively
	\begin{align}\label{barV}
	{\mathcal{V}_{i + 1}} = {\left[ {\begin{matrix}
			A\\
			C
			\end{matrix}} \right]^{ - 1}}\left(  {\left[ {\begin{matrix}
			I\\
			0
			\end{matrix}} \right]{\mathcal{V}_i} + {\rm Im\,}\left[ {\begin{matrix}
			{B}\\
			{D}
			\end{matrix}} \right]} \right) 
	\end{align}
	and $\mathcal{U}_i\subseteq {\mathscr U}$ for $i\in \mathbb{N}$ are given by
	\begin{align}\label{barU}
	{\mathcal{U}_i} = {\left[ {\begin{matrix}
			B\\
			D
			\end{matrix}} \right]^{ - 1}}\left[ {\begin{matrix}
		{{\mathcal{V}_i}}\\
		0
		\end{matrix}} \right].
	\end{align}
	Then  $\mathcal{V}^*=\mathcal{V}_q$ and $\mathcal{U}^*=\mathcal{U}_q$.
	
	Initialize $\mathcal{W}_0=\{0\}$, and for $i\in 
	\mathbb{N}$, define inductively
	\begin{align}\label{barW}
	{\mathcal{W}_{i + 1}} = \left[ {\begin{matrix}
		A&B
		\end{matrix}} \right]\left(  {\left[ {\begin{matrix}
			{{\mathcal{W}_i}}\\
			{\mathscr U}
			\end{matrix}} \right] \cap \ker \left[ {\begin{matrix}
			C&D
			\end{matrix}} \right]} \right) 
	\end{align}
	and $\mathcal{Y}_i\subseteq{\mathscr{Y}}$ for $i\in \mathbb{N}$ are given by
	\begin{align}\label{barY}
	{{\mathcal{Y}}_i} = \left[ {\begin{matrix}
		C&D
		\end{matrix}} \right]\left[ {\begin{matrix}
		{{\mathcal{W}_i}}\\
		\mathscr U
		\end{matrix}} \right].
	\end{align}
	Then $\mathcal{W}^*=\mathcal{W}_q$ and $\mathcal{Y}^*=\mathcal{Y}_q$.
\end{lem} 
\begin{rem}
(i) Lemma \ref{barVWrealized}  generalizes the results of Lemma \ref{tildeVWreali} and, if $D=0$, Lemma \ref{barVWrealized} reduces to Lemma~\ref{tildeVWreali};

(ii)  Even if $\Lambda$ is not strictly proper (if $D\neq0$), the prolonged system $\mathbf{\Lambda}$ always is \mage{and thus} throughout we will use $\mathcal V^*$, $\mathcal U^*$,  $\mathcal W^*$ and $\mathcal Y^*$ for $\Lambda$, \red{and} $\mathbfcal{V}^*$ and $\mathbfcal{W}^*$ for $\mathbf{\Lambda}$.
\end{rem}
Throughout the paper, for ease of notation, we will write $\mathscr V_i(\Delta)$ \red{to indicate that $\mathscr V_i$ is calculated for $\Delta$}, \red{similarly for $\mathcal V_i(\Lambda)$, $\mathbfcal{V}_i(\mathbf{\Lambda})$,} and all other subspaces defined in this section.  Now we give the main results of this section. 
\begin{pro}[geometric subspaces relations]\label{relation invariant subs}
Given a DAE $\Delta_{l,n}\!=\!(E,H)$, a $(Q,P)$-explicitation $\Lambda\!=\!(A,B,C,D)\in {\rm Expl}(\Delta)$, \red{and} the prolongation $\mathbf{\Lambda}\!=\!(\mathbf{A},\mathbf{B},\mathbf{C})$ of $\Lambda$,  consider the limits of the Wong sequences $\mathscr V^*$ and $\mathscr W^*$ of $\Delta$ and \red{of} $\Delta^{Impl}={\rm Impl}(\Lambda)$, given by Definition \ref{DAE invariances}, the invariant subspaces $\mathcal{V}^*$ and $\mathcal{W}^*$ of $\Lambda$, given by Definition \ref{ControlsysinvariancewithD}, and the invariant subspaces $\mathbfcal{V}^*$ and $\mathbfcal{W}^*$ of $\mathbf{\Lambda}$, given by Definition \ref{ControlsysinvariancewithoutD}.  Then the following holds
	
	(i) $P\mathscr V^*(\Delta)=\mathscr V^*(\Delta^{Impl})=\mathbfcal{V}^*(\mathbf{\Lambda})= \left[ {\begin{matrix}
		A&B\\
		C&D
		\end{matrix}} \right] ^{-1}\left[ {\begin{matrix}
		{\mathcal{V}^*({\Lambda})}\\
		{0}
		\end{matrix}} \right],$
	
	(ii)
	$P\mathscr W^*(\Delta)=\mathscr W^*(\Delta^{Impl})=\mathbfcal{W}^*(\mathbf{\Lambda})=\left[ {\begin{matrix}
		I_q&0\\
		0&0
		\end{matrix}} \right] ^{-1}\left[ {\begin{matrix}
		{\mathcal{W}^*({\Lambda})}\\
		{0}
		\end{matrix}} \right].$ 
\end{pro}
The proof is given in Section \ref{Pf:relation invariant subs}. 
\begin{rem}
	(i) The limits $\mathscr V^*$ and $\mathscr W^*$ of the Wong sequences coincide for $\Delta$ and $\tilde \Delta$ that are ex-equivalent via \red{$(P,Q)$, where $P=I_n$ and $Q$ is arbitrary}, and do not depend on $Q$. On the other hand, the system $\Lambda$, being a $(Q,P)$-explicitation of $\Delta$, depends on both $P$ and $Q$ (and so does its prolongation $\mathbf \Lambda$)  but the invariant subspaces $\mathcal V^*(\Lambda)$ and $\mathcal W^*(\Lambda)$ depend on $P$ only.
	
	(ii) \red{Some particular relations} between the Wong sequences of DAEs and the invariant subspaces of control systems is given in Theorem 5 of \cite{costantini2013regularity}, which can be seen as a corollary of Proposition \ref{relation invariant subs}.
\end{rem}

Now \red{we will study various} dualities of geometric subspaces by analyzing the dual system. The duality of the subspaces $\mathbfcal V^*$ and $\mathbfcal{W}^*$ is \red{well-known and} studied in \cite{wonham1970decoupling},\cite{morse1973structural},\cite{basile1992controlled}. Similarly, properties of the subspaces $\mathcal V^*, \mathcal W^*, \mathcal U^*, \mathcal Y^*$ for the dual system of a control system are \red{analyzed} in  \cite{molinari1976strong} and \cite{molinari1978structural}.
In \cite{Berger2012}, it is proved that the Wong sequences of the transposed matrix pencils have relations with the original matrix pencils. In the following,
we will show that \red{all} these results can be connected by the explicitation of DAEs. Together with $\Delta$ we consider its dual $\Delta^d_{n,l}=(E^T,H^T)$ of the  form:
\begin{align*}
E^T\dot{x}^d=H^Tx^d,
\end{align*}
where $x^d\in \mathbb{R}^l$ is \red{the ``generalized'' state of the dual system}.
\begin{pro}\label{explifdual}
	Consider \red{a DAE} $\Delta$ and its dual $\Delta^d$. Then $\Lambda\!=\!(A,B,C,D)\in {\rm Expl}(\Delta)$ if and only if $\Lambda^d=(A^T,C^T,B^T,D^T)\in {\rm Expl}(\Delta^d)$.
\end{pro}
\begin{proof}
	\red{For any invertible matrices $Q$ and $P$ of appropriate sizes that yield (\ref{Eq:ex-equi}), we have the following equivalence:}
	\begin{align*}
	{Q}\left( {sE - H} \right)P^{-1} =\left[ {\begin{matrix}
		{s{I_q} - {A}}&{{-B}}\\
		{{-C}}&{{-D}}
		\end{matrix}} \right]\Leftrightarrow  {P^{ - T}}\left( {s{E^T} - {H^T}} \right){Q^T}=\left[ {\begin{matrix}
		{s{I_q} - {{A}^T}}&{{{{-C}}^T}}\\
		{{{{-B}}^T}}&{-D}^T
		\end{matrix}} \right] .
	\end{align*}
	Suppose $\Lambda\in {\rm Expl}(\Delta)$, \red{then by Theorem \ref{main theorem}(i),} there exist $Q\in Gl(l,\mathbb{R})$ and $P\in Gl(n,\mathbb{R})$, such that the left-hand side of the above equivalence holds. Then from the right-hand side we can see \red{that}
	$\Lambda^d\in {\rm Expl}(\Delta^d)$.
	
	Conversely, suppose $\Lambda^d\in {\rm Expl}(\Delta^d)$. Then there exist $P^{-T}\in Gl(n,\mathbb{R})$ and $Q^{-T}\in Gl(l,\mathbb{R})$ such that the right-hand side of the above equivalence holds, then from the left-hand side we can see \red{that} $\Lambda\in {\rm Expl}(\Delta)$. 
\end{proof}
\begin{pro}[subspaces of the dual system]\label{Subspaces of the dual system}
	For $\Delta=(E,H)$ and its dual $\Delta^d=(E^T,H^T)$, consider the subspaces $\mathscr V^*$ and $\mathscr W^*$ of Definition \ref{DAE invariances}. For two control systems $\Lambda=(A,B,C,D)\in {\rm Expl}(\Delta)$ and the dual \red{$\Lambda^d$} of $\Lambda$, \red{given} by $\Lambda^d=(A^T,C^T,B^T,D^T)$, consider the subspaces $\mathcal{V}^*$ and $\mathcal{W}^*$ of  Definition \ref{ControlsysinvariancewithD}. Finally, for the prolongation of $\Lambda$, denoted by $\mathbf{\Lambda}=(\mathbf{A},\mathbf{B},\mathbf{C})$ and for the dual $\mathbf{\Lambda}^d$ of $\mathbf{\Lambda}$, \red{given} by $ \mathbf{\Lambda}^d=(\mathbf{A}^T,\mathbf{C}^T,\mathbf{B}^T)$, consider the subspaces $\mathbfcal{V}^*$ and $\mathbfcal{W}^*$ of Definition \ref{ControlsysinvariancewithoutD}. Then the following holds:
	\begin{itemize}
		\vspace{0.2cm}
		\item [(i)] ${\mathscr W}^*(\Delta^d)=(E{\mathscr V}^*(\Delta))^\bot$ and  ${\mathscr V}^*(\Delta^d)=(H{\mathscr W}^*(\Delta))^\bot$;
			\vspace{0.2cm}
		\item [(ii)] $\mathcal{W}^*(\Lambda^d)=(\mathcal{V}^*(\Lambda))^\bot$ and $\mathcal{V}^*(\Lambda^d)=(\mathcal{W}^*(\Lambda))^\bot$;
			\vspace{0.2cm}
		\item [(iii)] $\mathbfcal{W}^*(\mathbf{\Lambda}^d)=(\mathbfcal{V}^*(\mathbf{\Lambda}))^\bot$ and  $\mathbfcal{V}^*(\mathbf{\Lambda}^d)=(\mathbfcal{W}^*(\mathbf{\Lambda}))^\bot$.
			\vspace{0.2cm}
	\end{itemize}
Moreover, assuming one of the items (i), (ii), or (iii), we can conclude the two remaining ones by the relations given in Proposition \ref{relation invariant subs}.
\end{pro}
Note that item (i) is proved in \cite{Berger2012} by showing that for $i\in \mathbb N$,
\begin{align*}
{\mathscr W}_{i+1}(\Delta^d)=(E{\mathscr V}_{i}(\Delta))^\bot, \ \ \ {\mathscr V}_i(\Delta^d)=(H{\mathscr W}_i(\Delta))^\bot.
\end{align*}
Item (iii) is proved in \cite{morse1973structural} by showing $\mathbfcal{W}_i(\mathbf{\Lambda}^d)=(\mathbfcal{V}_i(\mathbf{\Lambda}))^\bot$,  $\mathbfcal{V}_i(\mathbf{\Lambda}^d)=(\mathbfcal{W}_i(\mathbf{\Lambda}))^\bot$. Item (ii) is proved in \cite{molinari1978structural} by showing $\mathcal{W}_i(\Lambda^d)=(\mathcal{V}_i(\Lambda))^\bot$,  $\mathcal{V}_i(\Lambda^d)=(\mathcal{W}_i(\Lambda))^\bot$ as well as observing a supplementary relation $\mathcal{U}_i(\Lambda^d)=(\mathcal Y_i(\Lambda))^\bot$, $\mathcal Y_i(\Lambda^d)=(\mathcal{U}_i(\Lambda))^\bot$. Our purpose is \red{to propose a new proof in Section \ref{sec:propolem} to show} that knowing one of the items (i), (ii) or (iii), we do not need to prove the two others but just to use the relations of Proposition \ref{relation invariant subs} (between $\mathscr V^*$, $\mathbfcal{V}^*$, $\mathcal{V}^*$ and $\mathscr W^*$, $\mathbfcal{W}^*$, $\mathcal{W}^*$) to simply conclude them. In other words, Proposition \ref{relation invariant subs} provides a dictionary allowing to go from one of (i), (ii), or (iii) to two remaining ones.
\section{Relations between the Kronecker invariants and the Morse invariants}\label{Chap1sec5} 
In this section, we discuss relations of the Kronecker invariants and the Morse invariants (see the Appendix). \red{An} early result   discussing these two \red{sets of} invariants goes back to \cite{kaiman1972kronecker}, where it is observed that the controllability indices of the pair $(A,B)$ and the Kronecker column indices of \red{the} matrix pencil \mage{$sE-H$, where $E=[I,0]$ and $H=[A,B]$,} coincide, which can be seen as a special case of the \red{result} in this section. Also in \cite{loiseau1985some}, it is shown that the Morse indices of \red{the} triple $(A,B,C)$ have direct relations with the Kronecker indices of the matrix pencil (called restricted matrix pencil, see \cite{jaffe1981matrix}) $N(sI-A)K$, where \red{the rows of $N$ span the} annihilator of ${\rm Im\,}B$ and \red{the colunms of $K$ span} $\ker C$.

It is known (see Appendix) that \red{any DAE can be transformed into its} \textbf{KCF} which is completely determined by the Kronecker invariants $\varepsilon_1,...,\varepsilon_a$, $\rho_1,...,\rho_b$, $\sigma_1,...,\sigma_c$, $\eta_1,...,\eta_d$, the numbers $a,b,c,d$ of blocks  and  the $(\lambda_{\rho_1},...,\lambda_{\rho_b})$-structure \mage{(by  the later we mean the eigenvalues, together with the dimensions $\rho_1,...,\rho_b$ of the corresponding blocks)}. The Kronecker invariants (except for $\rho_i$'s and the corresponding eigenvalues $\lambda_{\rho_i}$'s) can be computed using the Wong sequences as follows. For a DAE $\Delta=(E,H)$, consider the Wong sequences $\mathscr V_i$ and $\mathscr W_i$ of Definition \ref{VWrealized}, define $\mathscr{K}_i={\mathscr W}_i\cap {\mathscr V}^*$ and $\hat {\mathscr{K}}_i=(E{\mathscr V}_{i-1})^{ \bot }\cap (H{\mathscr W}^*)^{ \bot }$ for $i\in \mathbb{N}^+$.  
\begin{lem}[\rm \cite{Berger2012},\cite{berger2013addition}]\label{invariants of DAEs}
	For the \textbf{KCF} of $\Delta$, we have
	
	(i) $a=\dim\, (\mathscr{K}_1)$, $d=\dim\, (\hat{\mathscr{K}_1})$ and 
	\begin{align}
	\left\{\begin{array}{*{20}{c}}
	{\begin{array}{*{20}{l}}
		{{\varepsilon_j}}=0,\\
		{{\varepsilon_j}}=i,
		\end{array}}&{\begin{array}{*{20}{l}}
		{for}\\
		{for}
		\end{array}}&{\begin{array}{*{20}{c}}
		{1\le j \le a- \omega_0},\\
		{a - { \omega_{i - 1}} + 1\le j \le a-{ \omega_{i}}},
		\end{array}}
	\end{array} \right. \label{varepsilon}
\\
	\left\{\begin{array}{*{20}{c}}
	{\begin{array}{*{20}{l}}
		{{\eta_j}}=0,\\
		{{\eta_j}}=i,
		\end{array}}&{\begin{array}{*{20}{l}}
		{for}\\
		{for}
		\end{array}}&{\begin{array}{*{20}{c}}
		{1\le j \le d-\hat  \omega_0},\\
		{d- {\hat  \omega_{i - 1}} + 1\le j \le d-{\hat  \omega_{i}}},
		\end{array}}
	\end{array} \right. \label{eta}
	\end{align}
	where ${{ \omega _i} = \dim\, \left( {{\mathscr{K}_{i + 2}}} \right) - \dim\, \left( {{\mathscr{K}_{i+1}}} \right)}$ and ${{{\hat \omega  }_i} = \dim\, ( {{\hat {\mathscr{K}}_{i + 2}}} ) - \dim\, ( {{\hat {\mathscr{K}}_{i+1}}} )}$, $i\in \mathbb{N}$.
	
	(ii)  \red{Define} an integer $\nu $ by
	\begin{align}\label{nu}
	\nu={\rm min}\{i\in \mathbb{N}\,|\,{\mathscr V^*}+{\mathscr W}_i={\mathscr V^*}+{\mathscr W}_{i+1}\};
	\end{align}
	Then either $\nu=0$, implying that the nilpotent part $N(s)$ is absent, or $\nu>0$, in which case $c=\pi_0$ and 
	\begin{align}\label{sigma}		
	\begin{array}{*{20}{c}}
	{\sigma_j=i},&for&{{c - {\pi _{i - 1}} + 1}}\le j\le{c - {\pi _i}},&{ i=1,2,...,\nu},
	\end{array}	
	\end{align}
	where ${\pi _i} = \dim\, ({{{\mathscr W}}_{i + 1}} + {{\mathscr V}^*}) - \dim\, ({{ {\mathscr W}}_i} + { {\mathscr V}^*})$ for $i=0,1,2,...,\nu$ (in \red{the} case of $\pi_{i - 1}=\pi_i$, the respective index range is empty).
\end{lem}
Any control system $\Lambda=(A,B,C,D)$ can be transformed via \red{a} Morse transformation  into its Morse canonical form \textbf{MCF}, which is determined by the Morse indices $\varepsilon'_1,...,\varepsilon'_{a'}$, $\rho'_1,...,\rho'_{b'}$, $\sigma'_1,...,\sigma'_{c'}$, $\eta'_1,...,\eta'_{d'}$, the  $(\lambda_{\rho'_1},..,\lambda_{\rho'_{b'}})$-structure and the numbers $a', b', c', d'\in \mathbb{N}$ of blocks. The following results can be deduced from the results on the Morse indices \red{in} \cite{morse1973structural},\cite{molinari1978structural}. For $\Lambda=(A,B,C,D)$, consider the subspaces $\mathcal{V}_i$, $\mathcal{W}_i$, $\mathcal{U}_i$, $\mathcal Y_i$ as in Lemma~\ref{barVWrealized}, define $\mathcal R_i=\mathcal{W}_i\cap \mathcal{V}^*$ and $\hat{\mathcal R}_i=(\mathcal{V}_i)^{ \bot }\cap (\mathcal{W}^*)^{ \bot }$  for $i\in \mathbb{N}$.
\begin{lem}\label{invariants of contrsys}
	For the \textbf{MCF} of $\Lambda$, we have
	
	(i) $a'=\dim\, (\mathcal{U}^*)$, $d'=\dim\, (\mathcal Y^*)$ and
	\begin{align}\label{alpha}
	\left\{\begin{array}{*{20}{c}}
	{\begin{array}{*{20}{l}}
		{{\varepsilon'_j}}=0\\
		{{\varepsilon'_j}}=i
		\end{array}}&{\begin{array}{*{20}{l}}
		{for}\\
		{for}
		\end{array}}&{\begin{array}{*{20}{c}}
		{1\le j \le a'-\omega' _0},\\
		{a' - {\omega'_{i - 1}} + 1\le j \le a'-{\omega'_{i}}},
		\end{array}}
	\end{array} \right.
	\end{align}
	\begin{align}\label{tau}
	\left\{\begin{array}{*{20}{c}}
	{\begin{array}{*{20}{l}}
		{{\eta'_j}}=0\\
		{{\eta'_j}}=i
		\end{array}}&{\begin{array}{*{20}{l}}
		{for}\\
		{for}
		\end{array}}&{\begin{array}{*{20}{c}}
		{1\le j \le d'-\hat \omega'_0},\\
		{d'- {\hat \omega'_{i - 1}} + 1\le j \le d'-{\hat \omega'}},
		\end{array}}
	\end{array} \right.
	\end{align}
	where ${\omega' _i} = \dim\, \left( \mathcal R_{i + 1} \right) - \dim\, \left( \mathcal R_i \right)$ and ${{\hat \omega' }_i} = \dim\, ( \hat{\mathcal R}_{i+1} ) - \dim\, ( \hat{\mathcal R}_i )$, $i\in \mathbb{N}$.
	
	(ii)  \red{Define} an integer $\nu'$ \red{by}
	\[\nu' =min\{i\in \mathbb{N}\,|\,\mathcal V^*+\mathcal{W}_i=\mathcal V^*+\mathcal{W}_{i+1}\};\]
	Then $c'=\dim\, (\mathscr U)-\dim\, (\mathcal U^*)$, $\delta=c'-\pi'_0$ and
	\begin{align}\label{sigma'}
	\left\lbrace 	\begin{array}{*{20}{l}}
	{\sigma'_j=0}&for&1\le j \le\delta,\\
	{\sigma'_j=i}&for&{{c'- {\pi' _{i - 1}} + 1}}\le j\le{c'- {\pi' _i}},&{ i=1,2,...,\nu'},
	\end{array}	\right. 
	\end{align}
	where ${\pi' _i} =\dim\, ({\mathcal{W}_{i + 1}} + {\mathcal{V}^*}) - \dim\, ({\mathcal{W}_i} + {\mathcal{V}^*})$ for $i=0,1,2,...,\nu'$ (in case of $\pi'_{i - 1}=\pi'_i$ the respective index range is empty).
\end{lem}
Note that for $\Lambda=(A,B,C,D)$, the above index $\delta={\rm rank\,}D$. Formal similarities between the statements of Lemma \ref{invariants of DAEs} and \ref{invariants of contrsys} suggest possible relations between the Kronecker and the Morse invariants. In fact, we have the following result.
\begin{pro}[\mage{invariants} relations]\label{invariants relation}
	For a  DAE $\Delta_{l,n}=(E,H)$, consider its Kronecker invariants  $$(\varepsilon_1,...,\varepsilon_a), \ (\rho_1,...,\rho_b), \ (\sigma_1,...,\sigma_c), \  (\eta_1,...,\eta_d), \  (\lambda_{\rho_1},...,\lambda_{\rho_b}) \  {with} \  a,b,c,d\in \mathbb{N},$$ of the \textbf{KCF}, and for a control system $\Lambda_{q,m,p}=(A,B,C,D)\in {\rm Expl}(\Delta)$, consider its Morse invariants $$(\varepsilon'_1,...,\varepsilon'_{a'}), \  (\rho'_1,...,\rho'_{b'}), \  (\sigma'_1,...,\sigma'_{c'}), \  (\eta'_1,...,\eta'_{d'}), \  (\lambda_{\rho'_1},..,\lambda_{\rho'_{b'}}) \  {with} \ a', b', c', d'\in \mathbb{N}, $$ of the \textbf{MCF}. Then the following holds:
	
	(i)  $a=a'$, $\varepsilon_1=\varepsilon'_1, \cdots,\varepsilon_a=\varepsilon'_{a'},$    and    $d=d'$, $\eta_1=\eta'_1,\dots,\eta_d=\eta'_{d'}$;
	
	(ii) $N(s)$ of the \textbf{KCF} is present if and only if the subsystem $MCF^3$ of the \textbf{MCF} is present. Moreover, if they are present, then their invariants satisfy $$c=c', \ \ \sigma_1=\sigma'_1+1,\dots,\sigma_c=\sigma'_{c'}+1;$$	
	\indent	(iii) The invariant factors of $J(s)$ in the \textbf{KCF} of $\Delta$ coincide with \red{those} of $MCF^2$ in the \textbf{MCF} of $\Lambda$. Furthermore, the corresponding invariants \red{satisfy} $$b=b',\ \ \rho_1=\rho'_1,\dots,\rho_b=\rho'_{b'}, \ \ \lambda_{\rho_1}=\lambda_{\rho'_1},\dots,\lambda_{\rho_b}=\lambda_{\rho'_{b'}}.$$ 
\end{pro}
The proof is given in Section \ref{Pf:invariants relation}. \red{Notice that in item (ii) of Proposition \ref{invariants relation}, the invariants $\sigma_i$ and $\sigma'_i$ do not coincide but differ by one, the reason is that the nilpotent indices $\sigma_1,\dots,\sigma_c$  of $N(s)$ can not be zero (the minimum nilpotent index is $1$ and if $\sigma_i$ is $1$, then $N(s)$ contains the $1\times 1$ matrix pencil $0\cdot s-1$), but the controllability and observability indices $\sigma'_1,\dots,\sigma'_{c'}$ of $MCF^3$ can be zero (if $\sigma'_i=0$, then the output $y^3$ of $MCF^3$ contains the static relation $y^3_i=u^3_i$).} It is easy to see from Proposition~\ref{invariants relation} \red{that}, given a DAE, there exists a \red{perfect} correspondence between the \textbf{KCF} of the DAE and the \textbf{MCF} of its explicitation systems. More specifically, the four parts of the \textbf{KCF} correspond to the four subsystems of the \textbf{MCF}: the bidiagonal pencil $L(s)$ to the controllable but unobservable part $MCF^1$, the Jordan pencil $J(s)$ to the uncontrollable and unobservable part $MCF^2$, the nilpotent pencil $N(s)$ to the prime part $MCF^3$ and the ``pertranspose'' pencil $L^p(s)$ to the observable but uncontrollable part $MCF^4$.
\section{Internal equivalence and regularity of DAEs}\label{Chap1sec6}
An important difference between DAEs and ODEs is that DAEs are not always solvable and solutions of DAEs  exist on a subspace of the ``generalized'' state space \red{only} due to the presence of algebraic constrains. In the following, we show that the existence and uniqueness of  solutions of DAEs can be clearly explained using the \emph{explicitation} procedure and \red{the} notion of internal equivalence (see Definition \ref{in-equivalent} \red{below}).
\begin{defn}\label{Max invariant subspace}
	A linear subspace $\mathscr M$ of $\mathbb{R}^n$,  is called an invariant subspace of $\Delta_{l,n}=(E,H)$ if for any $x^0\in \mathscr M$, there \red{exists} a solution $x(t,x^0)$ of $\Delta$ \red{such that $x(0,x^0)=x^0$} and \red{$x(t,x^0)\in \mathscr M$ for all $t\in \mathbb{R}$}. An invariant subspace $\mathscr M^*$ of $\Delta_{l,n}=(E,H)$ is called the maximal invariant subspace if for any other invariant subspace $\mathscr M$ of $\mathbb{R}^n$, we have ${\mathscr M}\subseteq \mathscr M^*$.
\end{defn}
\begin{rem}
 Note that due to the existence of free variables among the ``generalized'' states, solutions of $\Delta$ are not unique. Thus it is possible that  one solution of $\Delta$ starting at $x^0\in \mathscr M$ stays in $\mathscr M$ but   other solutions starting at $x^0$ may escape from $\mathscr M$ (either immediately or \red{after a} finite time). 
\end{rem}
It is clear that the sum $\mathscr M_1+\mathscr M_2$ of two invariant subspaces of $\Delta$ is also invariant. Therefore, $\mathscr M^*$ exists and is, actually, the sum of all invariant subspaces.  If $\mathscr M$ is an invariant subspace of $\Delta_{l,n}$, then solutions pass through any $x^0\in \mathscr M$ and it is natural to restrict $\Delta$ to $\mathscr M$, in particular, to the largest invariant subspace $\mathscr M^*$. Moreover, we would like the restriction to be as simple as possible. We achieve the above goals by introducing, respectively, the notion of \emph{restriction} and that of \emph{reduction}. We will define the restriction of a DAE $\Delta$ to a linear subspace $\mathscr R$ (invariant or not) as follows.
\begin{defn}[restriction]\label{Def:restri_linearDAE}
	Consider a linear DAE $\Delta_{l,n}=(E,H)$.  Let $\mathscr R$ be a subspace of $\mathbb R^n$.
	The  restriction of $\Delta$ to $\mathscr R$, called $\mathscr R$-restriction  of $\Delta$ and denoted $\Delta|_{\mathscr R}$  is a linear DAE $\Delta|_{\mathscr R}=(E|_{\mathscr R},H|_{\mathscr R})$, where $ E|_{\mathscr R}$ and $ H|_{\mathscr R}$ are, respectively, the restrictions of the linear maps $E$ and $H$ to the linear subspace $\mathscr R$.  
\end{defn}
Throughout, we consider general DAEs $\Delta_{l,n}=(E,H)$ with no assumptions on the ranks of $E$ and $H$. In particular, if \red{the map $[E \ H]$} is not of full row rank, then $\Delta_{l,n}$ contains redundant equations. But even if we assume that $[E \ H]$ is of full row rank, then this property, in general, is not any longer true for the \red{restricted map} $[E|_{\mathscr R}\ H_{\mathscr R}]$, which may contain redundant equations. To get rid of redundant equations (in particular, of trivial algebraic equations $0=0$), we propose the notion of full row rank reduction.
\begin{defn}[reduction]\label{Def:redlinDAE}
	For a DAE $\Delta_{l,n}=(E,H)$ \mage{on $\mathscr X\cong\mathbb R^n$}, assume ${\rm rank\,}[E \ H]=l^*\le l$. Then there exists $Q\in Gl(l,\mathbb R^n)$ such that 
	$$
	Q\left[ \begin{matrix}
	E&H
	\end{matrix}\right] =\left[ \begin{matrix}
	E^{red}&H^{red}\\
	0&0
	\end{matrix}\right],
	$$
	where ${\rm rank\,}[E^{red}\ H^{red}]=l^*$ and the full row rank reduction, shortly reduction, of $\Delta_{l,n}$, denoted by $\Delta^{red}$, is   a DAE $\Delta^{red}_{l^*,n}=\Delta^{red}=(E^{red},H^{red})$ on $\mathscr X\cong\mathbb R^n$.
\end{defn}
\begin{rem}
	Clearly, the choice of $Q$ is not unique \red{and thus} the reduction of $\Delta$ is not unique. Nevertheless, since $Q$ preserves the solutions, each reduction $\Delta^{red}$ has the same solutions as the original DAE $\Delta$.
\end{rem} 
For an invariant subspace $\mathscr M$, we consider the $\mathscr M$-restriction $\Delta|_{\mathscr M}$ of $\Delta$, and then we construct a reduction of $\Delta|_{\mathscr M}$ and denote it by $\Delta|^{red}_{\mathscr M}=(E|^{red}_{\mathscr M},H|^{red}_{\mathscr M})$. Notice that the order matters: to construct $\Delta|^{red}_{\mathscr M}$, we first restrict and then reduce while  reducing first and then restricting will, in general, not give $\Delta|^{red}_{\mathscr M}$ but another DAE $\Delta^{red}|_{\mathscr M}$.
\begin{pro}\label{Pro:inv-subspace}
	Consider a linear DAE $\Delta_{l,n}=(E,H)$.  Let $\mathscr M$ be a subspace of $\mathbb R^n$. The following   are equivalent
	\begin{itemize}
		\item [(i)] $\mathscr M$ is an invariant subspace of $\Delta_{l,n}$;
		
		\item [(ii)] $H\mathscr M\subseteq E \mathscr M$;
		
		\item [(iii)] For a (and thus any) reduction   $\Delta|^{red}_{\mathscr M}=(E|^{red}_{\mathscr M},H|^{red}_{\mathscr M})$ of $\Delta|_{\mathscr M}$, the map $E|^{red}_{\mathscr M}$ is of full row rank, i.e., ${\rm rank\,}E|^{red}_{\mathscr M}={\rm rank\,}[E|^{red}_{\mathscr M}\ H|^{red}_{\mathscr M}]$.
	\end{itemize}	
\end{pro}
\begin{proof}
	(i)$\Leftrightarrow$(ii): Theorem 4 of \cite{berger2016controlled}, for $B=0$, implies that $\mathscr M$ is an invariant subspace if and only if $H\mathscr M\subseteq E\mathscr M$.
	
	(ii)$\Leftrightarrow$(iii): For $\Delta_{l,n}=(E,H)$, choose a full column rank matrix $P_1\in \mathbb R^{n\times n_1}$ such that ${\rm Im\,}P_1=E\mathscr M$, where $n_1=\dim \mathscr M$. Find any $P_2\in \mathbb R^{n\times n_2}$ such that the matrix $[P_1\ P_2]$ is invertible, where $n_2=n-n_1$. Choose new coordinates $z=Px$, where $P=[P_1\ P_2]^{-1}$, then we have
	$$
	\Delta:EP^{-1}P\dot x=HP^{-1}Px\Rightarrow [E_1\ E_2]\left[ \begin{matrix}
	\dot z_1\\ \dot z_2
	\end{matrix}\right]= [H_1\ H_2]\left[ \begin{matrix}
	z_1\\    z_2
	\end{matrix}\right],$$
	where $E_1=EP_1$, $E_2=EP_2$, $H_1=HP_1$, $H_2=HP_2$, and $z=(z_1,z_2)$. Now by Definition \ref{Def:restri_linearDAE}, the $\mathscr M$-restriction of $\Delta$ is:
	$$
	\Delta|_{\mathscr M}:E_1\dot z_1=H_1z_1.
	$$
	Find   $Q\in Gl(l,\mathbb R)$ such that $QE_1=\left[ \begin{matrix}
	\tilde E_1\\    0
	\end{matrix}\right]$, where $\tilde E_1$ is of full row rank, then denote $QH_1=\left[ \begin{matrix}
	\tilde H_1\\   \bar H_1
	\end{matrix}\right]$. By $H\mathscr M\subseteq E\mathscr M$, we can deduce that $\bar H_1=0$ (since $QH\mathscr M\subseteq QE\mathscr M\Rightarrow{\rm Im}\left[ \begin{matrix}
	\tilde H_1\\   \bar H_1
	\end{matrix}\right]\subseteq{\rm Im}\left[ \begin{matrix}
	\tilde E_1\\    0
	\end{matrix}\right]$). Thus a reduction of $\Delta|_{\mathscr M}$, according to Definition \ref{Def:redlinDAE}, is $\Delta|^{red}_{\mathscr M}=(E|^{red}_{\mathscr M},H|^{red}_{\mathscr M})=(\tilde E_1,\tilde H_1)$. Clearly $E|^{red}_{\mathscr M}$ is of full row rank.
\end{proof}
Define $\Lambda|_{(\mathcal V^*,\mathcal U^*)}$ as the control system $\Lambda=(A,B,C,D)$ restricted to $\mathcal V^*$ (which is well-defined because $\mathcal V^*$ can be made invariant by a suitable feedback) and with controls $u$ restricted to $\mathcal U^*=(B^{-1}\mathcal V^*)\cap \ker D$. The output $y=Cx+Du$ of $\Lambda$ becomes $y=0$ and  $\Lambda|^{red}_{(\mathcal V^*,\mathcal U^*)}$ \red{is, by its construction, the system $\Lambda|_{(\mathcal V^*,\mathcal U^*)}$ without the trivial output $y=0$.}
\begin{pro}\label{Pro:M}
	For a DAE $\Delta_{l,n}=(E,H)$, consider its maximal invariant subspace $\mathscr M^*$ and its largest $(E^{-1},H)$-invariant subspace ${\mathscr V^*}$. Then we have
	
	(i) ${\mathscr M^*}$=${\mathscr V^*}$;
	
	(ii) Let $\Lambda\in {\rm Expl}(\Delta)$ and $\Lambda^*\in {\rm Expl}(\Delta|^{red}_{\mathscr M^*})$. Then $\Lambda|^{red}_{(\mathcal V^*,\mathcal U^*)}$ and  $\Lambda^*$ are explicit control systems without outputs \red{i.e., the \textbf{MCF} of the two control systems has no $MCF^3$ and $MCF^4$ parts,} and $\Lambda|^{red}_{(\mathcal V^*,\mathcal U^*)} $ is feedback equivalent to $  \Lambda^*$.
\end{pro}
The proof is given in Section \ref{Pf:Pro:M}.  \red{Using the reduction of $\mathscr M^*$-restriction and the ex-equivalence} of DAEs, we define the internal equivalence of two DAEs as follows.
\begin{defn}\label{in-equivalent}
	For two DAEs  $ \Delta_{l,n}=(E,H)$ and $\tilde \Delta_{\tilde l,\tilde n}=(\tilde E,\tilde H)$, let $\mathscr M^*$ and $\tilde {\mathscr M}^*$ be the maximal invariant subspace of  $ \Delta$ and $\tilde \Delta$, respectively. Then $ \Delta$ and $\tilde \Delta$ are called internally equivalent, shortly in-equivalent,  if $\Delta|^{red}_{\mathscr M^*}$ and $\tilde \Delta|^{red}_{\tilde {\mathscr M}^*}$ are ex-equivalent \red{and we will} denote the in-equivalence of two DAEs as $ \Delta\mathop  \sim \limits^{in} \tilde \Delta$.
\end{defn}
\begin{rem}
	A similar definition to the above internal equivalence above is given in \cite{berger2015regularization},  called the behavioral equivalence, \red{proposed} via \red{the} behavioral approach of DAEs. \red{A} difference between the internal equivalence and  the behavioral equivalence is that, in the definition of internal equivalence, two DAEs are not necessarily of the same dimension, \red{we only require  their reductions of $\mathscr M^*$-restrictions to be of the same dimension (since they are ex-equivalent), but for the behavioral equivalence, the two DAEs are required to have the same dimension.}
\end{rem}
Any $\Lambda^*\in {\rm Expl}(\Delta|^{red}_{\mathscr M^*})$ is an explicit system without outputs (see Proposition \ref{Pro:M}(ii)) and denote the dimensions of its state space and input space by $n^*$ and $m^*$, respectively, and its corresponding matrices by $A^*$, $B^*$ and thus $\Lambda^*_{n^*,m^*}=(A^*,B^*)$.
\begin{thm}\label{in-equi}
	Let $\mathscr M^*$ and $\tilde {\mathscr M}^*$ be the maximal invariant subspaces of  $\Delta$ and $\tilde \Delta$, respectively. Consider two control systems:
	\begin{align*}
	\Lambda^*=(A^*,B^*)\in {\rm Expl}(\Delta|^{red}_{\mathscr M^*}), \ \ \ \tilde \Lambda^*=(\tilde A^*,\tilde B^*)\in {\rm Expl}(\tilde \Delta|^{red}_{\tilde {\mathscr M}^*}).
	\end{align*}
	Then the following is equivalent:
	\begin{itemize}
		\item [(i)] $\Delta\mathop  \sim \limits^{in} \tilde \Delta$;
		\item [(ii)] $\Lambda^*$ and $\tilde {\Lambda}^*$ are feedback equivalent;
		\item [(iii)] $\Delta$ and $\tilde \Delta$ have isomorphic trajectories, i.e, \red{there exists} a linear and invertible map $S:\mathscr M^*\rightarrow\tilde {\mathscr M}^*$ \red{transforming} any trajectory $x(t,x^0)$, where $x^0\in \mathscr M^*$ of $\Delta|^{red}_{\mathscr M^*}$ into \red{a} trajectory $\tilde x(t,\tilde x^0)$, $\tilde x^0\in \tilde {\mathscr{M}}^*$ of $\tilde \Delta|^{red}_{\tilde {\mathscr M^*}}$, where $\tilde x^0=Sx^0$, and vice versa.
	\end{itemize}
\end{thm}
The proof is given in Section \ref{Pf:in-equi}.
In most of \red{the} DAEs literature, regularity of DAEs is frequently studied and various definitions are proposed. From the point of view of the existence and uniqueness of solutions, we propose the following definition of internal regularity of DAEs.
\begin{defn}\label{regularity}
	$\Delta$ is internally regular if through any point $x^0\in \mathscr M^*$, there passes only one solution.
\end{defn}
Recall that	${\rm rank\,}_{\mathbb R[s]} (sE-H)$ denotes the rank of a polynomial matrix $sE-H$ over the ring $\mathbb R[s]$.	
\begin{pro}[internal regularity]\label{Pro:in-regular}
	\red{For a DAE $\Delta_{l,n}=(E,H)$, denote ${\rm rank\,} E=q$. The following statements are equivalent:}
	\begin{itemize}
		\item [(i)] $\Delta$ is internally regular;
		\item [(ii)] \red{Any} $\Lambda^*\in {\rm Expl}(\Delta|^{red}_{\mathscr M^*})$ has no inputs;
		\item [(iii)] The \textbf{MCF} of $\Lambda\in {\rm Expl}(\Delta)$ has no $MCF^1$ part.
		\item [(iv)] \red{${\rm rank\,} E=\dim\, E\mathscr M^*$};
		\item [(v)] \blue{${\rm rank\,}_{\mathbb R[s]}(sE-H)=q$};
		\item [(vi)] \red{The \textbf{MCF} of $\Lambda^*\in {\rm Expl}(\Delta|^{red}_{\mathscr M^*})$ has the $MCF^2$ part only.}
	\end{itemize}
\end{pro}
The proof is given in Section \ref{Pf:Pro:in-regular}.
\begin{rem}
	(i)	The \red{above} definition  of internal regularity is actually equivalent to the definition of \red{an autonomous DAE} in \cite{berger2013differential}. Both of them mean that the DAE is not under-determined (there is no $L(s)$ in the \textbf{KCF} of $sE-H$).  
	
	(ii) Our notion of internal regularity  does not imply that the matrices $E$ and $H$ are square, since the presence of the over-determined part $KCF^4$ (or $L^p(s)$) is allowed \red{for $\Delta=(E,H)$}. 
	
	(iii) If $E$ and $H$ are square ($l=n$), then $\Delta$ (equivalently, $sE-H$) is internally regular if and only if $|sE-H|\not\equiv 0$. It means \red{that for the case of square matrices}, the \red{classical} notion of regularity and internal regularity coincide.
\end{rem}

\section{Proofs of the results}\label{Chap1sec7}
\subsection{Proof of Theorem \ref{main theorem}}\label{Pf:main theorem}
\begin{proof}
	(i) This result can be easily \red{deduced} from Definition \ref{implicitation} and \ref{Def:QPexpl} \mage{and the explicitation procedure}.
	
	(ii) Consider two control systems $${\Lambda=(A,B,C,D)}\in {\rm Expl}(\Delta)\ \ \ {\rm and} \ \ \ \tilde \Lambda =(\tilde A,\tilde B,\tilde C,\tilde D)\in {\rm Expl}(\tilde \Delta).$$ Then by (i) of Theorem \ref{main theorem}, there exist invertible matrices $Q, \tilde Q, P, \tilde P$ of appropriate sizes such that 
	\begin{align}\label{Equation ex-equ}
	Q\left( {sE - H} \right){P^{ - 1}} = \left[ {\begin{matrix}
		{sI - {A}}&{ - {B}}\\
		{ - {C}}&{ - {D}}
		\end{matrix}} \right],\ \ \  \ \ \ \tilde Q\left( {s\tilde E - \tilde H} \right){{\tilde P}^{ - 1}} = \left[ {\begin{matrix}
		{sI - {{\tilde A}}}&{ - {{\tilde B}}}\\
		{ - {{\tilde C}}}&{ - {{\tilde D}}}
		\end{matrix}} \right].
	\end{align} 
	\emph{``If''}. Suppose $\Lambda\mathop  \sim \limits^{M} {{\tilde \Lambda }}$, then there exist Morse transformation matrices $T_s, T_i, T_o, F, K$ such that
	\begin{align}\label{lambatildelambdamorseequ}
	\left[ {\begin{matrix}
		{{T_s}}&{{T_s}K}\\
		0&{{T_o}}
		\end{matrix}} \right]\left[ {\begin{matrix}
		{sI - {A}}&{ - {B}}\\
		{ - {C}}&{ - {D}}
		\end{matrix}} \right]\left[ {\begin{matrix}
		{T_s^{ - 1}}&0\\
		{FT_s^{ - 1}}&{T_i^{ - 1}}
		\end{matrix}} \right] = \left[ {\begin{matrix}
		{sI - {{\tilde A}}}&{ - {{\tilde B}}}\\
		{ - {{\tilde C}}}&{ - {{\tilde D}}}
		\end{matrix}} \right]. 
	\end{align}
	By (\ref{lambatildelambdamorseequ}), we have
	$$
	\begin{array}{c}
	\left[ {\begin{matrix}
		{{T_s}}&{{T_s}K}\\
		0&{{T_o}}
		\end{matrix}} \right]Q\left( {{Q^{ - 1}}\left[ {\begin{matrix}
			{sI - {A}}&{ - {B}}\\
			{ - {C}}&{ - {D}}
			\end{matrix}} \right]P} \right){P^{ - 1}}\left[ {\begin{matrix}
		{T_s^{ - 1}}&0\\
		{FT_s^{ - 1}}&{T_i^{ - 1}}
		\end{matrix}} \right]\\ = \tilde Q\left( {{{\tilde Q}^{ - 1}}\left[ {\begin{matrix}
			{sI - {{\tilde A}}}&{ - {{\tilde B}}}\\
			{ - {{\tilde C}}}&{ - {{\tilde D}}}
			\end{matrix}} \right]\tilde P} \right){{\tilde P}^{ - 1}}.
	\end{array}
	$$
	Substitute (\ref{Equation ex-equ}) into the above equation, to have  
	$${{\tilde Q}^{ - 1}}\left[ {\begin{matrix}
		{{T_s}}&{{T_s}K}\\
		0&{{T_0}}
		\end{matrix}} \right]Q\left( {sE - H} \right){P^{ - 1}}\left[ {\begin{matrix}
		{T_s^{ - 1}}&0\\
		{FT_s^{ - 1}}&{T_i^{ - 1}}
		\end{matrix}} \right]\tilde P =  s\tilde E - \tilde H.
	$$
	\red{Thus} ${\Delta }\mathop  \sim \limits^{ex} {{\tilde \Delta }}$ via $(\bar Q,\bar P)$, where 
	\begin{align*}
	\bar Q = \tilde Q^{-1} \left[ {\begin{matrix}
		{{T_s}}&{{T_s}K}\\
		0&{{T_0}}
		\end{matrix}} \right]Q \ \ {\rm and} \ \ {{\bar P}^{ - 1}} = {P^{ - 1}}\left[ {\begin{matrix}
		{T_s^{ - 1}}&0\\
		{FT_s^{ - 1}}&{T_i^{ - 1}}
		\end{matrix}} \right]\tilde P.
	\end{align*}
	\emph{``Only if''}.  Suppose ${\Delta }\mathop  \sim \limits^{ex} {{\tilde \Delta }}$, then there exist invertible matrices $\bar Q$ and $\bar P$ of appropriate sizes such that $\bar Q\left( {sE - H} \right){\bar P^{ - 1}}=   {s\tilde E - \tilde H}  $, which implies that
	\begin{align*}
	&\bar QQ^{-1}\left(Q\left( {sE - H} \right)P^{-1}\right)P{\bar P^{ - 1}}={\tilde Q}^{-1}\left({\tilde Q}\left( {s\tilde E - \tilde H} \right){\tilde P}^{-1}\right){\tilde P }\\&\mathop  \Rightarrow \limits^{(\ref{Equation ex-equ})}\tilde Q\bar Q{Q^{ - 1}}\left[ {\begin{matrix}
		{sI - {A}}&{ - {B}}\\
		{ - {C}}&{ - {D}}
		\end{matrix}} \right]P{{\bar P}^{ - 1}}{{\tilde P}^{ - 1}} = \left[ {\begin{matrix}
		{sI - {{\tilde A}}}&{ - {{\tilde B}}}\\
		{ - {{\tilde C}}}&{ - {{\tilde D}}}
		\end{matrix}} \right].
	\end{align*}
	Denote $\tilde Q\bar Q{Q^{ - 1}}=\left[ {\begin{matrix}
		{{Q^1}}&{{Q^2}}\\
		{{Q^3}}&{{Q^4}}
		\end{matrix}} \right]$ and $ P{{\bar P}^{ - 1}}{{\tilde P}^{ - 1}}=\left[ {\begin{matrix}
		{P^1}&{P^2}\\
		{P^3}&{P^4}
		\end{matrix}} \right]$, where \red{$Q^i$ and $P^i$, for $i=1,2,3,4$,} are matrices of suitable sizes. Then we get
	\[\left[ {\begin{matrix}
		{Q^1}&{Q^2}\\
		{Q^3}&{Q^4}
		\end{matrix}} \right]\left[ {\begin{matrix}
		{sI - {A}}&{ - {B}}\\
		{ - {C}}&{ - {D}}
		\end{matrix}} \right]\left[ {\begin{matrix}
		{P^1}&{P^2}\\
		{P^3}&{P^4}
		\end{matrix}} \right] = \left[ {\begin{matrix}
		{sI - {{\tilde A}}}&{ - {{\tilde B}}}\\
		{ - {{\tilde C}}}&{ - {{\tilde D}}}
		\end{matrix}} \right].\] 
	Now  by the invertibility of $\tilde Q\bar Q{Q^{ - 1}}$ and $ P{{\bar P}^{ - 1}}{{\tilde P}^{ - 1}}$, we get $\left[ {\begin{matrix}
		{Q^1}&{Q^2}\\
		{Q^3}&{Q^4}
		\end{matrix}} \right]$ and  $\left[ {\begin{matrix}
		{P^1}&{P^2}\\
		{P^3}&{P^4}
		\end{matrix}} \right]$ are invertible. By \red{a direct} calculation, we get $Q^3=0$, $P^2=0$, $Q^1=(P^1)^{-1}$, thus $Q^4$ and $P^4$ are invertible as well.  Therefore, $\Lambda\mathop  \sim \limits^{M}  \tilde {\Lambda}$ via the Morse transformation $${M}_{tran}=\left(Q^1, (P^4)^{-1}, Q^4, P^3Q^1,(Q^1)^{-1}Q^2\right).$$
	
	(iii) Given two control systems $\Lambda=(A,B,C,D)$ and $\tilde \Lambda=(\tilde A,\tilde B,\tilde C,\tilde D)$, the corresponding matrix pencils of \red{$\Delta^{Impl}={\rm Impl}(\Lambda)$} and \red{$\tilde \Delta^{Impl}={\rm Impl}(\tilde \Lambda)$}, \red{by Definition~\ref{implicitation}}, are $\left[ {\begin{matrix}
		{sI - A}&{ - B}\\
		{ - C}&{ - D}
		\end{matrix}} \right]$ and  $\left[ {\begin{matrix}
		{sI - \tilde A}&{ - \tilde B}\\
		{ - \tilde C}&{ - \tilde D}
		\end{matrix}} \right]$, respectively.
	
	\emph{``If''}. Suppose $\Delta^{Impl}\mathop  \sim \limits^{ex}\tilde \Delta^{Impl}$, that is, there exist invertible matrices $Q$ and $P$ such that
	\begin{align}\label{exequofimpl}
	Q\left[ {\begin{matrix}
		{sI - A}&{ - B}\\
		{ - C}&{ - D}
		\end{matrix}} \right]P^{-1}=\left[ {\begin{matrix}
		{sI - \tilde A}&{ - \tilde B}\\
		{ - \tilde C}&{ - \tilde D}
		\end{matrix}} \right].
	\end{align}
	Denote $Q=\left[ {\begin{matrix}
		{{Q_1}}&{{Q_2}}\\
		{{Q_3}}&{{Q_4}}
		\end{matrix}} \right]$ and $P=\left[ {\begin{matrix}
		{{P_1}}&{{P_2}}\\
		{{P_3}}&{{P_4}}
		\end{matrix}} \right]$ with matrices \red{$Q_i$ and $P_i$, for $i=1,2,3,4$, of suitable dimensions}. Then by (\ref{exequofimpl}), we get $Q_3=0$, $P_2=0$, $Q_1=(P_1)^{-1}$. Since $Q$ and $P$ are invertible, we can \red{conclude that} $Q_4$ and $P_4$ are invertible as well. Therefore, $\Lambda\mathop  \sim \limits^{M}  \tilde {\Lambda}$ via the Morse transformation ${M}_{tran}=\left(Q_1, (P_4)^{-1}, Q_4, P_3Q_1,(Q_1)^{-1}Q_2\right)$. 
	
	\emph{``Only if''}. Suppose $\Lambda\mathop  \sim \limits^{M}\tilde \Lambda$ via a Morse transformation $M_{tran}=(T_s, T_i, T_o, F, K)$ (see equation (\ref{morse-trans})), then we have $\Delta^{Impl}\mathop  \sim \limits^{ex}\tilde \Delta^{Impl}$ \red{via $(Q,P)$, where $Q=\left[ {\begin{matrix}
			{T_s}&{T_sK}\\
			0&{{T_o}}
			\end{matrix}} \right]$ and $P^{-1}=\left[ {\begin{matrix}
			{T_s^{ - 1}}&0\\
			{FT_s^{ - 1}}&{T_i^{ - 1}}
			\end{matrix}} \right]$. }
\end{proof}
\subsection{Proof of Proposition \ref{Pro:VW}}\label{Pf:Pro:VW}
\begin{proof}	
	(i)	It can be observed from (\ref{Vrealized}) that $\mathscr V_i$ is non-increasing. By a dimensional argument, \red{the sequence $\mathscr V_i$ gets} stabilized at \red{$i=k^*\le n$} and
	it can be directly seen from $\mathscr V_{k^*}=H^{-1}E\mathscr V_{k^*}$ that $\mathscr V_{k^*}$ is a $(H^{-1},E)$-invariant subspace. We now prove by induction that it is the largest. Choose any other $(H^{-1},E)$-invariant subspace $\hat {\mathscr V}$ and consider (\ref{Vrealized}). For $i=0$, $\hat {\mathscr V}\subseteq \mathscr V_0$; Suppose $\hat {\mathscr V}\subseteq \mathscr V_i$, then $H^{-1}E\hat {\mathscr V}\subseteq H^{-1}E\mathscr V_i$ (since taking the image and preimage preserves inclusion), thus $\hat {\mathscr V}=H^{-1}E\hat {\mathscr V}\subseteq H^{-1}E\mathscr V_i =\mathscr V_{i+1}$. Therefore,  $\hat {\mathscr V}\subseteq \mathscr V_{i}$ for $i\in \mathbb{N}$, i.e., $\hat {\mathscr V}\subseteq {\mathscr V_{{k^*}}}$, it follows \red{that} $\mathscr V_{k^*}$ is the largest $(H^{-1},E)$-invariant subspace.
	
	Now consider (\ref{Wrealized}), observe that \red{the sequence $\mathscr W_i$} is non-decreasing and by a dimensional argument, $\mathscr W_i$ gets stabilized at \red{$i=l^*\le n$}.
	It can be directly seen from $\mathscr W_{l^*}=E^{-1}H\mathscr W_{l^*}$ that $\mathscr W_{l^*}$ is a $(E^{-1},H)$-invariant subspace. We then prove that any other $(E^{-1},H)$-invariant subspace $\hat {\mathscr W}$ contains $\mathscr W^*$, for $i=0$, ${\mathscr W}_0\subseteq \hat {\mathscr W}$; if ${\mathscr W}_i\subseteq \hat {\mathscr W}$, then $E^{-1}H{\mathscr W}_i\subseteq E^{-1}H\hat {\mathscr W}$, so ${\mathscr W}_{i+1}=E^{-1}H{\mathscr W}_i\subseteq E^{-1}H\hat {\mathscr W} =\hat {\mathscr W}$, that is, ${\mathscr W}_i\subseteq \hat {\mathscr W}$ for $i\in \mathbb{N}$, which gives ${\mathscr W}_{l^*}\subseteq \hat {\mathscr W}$ and ${\mathscr W}_{l^*}$ is the smallest $(E^{-1},H)$-invariant subspace.
	
	(ii) By Definition \ref{DAE invariances}, ${\mathscr V}^*$ satisfies $\mathscr V^*=H^{-1} E\mathscr V^*$, \red{thus} it is seen that $H\mathscr V^*\subseteq E\mathscr V^*$. \red{We then prove, by induction that, $\mathscr V^*$ is the largest satisfying that property.} Choose any other subspace $\hat{\mathscr V}$ which satisfies $H\hat{\mathscr V}\subseteq E\hat{\mathscr V}$, consider (\ref{Vrealized}), for $i=0$, \red{so} $\hat {\mathscr V}\subseteq \mathscr V_0$. Suppose $\hat {\mathscr V}\subseteq \mathscr V_i$, then $\hat {\mathscr V}\subseteq H^{-1}E\hat {\mathscr V}\subseteq H^{-1}E\mathscr V_i=\mathscr V_{i+1}$, thus $\hat {\mathscr V}\subseteq H^{-1}E\mathscr V_i =\mathscr V_{i+1}$, therefore  $\hat {\mathscr V}\subseteq \mathscr V_{i}$ for $i\in \mathbb{N}$, i.e., $\hat {\mathscr V}\subseteq {\mathscr V_{{k^*}}}$, which implies  ${\mathscr V^*}={\mathscr V_{{k^*}}}$  is the largest subspace such that $H\mathscr V^*\subseteq E\mathscr V^*$
	
	Obviously, $\{0\}$ is the smallest subspace satisfying $H\{0\}\subseteq E\{0\}$, but ${\mathscr W^*}$ is not always $\{0\}$, so we prove that $\mathscr W^*$ is not necessarily the smallest subspace such that $E\mathscr W^*\subseteq H\mathscr W^*$.	
\end{proof}
\subsection{Proof of Proposition \ref{relation invariant subs}}\label{Pf:relation invariant subs}
\begin{proof}
	\blue{Observe that, by Definition \ref{ex-equivalence} and  \ref{VWrealized}, if two DAEs $\Delta$ and  $\tilde \Delta$ are ex-equivalent via $(Q,P)$, then direct calculations of the Wong sequences of $\Delta$ and $\tilde \Delta$ give that $\mathscr V_i(\tilde \Delta)=P\mathscr V_i(\Delta)$ and $\mathscr W_i(\tilde \Delta)=P\mathscr W_i(\Delta)$. As $\Lambda$ is a $(Q,P)$-explicitation of $\Delta$, by Theorem \ref{main theorem}(i), we have $\Delta\mathop \sim \limits^{ex}\Delta^{Impl}$ via $(Q,P)$, where $\Delta^{Impl}={\rm Impl}(\Lambda)$. Thus we have}
	\begin{align}\label{VWDeltaIm_Delta}
	\mathscr V_i(\Delta^{Impl})=P\mathscr V_i(\Delta), \ \ \ \ \mathscr W_i(\Delta^{Impl})=P\mathscr W_i(\Delta).
	\end{align}
	Notice that 
	$$
	\Delta^{Impl}_{l,n}=\left(\left[ {\begin{matrix}
		I_q&0\\
		0&0
		\end{matrix}} \right], \left[ \begin{matrix}
	A&B\\
	C&D
	\end{matrix} \right]\right), \ \  \mathbf{\Lambda}_{n,m,p}=\left( \mathbf A,\mathbf{B},\mathbf{C}\right) =\left(\left[ \begin{matrix}
	A&B\\
	0&0
	\end{matrix} \right],\left[ {\begin{matrix}
		0\\
		{{I_m}}
		\end{matrix}} \right], \left[ {\begin{matrix}
		{C}&{D}
		\end{matrix}} \right]\right), 
	$$
	\red{where $m=n-q$ and $p=l-q$}. The proof of (i) \red{will be done in} 3 steps :
	
	Step 1: First we show that for $ i\in \mathbb{N}$, 
	\begin{align}\label{rela VtildeV}
	\mathscr V_i(\Delta^{Impl})=\mathbfcal{V}_i(\mathbf{\Lambda}).
	\end{align}
	Calculate ${\mathbfcal{V}_{i + 1}}\left( {{\mathbf{\Lambda}}} \right)$ \red{using} (\ref{tildeVreali}), \red{to get}
	\begin{align}\label{tildeVofpsys}
	{\mathbfcal{V}_{i + 1}}\left( {{\mathbf{\Lambda}}} \right) = \ker \left[ {\begin{matrix}
		{C}&{D}
		\end{matrix}} \right] \cap \left[ {\begin{matrix}
		{A}&{B}\\
		0&0
		\end{matrix}} \right]^{-1}\left(  {{\mathbfcal{V}_i}\left( {{\mathbf{\Lambda}}} \right) + {\rm Im\,}\left[ {\begin{matrix}
			0\\
			{{I_m}}
			\end{matrix}} \right]} \right).&
	\end{align}
	Equation (\ref{tildeVofpsys}) can be written as 
	\begin{align*}
	{\mathbfcal{V}_{i + 1}}\left( {{\mathbf{\Lambda}}} \right) = \left\{ \tilde v\,|\, \left[ {\begin{matrix}
		{A}&{B}
		\end{matrix}} \right]\tilde v \in \left[ {\begin{matrix}
		I_q&0
		\end{matrix}} \right]{\mathbfcal{V}_i}\left( {{\mathbf{\Lambda}}} \right), \ \ \left[ {\begin{matrix}
		{C}&{D}
		\end{matrix}} \right]\tilde v = 0  \right\}
	\end{align*}
	\red{or}, equivalently,  
	\begin{align} \label{tildeVwongform}
	{\mathbfcal{V}_{i + 1}}\left( {{\mathbf{\Lambda}}} \right) = {\left[ {\begin{matrix}
			{A}&{B}\\
			{C}&{D}
			\end{matrix}} \right]^{ - 1}}\left[ {\begin{matrix}
		I_q&0\\
		0&0
		\end{matrix}} \right]{\mathbfcal{V}_i}\left( {{\mathbf{\Lambda}}} \right) .
	\end{align}
	Now, \red{observe that the inductive formula (\ref{tildeVwongform}) for ${\mathbfcal{V}_{i + 1}}(\mathbf{\Lambda})$ coincides with the inductive formula (\ref{Vrealized}) for the Wong sequence $\mathscr V_{i+1}(\Delta^{Impl})$. Since  $\mathscr V_0(\Delta^{Impl})=\mathbfcal{V}_0(\mathbf{\Lambda})=\mathbb{R}^n$, we conclude that $\mathscr V_i(\Delta^{Impl})=\mathbfcal{V}_i(\mathbf{\Lambda})$ for all $ i\in \mathbb{N}$.}  
	
	Step 2:	We then prove \red{that} for $i\in \mathbb{N}$, 
	\begin{align}\label{relation of VbarV}
	\mathscr V_{i+1}(\Delta^{Impl})=\left[ {\begin{matrix}
		A&B\\
		C&D
		\end{matrix}} \right]^{-1}\left[ {\begin{matrix}
		{\mathcal{V}_{i}({\Lambda})}\\
		{0}
		\end{matrix}} \right].
	\end{align}
	By calculating $\mathcal{V}_{i+1}({\Lambda})$ \red{via} (\ref{barV}), we get
	$$
	{{\mathcal{V}_{i + 1}({\Lambda})}} = {\left[ {\begin{matrix}
			A\\
			C
			\end{matrix}} \right]^{ - 1}}\left(  {\left[ {\begin{matrix}
			I\\
			0
			\end{matrix}} \right]{{\mathcal{V}_i({\Lambda})}} + {\rm Im\,}\left[ {\begin{matrix}
			{B}\\
			{D}
			\end{matrix}} \right]} \right) .
	$$
	We can rewrite the above equation as
	\begin{align}\label{barV1}
	\mathcal{V}_{i+1}({\Lambda})=\left[ {\begin{matrix}
		{{I_q}}&0_{q \times m}&0
		\end{matrix}} \right]\ker \left[ {\begin{matrix}
		{A}&{B}&{{{\bar V}_i}}\\
		{C}&{D}&0
		\end{matrix}} \right],
	\end{align}
	where ${\bar V}_i$ is a matrix with independent columns such that ${\rm Im\,}\bar V_i=\mathcal{V}_i({\Lambda})$.
	
	From \red{basic} knowledge of linear algebra, for two matrices $M\in \mathbb R^{l\times n}$ and $N\in \mathbb R^{l\times m}$, the preimage $M^{-1}{\rm Im}N=[I_{n},0]\ker\,[M,N]$. With this formula, calculate ${\mathscr V}_{i+1}(\Delta^{Impl})$ \red{via} (\ref{Vrealized}), \red{to} get
	\begin{align}\label{Vrealized1}
	\mathscr V_{i+1}(\Delta^{Impl})=\left[ {\begin{matrix}
		A&B\\
		C&D
		\end{matrix}} \right]^{-1}\left[ {\begin{matrix}
		I_q&0\\
		0&0
		\end{matrix}} \right]= \left[ {\begin{matrix}
		{{I_q}}&0&0\\
		0&{{I_m}}&0
		\end{matrix}} \right]\ker \left[ {\begin{matrix}
		{\begin{matrix}
			A&B\\
			C&D
			\end{matrix}}&\left[ {\begin{matrix}
			{{I_q}}&0\\
			0&0
			\end{matrix}} \right]V_i
		\end{matrix}} \right],
	\end{align}
	where $V_i$  is a matrix with independent columns such that ${\rm Im\,} V_i=\mathscr V_i(\Delta)$.
	
	In order to show \red{that} (\ref{relation of VbarV}) holds, we will first prove \red{inductively that} for all $i\in \mathbb{N}$, 
	\begin{align}\label{VbarVRelationpre}
	\left[ {\begin{matrix}
		\mathcal{V}_i({\Lambda})\\
		0
		\end{matrix}} \right]=\left[ {\begin{matrix}
		{{I_q}}&0\\
		0&0
		\end{matrix}} \right]\mathscr V_i(\Delta^{Impl}).
	\end{align}
	\red{For} $i=0$, $\left[ {\begin{matrix}
		\mathcal{V}_0({\Lambda})\\
		0
		\end{matrix}} \right]=\left[ {\begin{matrix}
		\mathbb{R}^q\\
		0
		\end{matrix}} \right]=\left[ {\begin{matrix}
		I_q&0\\
		0&0
		\end{matrix}} \right]{\mathscr V}_0(\Delta^{Impl})$. Suppose \red{that} for $i=k\in \mathbb N$, equation (\ref{VbarVRelationpre}) holds or, equivalently, ${\rm Im\,}\left[ {\begin{matrix}
		\bar{{V}}_k\\
		0
		\end{matrix}} \right]={\rm Im\,}\left[ {\begin{matrix}
		{{I_q}}&0\\
		0&0
		\end{matrix}} \right]{{V}}_k$. Then we have
	\begin{align*}
	\left[ {\begin{matrix}
		\mathcal{V}_{k+1}({\Lambda})\\
		0
		\end{matrix}} \right] &\mathop = \limits^{(\ref{barV1})} \left[ {\begin{matrix}
		{{I_q}}&0_{q \times m}&0\\
		0&0_{p \times m}&0
		\end{matrix}} \right]\ker \left[ {\begin{matrix}
		{A}&{B}&{{{\bar V}_k}}\\
		{C}&{D}&0
		\end{matrix}} \right] \\& = \left[ {\begin{matrix}
		{{I_q}}&0\\
		0&0
		\end{matrix}} \right]\left[ {\begin{matrix}
		{{I_q}}&0&0\\
		0&{{I_m}}&0
		\end{matrix}} \right]\ker \left[ {\begin{matrix}
		{\begin{matrix}
			{A}&{B}\\
			{C}&{D}
			\end{matrix}}&\left[ {\begin{matrix}
			{{I_q}}&0\\
			0&0
			\end{matrix}} \right]{V_k}
		\end{matrix}} \right]\\& \mathop = \limits^{(\ref{Vrealized1})}\left[ {\begin{matrix}
		{{I_q}}&0\\
		0&0
		\end{matrix}} \right]\mathscr V_{k+1}(\Delta^{Impl}).
	\end{align*}
	Therefore, equation $(\ref{VbarVRelationpre})$ holds
	for all $i\in \mathbb{N}$.
	
	Consequently, we have for $i\in \mathbb{N}$,
	\begin{align*}
	\mathscr V_{i+1}(\Delta^{Impl})\mathop = \limits^{(\ref{Vrealized})}\left[ {\begin{matrix}
		A&B\\
		C&D
		\end{matrix}} \right]^{-1}\left[ {\begin{matrix}
		{{I_q}}&0\\
		0&0
		\end{matrix}} \right]\mathscr V_i(\Delta^{Impl})\mathop = \limits^{(\ref{VbarVRelationpre})}\left[ {\begin{matrix}
		A&B\\
		C&D
		\end{matrix}} \right]^{-1}\left[ {\begin{matrix}
		\mathcal{V}_i({\Lambda})\\
		0
		\end{matrix}} \right].
	\end{align*}
	
	Step 3:	Finally, since $\mathscr V^*$ and $\mathbfcal{V}^*$ are the limits of the sequences $\mathscr V_i$ and $\mathbfcal{V}_i$, respectively, it follows from (\ref{rela VtildeV}) that $\mathscr V^*(\Delta^{Impl})=\mathbfcal{V}^*(\mathbf{\Lambda})$. Since $\mathscr V^*$ and $\mathcal{V}^*$ are the limits of $\mathscr V_i$ and $\mathcal{V}_i$, respectively,  it follows from (\ref{relation of VbarV}) that $
	\mathscr V^*(\Delta^{Impl})=\left[ {\begin{matrix}
		A&B\\
		C&D
		\end{matrix}} \right]^{-1}\left[ {\begin{matrix}
		\mathcal{V}^*({\Lambda})\\
		0
		\end{matrix}} \right]$. Thus by (\ref{VWDeltaIm_Delta}), we have 
	$
	P\mathscr V^*(\Delta)=\mathscr V^*(\Delta^{Impl})=\mathbfcal{V}^*(\mathbf{\Lambda})= \left[ {\begin{matrix}
		A&B\\
		C&D
		\end{matrix}} \right] ^{-1}\left[ {\begin{matrix}
		{\mathcal{V}^*({\Lambda})}\\
		{0}
		\end{matrix}} \right].
	$
	
	The proof of (ii) \red{will be done in} 3 steps :
	
	Step 1:	Firstly, we show that for $i\in \mathbb{N}$,
	\begin{align}\label{rela WtildeW}
	\mathscr W_i(\Delta^{Impl})=\mathbfcal{W}_i(\mathbf{\Lambda}).
	\end{align}
	Calculate $\mathbfcal W_{i + 1}(\mathbf{\Lambda})$ by (\ref{tildeWreali}), \red{as}
	\begin{align*}
	{\mathbfcal{W}_{i + 1}}(\mathbf{\Lambda}) &= \left[ {\begin{matrix}
		{A}&{B}\\
		0&0
		\end{matrix}} \right]\left(  {{\mathbfcal{W}_i}(\mathbf{\Lambda}) \cap \ker \left[ {\begin{matrix}
			{C}&{D}
			\end{matrix}} \right]}  + {\rm Im\,}\left[ {\begin{matrix}
		0\\
		{{I_m}}
		\end{matrix}} \right]\right) \\&= {\left[ {\begin{matrix}
			I_q&0\\
			0&0
			\end{matrix}} \right]^{ - 1}}\left[ {\begin{matrix}
		{A}&{B}\\
		0&0
		\end{matrix}} \right]\left(  {{\mathbfcal{W}_i}(\mathbf{\Lambda}) \cap \ker \left[ {\begin{matrix}
			{C}&{D}
			\end{matrix}} \right]} \right)  \\
	&=\left( \left[ {\begin{matrix}
		I_q&0\\
		0&0
		\end{matrix}} \right]^{ - 1}\left[ {\begin{matrix}
		{A}&{B}\\
		C&D
		\end{matrix}} \right]{\mathbfcal{W}_i}(\mathbf{\Lambda}) \right) \cap \left( \left[ {\begin{matrix}
		I_q&0\\
		0&0
		\end{matrix}} \right]^{ - 1}\left[ {\begin{matrix}
		{A}&{B}\\
		C&D
		\end{matrix}} \right] \ker \left[ {\begin{matrix}
		{C}&{D}
		\end{matrix}} \right]\right) . 
	\end{align*}
	Observe  that
	\begin{align*}
	{\left[ {\begin{matrix}
			I_q&0\\
			0&0
			\end{matrix}} \right]^{ - 1}}\left[ {\begin{matrix}
		{A}&{B}\\
		{C}&{D}
		\end{matrix}} \right]\ker \left[ {\begin{matrix}
		{C}&{D}
		\end{matrix}} \right]& =\left[ {\begin{matrix}
		{\left[ {\begin{matrix}
				A&B
				\end{matrix}} \right]\ker \left[ {\begin{matrix}
				C&D
				\end{matrix}} \right]}\\
		*
		\end{matrix}} \right] + {\rm Im\,}\left[ {\begin{matrix}
		0\\
		{{I_m}}
		\end{matrix}} \right] \\&= {\left[ {\begin{matrix}
			I_q&0\\
			0&0
			\end{matrix}} \right]^{ - 1}}{\mathop{\rm Im\,}\nolimits} \left[ {\begin{matrix}
		A&B\\
		C&D
		\end{matrix}} \right]. 
	\end{align*}
	Then we have
	\begin{align}\label{tildeWwongform}
	{\mathbfcal{W}_{i + 1}}(\mathbf{\Lambda}) =\left[ {\begin{matrix}
		I_q&0\\
		0&0
		\end{matrix}} \right]^{ - 1}\left[ {\begin{matrix}
		{A}&{B}\\
		C&D
		\end{matrix}} \right] {\mathbfcal{W}_i}(\mathbf{\Lambda}).
	\end{align}
	\red{Observe that the inductive formula (\ref{tildeWwongform}) for ${\mathbfcal{W}_{i + 1}}(\mathbf{\Lambda})$ coincides with the inductive formula (\ref{Wrealized}) for the Wong sequence $\mathscr W_{i+1}(\Delta^{Impl})$. Since $\mathscr W_0(\Delta^{Impl})=\mathbfcal{W}_0(\mathbf{\Lambda})=\{0\}$, we deduce that $\mathscr W_i(\Delta^{Impl})=\mathbfcal{W}_i(\mathbf{\Lambda})$ for $i\in \mathbb{N}$.}
	
	Step 2:	Subsequently, we \red{will} prove \red{that} for $i\in \mathbb{N}$, 
	\begin{align}\label{rela tildeWbarW}
	\mathbfcal{W}_{i+1}(\mathbf{\Lambda})=\left[ \begin{matrix}
	I_q&0\\
	0&0
	\end{matrix} \right]^{ - 1}\left[ {\begin{matrix}
		{\mathcal{W}_{i}({\Lambda})}\\
		{0}
		\end{matrix}} \right].
	\end{align}
	Considering (\ref{barW}) for $\Lambda$, we have
	\begin{align*}
	\left[ {\begin{matrix}
		{{\mathcal{W}_{i + 1}}}(\Lambda)\\
		0
		\end{matrix}} \right]& = \left[ {\begin{matrix}
		{A}&{B}\\
		0&0
		\end{matrix}} \right]\left(  {\left[ {\begin{matrix}
			{{\mathcal{W}_i}}(\Lambda)\\
			{{\mathbb{R}^m}}
			\end{matrix}} \right] \cap \ker \left[ {\begin{matrix}
			C&D
			\end{matrix}} \right]} \right) \\ &= \left[ {\begin{matrix}
		{A}&{B}\\
		0&0
		\end{matrix}} \right]\left(  {\left(\left[ \begin{matrix}
		I_q&0\\
		0&0
		\end{matrix} \right]^{ - 1}\left[ \begin{matrix}
		\mathcal{W}_i(\Lambda)\\
		0
		\end{matrix} \right] \right) \cap \ker \left[ {\begin{matrix}
			C&D
			\end{matrix}} \right]} \right) ,
	\end{align*}
	which implies that
	\begin{align}\label{tildeWbarWrelation}
	\left[ \begin{matrix}
	I_q&0\\
	0&0
	\end{matrix} \right]^{ - 1}\left[ {\begin{matrix}
		{{\mathcal{W}_{i + 1}}}(\Lambda)\\
		0
		\end{matrix}} \right] = \left[ {\begin{matrix}
		{A}&{B}\\
		0&0
		\end{matrix}} \right]\left(  \left[ \begin{matrix}
	I_q&0\\
	0&0
	\end{matrix} \right]^{ - 1}\left[ {\begin{matrix}
		{{\mathcal{W}_i}}(\Lambda)\\
		0
		\end{matrix}} \right] \cap \ker \left[ {\begin{matrix}
		C&D
		\end{matrix}} \right] \right)  + {\rm Im\,}\left[ {\begin{matrix}
		0\\
		{{I_m}}
		\end{matrix}} \right].
	\end{align}
	\red{Observe that the inductive formula (\ref{tildeWbarWrelation}) for $\left[ \begin{matrix}
		I_q&0\\
		0&0
		\end{matrix} \right]^{ - 1}\left[ {\begin{matrix}
			{{\mathcal{W}_{i+1}}}(\Lambda)\\
			0
			\end{matrix}} \right]$ coincides with the inductive formula (\ref{tildeWreali}) for $\mathbfcal{W}_{i+1}(\mathbf{\Lambda})$.}   Since ${\mathbfcal{W}_1}(\Lambda) = \left[ \begin{matrix}
	I_q&0\\
	0&0
	\end{matrix} \right]^{ - 1}\left[ {\begin{matrix}
		{{\mathcal{W}_0}}(\Lambda)\\
		0
		\end{matrix}} \right] = {\rm Im\,}\left[ {\begin{matrix}
		0\\
		{{I_m}}
		\end{matrix}} \right]$, we have $\mathbfcal{W}_{i+1}(\mathbf{\Lambda})=\left[ \begin{matrix}
	I_q&0\\
	0&0
	\end{matrix} \right]^{-1}\left[ {\begin{matrix}
		{\mathcal{W}_{i}({\Lambda})}\\
		{0}
		\end{matrix}} \right]$ for all $i\in \mathbb N$.
	
	Step 3: Equation (\ref{rela WtildeW}) and \red{the fact that} $\mathscr W^*$ and $\mathbfcal{W}^*$ are the limits of $\mathscr W_i$ and $\mathbfcal{W}_i$, respectively, yield $\mathscr W^*(\Delta)=\mathbfcal{W}^*(\mathbf{\Lambda})$. Equation (\ref{rela tildeWbarW}) and \red{the fact that} $\mathcal{W}^*$ and $\mathbfcal{W}^*$ are the limits of $\mathcal{W}_i$ and $\mathbfcal{W}_i$, respectively, yield $\mathbfcal{W}^*(\mathbf{\Lambda})=\left[ \begin{matrix}
	I_q&0\\
	0&0
	\end{matrix} \right]^{ - 1}\left[ \begin{matrix}
	\mathcal{W}^*({\Lambda})\\
	0
	\end{matrix} \right]$. Thus \red{using} equation (\ref{VWDeltaIm_Delta}), we  prove (ii) of Proposition \ref{relation invariant subs}.
\end{proof}
\subsection{Proof of Proposition \ref{Subspaces of the dual system}}\label{sec:propolem}
\red{In this proof, we will need the following two lemmata.  Denote by $\mathbb{F}({{\mathcal{V}_i}(\Lambda)})$ the class of maps $F:\mathbb{R}^q\rightarrow\mathbb{R}^m$ satisfying $(A+BF){\mathcal{V}_{i+1}(\Lambda)}\subset{{\mathcal{V}_i}(\Lambda)}$ and $(C+DF)\mathcal{V}_{i+1}(\Lambda)=0$.}
\begin{lem}\label{subs relationex}
	\red{	Given $\Delta_{l,n}=(E,H)$, its $(Q,P)$-explicitation $\Lambda=(A,B,C,D)\in {\rm Expl}(\Delta)$, and $\Delta^{Impl}={\rm Impl}(\Lambda)$,  consider the Wong sequences $\mathscr V_i$, $\mathscr W_i$ of both $\Delta$ and $\Delta^{Impl}$, given by Definition \ref{VWrealized} and the subspaces $\mathcal{V}_i$, $\mathcal{W}_i$ of $\Lambda$, given by Lemma \ref{tildeVWreali}.  Then for $i\in \mathbb N$, we have}
	\begin{align}\label{RelationV2}
	{\mathscr V}_{i+1}({\Delta^{Impl}})=P{\mathscr V}_{i+1}({\Delta}) = \left[ {\begin{matrix}
		{{\mathcal{V}_{i+1}}(\Lambda)}\\
		{F_i  {\mathcal{V}_{i+1}}(\Lambda)}
		\end{matrix}} \right] +  \left[ {\begin{matrix}
		0\\
		{{\mathcal{U}_i}(\Lambda)}
		\end{matrix}} \right],
	\end{align}
	where $F_i\in \mathbb{F}({{\mathcal{V}_i}(\Lambda)})$ \red{and} 
	\begin{align}\label{RelationW2}
	{\mathscr W}_{i+1}({\Delta^{Impl}})=P{\mathscr W}_{i + 1}({\Delta}) = \left[ {\begin{matrix}
		{{\mathcal{W}_i}(\Lambda)}\\
		* 
		\end{matrix}} \right] + \left[ {\begin{matrix}
		0\\
		{\mathscr U}(\Lambda)
		\end{matrix}} \right].
	\end{align}
\end{lem}
\begin{lem}\label{tildeVWdualwongform}
	Consider the subspace sequences ${\mathbfcal{V}_{i}}$ and ${\mathbfcal{W}_{i}}$ of $\mathbf{\Lambda}^d$, given by Lemma \ref{tildeVWreali}. Then for $i\in \mathbb{N}$, the following holds
	\begin{align}
	P^{T}\mathbfcal{W}_{i+1}( \mathbf{\Lambda}^d)=H^T(E^T)^{-1}\left( P^{T}\mathbfcal{W}_{i}( \mathbf{\Lambda}^d )\right) 	\label{tildeWdualwongform},\\
	P^{T}\mathbfcal{V}_{i+1}\left( {\mathbf{\Lambda}^d} \right)=E^T(H^T)^{-1}\left( P^{T}\mathbfcal{V}_{i}( {\mathbf{\Lambda}^d} )\right) \label{tildeVdualwongform}.
	\end{align}
\end{lem}
\begin{proof}[Proof of Lemma \ref{subs relationex}]
	We first show \red{that} equation (\ref{RelationV2}) holds. Let independent vectors $v_1=\left[ {\begin{matrix}
		{{v^1_1}}\\
		{{v^2_1}}
		\end{matrix}} \right],...,v_\alpha=\left[ {\begin{matrix}
		{{v^1_\alpha}}\\
		{{v^2_\alpha}}
		\end{matrix}} \right]\in \mathbb{R}^n$ form a basis of 
	$$
	P\mathscr V_{i+1}(\Delta)\mathop = \limits^{(\ref{VWDeltaIm_Delta})}\mathscr V_{i+1}(\Delta^{Impl})\mathop = \limits^{(\ref{relation of VbarV})}{\left[ {\begin{matrix}
			{A}&{B}\\
			{C}&{D}
			\end{matrix}} \right]^{ - 1}} \left[ {\begin{matrix}
		{{\mathcal{V}_{i}}(\Lambda)}\\
		0
		\end{matrix}} \right], 
	$$ 
	where $v^1_j\in \mathbb{R}^q, v^2_j\in \mathbb{R}^m, j=1,2,..., \alpha$ (\red{implying that} $\dim\,(\mathscr V_{i+1}(\Delta^{Impl}))=\alpha$). Now without loss of generality, assume  $v^1_j \ne 0$ for $j=1,...,\kappa$ and  $v^1_j=0$ for $j=\kappa+1,...,\alpha$, \red{where $\kappa< \alpha$ is the number of non-zero vectors $v^1_j$}.  Then from equation (\ref{VbarVRelationpre}), it can be deduced that $v^1_j$ for $j=1,...,\kappa$ form a basis of ${\mathcal{V}_{i+1}({\Lambda})}$. Moreover, from (\ref{relation of VbarV}), it is not hard to see that  $v^2_j$ for $j=\kappa+1,...,\alpha$ form a basis of $\mathcal{U}_i(\Lambda)$.  Let $F_i\in \mathbb R^{m\times \kappa}$ \red{be} such that $F_iv_j^1=v_j^2$ for $j=1,...,\kappa$ (such $F_i$ exists), then $v_1,\dots,v_\alpha$ form a basis of 
	$
	\left[ {\begin{matrix}
		{{\mathcal{V}_{i+1}}(\Lambda)}\\
		{F_i  {\mathcal{V}_{i+1}}(\Lambda)}
		\end{matrix}} \right]+\left[ {\begin{matrix}
		0\\
		{{\mathcal{U}_i}(\Lambda)}
		\end{matrix}} \right]
	$. \red{Therefore,}
	$$
	\left[ {\begin{matrix}
		{{\mathcal{V}_{i+1}}(\Lambda)}\\
		{F_i {\mathcal{V}_{i+1}}(\Lambda)}
		\end{matrix}} \right] +  \left[ {\begin{matrix}
		0\\
		{{\mathcal{U}_{i}}(\Lambda)}
		\end{matrix}} \right]= {\left[ {\begin{matrix}
			{A}&{B}\\
			{C}&{D}
			\end{matrix}} \right]^{ - 1}} \left[ {\begin{matrix}
		{{\mathcal{V}_{i}}(\Lambda)}\\
		0
		\end{matrix}} \right],
	$$
	\red{because both spaces have the same basis $v_1,\dots,v_{\alpha}$.}
	We now prove \red{that for any choice} of $F_i$, \red{we have} $F_i\in \mathbb{F}({{\mathcal{V}_i}(\Lambda)})$. Pre-multiply the above equation by $\left[ {\begin{matrix}
		{A}&{B}\\
		{C}&{D}
		\end{matrix}} \right]$ \red{on the left} \red{to obtain}
	\begin{align*}
	\left[ {\begin{matrix}
		(A+BF_i)	{{\mathcal{V}_{i+1}}(\Lambda)}\\
		{(C+DF_i) {\mathcal{V}_{i+1}}(\Lambda)}
		\end{matrix}} \right] +  \left[ {\begin{matrix}
		B{{\mathcal{U}_{i}}(\Lambda)}\\
		D{{\mathcal{U}_{i}}(\Lambda)}
		\end{matrix}} \right]\subseteq \left[ {\begin{matrix}
		{{\mathcal{V}_{i}}(\Lambda)}\\
		0
		\end{matrix}} \right]. 
	\end{align*} 
	Moreover, we get $
	\left[ {\begin{matrix}
		B{{\mathcal{U}_i}(\Lambda)}\\
		D{{\mathcal{U}_i}(\Lambda)}
		\end{matrix}} \right]\subseteq \left[ {\begin{matrix}
		{{\mathcal{V}_i}(\Lambda)}\\
		0
		\end{matrix}} \right]$ by (\ref{barU}). Thus it is easy to see \red{that} $(A+BF_i)	{{\mathcal{V}_{i+1}}(\Lambda)}\subseteq {\mathcal{V}_i}$ and $(C+DF_i)	{{\mathcal{V}_{i+1}}(\Lambda)}=0$. 
	
	Subsequently, we show \red{that} equation (\ref{RelationW2}) holds. By (\ref{rela WtildeW}) and (\ref{rela tildeWbarW}), it follows that for $i\in \mathbb{N}$,
	\begin{align}\label{WbarW relationpre}
	{{\mathscr W}_{i + 1}}({\Delta^{Impl}}) = \left[ {\begin{matrix}
		I_q&0\\
		0&0
		\end{matrix}} \right]^{-1}\left[ {\begin{matrix}
		{{\mathcal{W}_i}(\Lambda)}\\
		0
		\end{matrix}} \right].
	\end{align}
	Then by (\ref{VWDeltaIm_Delta}), we have $\mathscr W_{i+1}(\Delta^{Impl})=P\mathscr W_{i+1}(\Delta)$ \red{and we complete the proof of (\ref{RelationW2}) by calculating explicitly the right-hand side of (\ref{WbarW relationpre}).}
\end{proof}
\begin{proof}[Proof of Lemma \ref{tildeVWdualwongform}]
	Notice that
	$
	\mathbf{\Lambda}^d_{n,p,m} = \left( {\left[ {\begin{matrix}
			A^T&0\\
			B^T&0
			\end{matrix}} \right],\left[ {\begin{matrix}
			C^T\\
			D^T
			\end{matrix}} \right],\left[ {\begin{matrix}
			0&{{I_m}}
			\end{matrix}} \right]} \right).
	$
	We first prove \red{that} the following \red{relations} hold 
	\begin{align}\label{tildeWdualwongformpre}
	\begin{array}{c}
	\mathbfcal{W}_{i+1}( \mathbf{\Lambda}^d)=\left[ {\begin{matrix}
	A&B\\
	C&D
	\end{matrix}} \right]^T\left[ {\begin{matrix}
	I_q&0\\
	0&0
	\end{matrix}} \right]^{-T}\mathbfcal{W}_{i}( \mathbf{\Lambda}^d ), \\
\mathbfcal{V}_{i+1}( \mathbf{\Lambda}^d)=\left[ {\begin{matrix}
	I_q&0\\
	0&0
	\end{matrix}} \right]^T\left[ {\begin{matrix}
	A&B\\
	C&D
	\end{matrix}} \right]^{-T}\mathbfcal{V}_{i}( \mathbf{\Lambda}^d ).
	\end{array}
	\end{align}
	For ${\mathbf{\Lambda}^d}$, calculate $\mathbfcal{W}_{i+1}$ \red{via} (\ref{tildeWreali}), to get for $i\in \mathbb N$:
	\begin{align*}
	\mathbfcal{W}_{i+1}\left( {\mathbf{\Lambda}^d} \right) = \left[ {\begin{matrix}
		A^T&0\\
		B^T&0
		\end{matrix}} \right]\left(  {{\mathbfcal{W}_i}\left( {\mathbf{\Lambda}^d} \right) \cap \ker \left[ {\begin{matrix}
			0&{{I_m}}
			\end{matrix}} \right]} \right)  + {\mathop{\rm Im\,}\nolimits} \left[ {\begin{matrix}
		C^T\\
		D^T
		\end{matrix}} \right].
	\end{align*}
	Moreover, it is not hard to see that
	\begin{align*}
	\left[ {\begin{matrix}
		I_q&0\\
		0&0
		\end{matrix}} \right]^{-T}{\mathbfcal{W}_{i}}\left( {\mathbf{\Lambda}^d} \right)&=\left[ {\begin{matrix}
		I_q&0\\
		0&0
		\end{matrix}} \right]\left( \mathbfcal{W}_i\left( {\mathbf{\Lambda}^d} \right)\cap \ker \left[ {\begin{matrix}
		0&I_m
		\end{matrix}} \right]\right) +{\rm Im\,}\left[ \begin{matrix}
	0\\
	I_p
	\end{matrix} \right]. 
	\end{align*}
	Pre-multiply both sides of the above equation by $\left[ {\begin{matrix}
		A&B\\
		C&D
		\end{matrix}} \right]^{T}$, it follows that
	\begin{align*}   
	\left[ {\begin{matrix}
		A&B\\
		C&D
		\end{matrix}} \right]^T\left[ {\begin{matrix}
		I_q&0\\
		0&0
		\end{matrix}} \right]^{-T}\mathbfcal{W}_{i}( \mathbf{\Lambda}^d )
	 &= \left[ {\begin{matrix}
		A^T&0\\
		B^T&0
		\end{matrix}} \right]\left(  {{\mathbfcal{W}_i}\left( {\mathbf{\Lambda}^d} \right) \cap \ker \left[ {\begin{matrix}
			0&{{I_m}}
			\end{matrix}} \right]} \right)  + {\mathop{\rm Im\,}\nolimits} \left[ {\begin{matrix}
		C^T\\
		D^T
		\end{matrix}} \right]\\& ={\mathbfcal{W}_{i+1}}\left( {\mathbf{\Lambda}^d} \right).
	\end{align*}
	Then calculate $\mathbfcal{V}_{i+1}$ for $\mathbf{\Lambda}^d$, \red{via} (\ref{tildeVreali}),  to get for $i\in \mathbb{N}$,
	\begin{align}\label{tildeVdualpro}
	{\mathbfcal{V}_{i+1}}\left( {\mathbf{\Lambda}^d} \right) = \ker \left[ {\begin{matrix}
		0&{{I_m}}
		\end{matrix}} \right] \cap \left[ {\begin{matrix}
		A^T&0\\
		B^T&0
		\end{matrix}} \right]^{-1}\left(  \mathbfcal{V}_i\left( {\mathbf{\Lambda}^d} \right)   + {\mathop{\rm Im\,}\nolimits} \left[ {\begin{matrix}
		C^T\\
		D^T
		\end{matrix}} \right] \right) .
	\end{align}
	Rewrite  (\ref{tildeVdualpro}) as
	\begin{align*}
	\mathbfcal{V}_{i+1}\left( {\mathbf{\Lambda}^d} \right)&= {\rm Im\,} \left[ {\begin{matrix}
		{{I_q}}\\
		0
		\end{matrix}} \right] \cap \left( \left[ {\begin{matrix}
		I_n&0
		\end{matrix}} \right]\ker \left[ {\begin{matrix}
		{ {\begin{matrix}
				{{{\left( {A} \right)}^T}}&0\\
				{{{\left( {B} \right)}^T}}&0
				\end{matrix}} }&{{{\tilde V}_i}\left( {\mathbf{\Lambda}^d} \right)}&{ {\begin{matrix}
				{{{\left( {C} \right)}^T}}\\
				{{{\left( {D} \right)}^T}}
				\end{matrix}} }
		\end{matrix}} \right]\right)\\& =\left[ \begin{matrix}
	\left[ {\begin{matrix}
		I_q&0
		\end{matrix}} \right]\ker \left[ \begin{matrix}
	\begin{matrix}
	A^T&C^T\\
	B^T&D^T
	\end{matrix}&{{{\tilde V}_i}\left( {\mathbf{\Lambda}^d} \right)}
	\end{matrix} \right]\\
	0
	\end{matrix} \right]=\left[ \begin{matrix}
	I_q&0\\
	0&0
	\end{matrix} \right]\left[ \begin{matrix}
	A^T&C^T\\
	B^T&D^T
	\end{matrix}\right]^{-1}{\mathbfcal{V}_{i}}\left( {\mathbf{\Lambda}^d} \right).
	\end{align*}
	Therefore, the proof of (\ref{tildeWdualwongformpre}) is complete. Consequently, substitute
	$$
	\left[ {\begin{matrix}
		A&B\\
		C&D
		\end{matrix}} \right]=QHP^{-1}, \ \ \ \   \left[ {\begin{matrix}
		I_q&0\\
		0&0
		\end{matrix}} \right]=~QEP^{-1}
	$$
	into (\ref{tildeWdualwongformpre}), \red{then} it is straightforward to see \red{that} (\ref{tildeWdualwongform}) and (\ref{tildeVdualwongform}) hold for \red{any} $i\in \mathbb N$. 
\end{proof}
\begin{proof}[Proof of Proposition \ref{Subspaces of the dual system}]
	Notice that since $\Lambda\in {\rm Expl}(\Delta)$, by Proposition \ref{explifdual}, we have $\Lambda^d\in {\rm Expl}(\Delta^d)$. Moreover, it is easy to see \blue{if $\Lambda$ is the $(Q,P)$-explicitation of $\Delta$, then $\Lambda^d$ is the $(P^{-T},Q^{-T})$-explicitation of $\Delta^d$.} The proof \red{will be done in} 3 steps.
	
	Step 1; Step 1a: We show that for $i\in \mathbb{N}$,
	\begin{align}\label{relWdual}
	{\mathscr W}_{i+1}(\Delta^d)=(E{\mathscr V}_i(\Delta))^\bot\Leftrightarrow \mathcal{W}_i(\Lambda^d)=(\mathcal{V}_i(\Lambda))^\bot. 
	\end{align}
	By $\Lambda^d\in {\rm Expl}(\Delta^d)$ and (\ref{RelationW2}) of Lemma \ref{subs relationex}, we \red{get}
	\begin{align*}
	Q^{-T}{\mathscr W}_{i+1}(\Delta^d)=\left[ {\begin{matrix}
		\mathcal{W}_i(\Lambda^d) \\
		* 
		\end{matrix}} \right] +{\rm Im\,}\left[ {\begin{matrix}
		0\\
		{I_p}
		\end{matrix}} \right].
	\end{align*}
	Moreover, we have 
	\begin{align*}
	(E{\mathscr V}_i(\Delta))^\bot&=(Q^{-1}QEP^{-1}P{\mathscr V}_i(\Delta))^\bot\mathop = \limits^{(\ref{VWDeltaIm_Delta})}(Q^{-1}\left[ {\begin{matrix}
		I_q&0\\
		0&0
		\end{matrix}} \right]{\mathscr V}_i(\Delta^{Impl}))^\bot\\&\mathop = \limits^{(\ref{VbarVRelationpre})}Q^{T}\left[ {\begin{matrix}
		{\mathcal{V}_{i}({\Lambda})}\\
		{0}
		\end{matrix}} \right]^\bot=Q^T\left( \left[ {\begin{matrix}
		(\mathcal{V}_i(\Lambda))^\bot \\
		* 
		\end{matrix}} \right] + {\rm Im\,}\left[ {\begin{matrix}
		0\\
		{I_p}
		\end{matrix}} \right]\right) .
	\end{align*}
	It is seen that ${\mathscr W}_{i+1}(\Delta^d)=(E{\mathscr V}_i(\Delta))^\bot$ if and only if $\mathcal{W}_i(\Lambda^d)=(\mathcal{V}_i(\Lambda))^\bot$.
	
	Step 1b: \red{In this step, we will prove} that for $i\in \mathbb{N}$,
	\begin{align}\label{relVdual}
	{\mathscr V}_{i}(\Delta^d)=(H{\mathscr W}_i(\Delta))^\bot\Leftrightarrow \mathcal{V}_i(\Lambda^d)=(\mathcal{W}_i(\Lambda))^\bot. 
	\end{align}
	
	\red{We first prove ``$\Rightarrow$'' of (\ref{relVdual}):} \red{Considering} equation (\ref{VWDeltaIm_Delta}) and  (\ref{VbarVRelationpre}) for $\Delta^d$, we can deduce that 
	\begin{align*}
	E^T\mathscr V_i(\Delta^d)=P^T\left[ \begin{matrix}
	I_q&0\\
	0&0
	\end{matrix}\right]^TQ^{-T}\mathscr V_i(\Delta^d)=P^T\left[ {\begin{matrix}
		\mathcal{V}_i(\Lambda^d)\\
		0
		\end{matrix}} \right].
	\end{align*}
	On the other hand, we have
	\begin{align*}
	E^T(H{\mathscr W}_{i}(\Delta))^\bot&=(E^{-1}H{\mathscr W}_{i}(\Delta))^\bot\mathop = \limits^{(\ref{Wrealized})}({\mathscr W}_{i+1}(\Delta))^\bot=(P^{-1}P{\mathscr W}_{i+1}(\Delta))^\bot\\&=(P^{-1})^{-T}(P{\mathscr W}_{i+1}(\Delta))^{\bot}\mathop = \limits^{(\ref{RelationW2})}P^T\left(\left[ {\begin{matrix}
		\mathcal{W}_i(\Lambda)\\
		* 
		\end{matrix}}\right] + \left[ \begin{matrix}
	0\\
	{\mathscr U}(\Lambda)
	\end{matrix} \right] \right)^\bot\\&=P^T\left[ {\begin{matrix}
		{{\mathcal{W}_i}(\Lambda)}^\bot\\
		0
		\end{matrix}} \right].
	\end{align*} 
	Now we can see \red{that} for $i\in \mathbb{N}$, if ${\mathscr V}_{i}(\Delta^d)=(H{\mathscr W}_i(\Delta))^\bot$, then $   \mathcal{V}_i(\Lambda^d)=(\mathcal{W}_i(\Lambda))^\bot$.	
	
	We then prove ``$\Leftarrow$'' of (\ref{relVdual}): By equation (\ref{VWDeltaIm_Delta}) and (\ref{relation of VbarV}), we can deduce that 
	\begin{align}\label{VtildeVreldual}
	Q^{-T}{\mathscr V}_{i+1}(\Delta^d)=\left[ {\begin{matrix}
		A^T&C^T\\
		B^T&D^T
		\end{matrix}} \right]^{-1}\left[ {\begin{matrix}
		{{\mathcal{V}_i}(\Lambda ^d)}\\
		0
		\end{matrix}} \right].
	\end{align} 
	We have 
	$$
	\begin{array}{ll}
	(H{\mathscr W}_{i+1}(\Delta))^\bot&=(Q^{-1}QHP^{-1}P{\mathscr W}_{i+1}(\Delta))^\bot=(Q^{-1}\left[ {\begin{matrix}
		A&B\\
		C&D
		\end{matrix}} \right]P{\mathscr W}_{i+1}(\Delta))^\bot\\&=\left(Q^{-1}\left[ {\begin{matrix}
		A&B\\
		C&D
		\end{matrix}} \right] \right)^{-T} \left( P{\mathscr W}_{i+1}(\Delta)\right)^\bot\\&\mathop = \limits^{(\ref{RelationW2})}Q^{T}\left[ \begin{matrix}
	A^T&C^T\\
	B^T&D^T
	\end{matrix} \right]^{-1}\left( \left[ {\begin{matrix}
		\mathcal{W}_i(\Lambda)\\
		* 
		\end{matrix}} \right] + \left[ \begin{matrix}
	0\\
	{\mathscr U}(\Lambda)
	\end{matrix} \right]\right)^\bot.
	\end{array}
	$$
	The above equation gives
	\begin{align}\label{WtildeWreldual}
	Q^{-T}(H{\mathscr W}_{i+1}(\Delta))^\bot=\left[ {\begin{matrix}
		A^T&C^T\\
		B^T&D^T
		\end{matrix}} \right]^{-1}\left[ {\begin{matrix}
		(\mathcal{W}_i({\Lambda}))^\bot\\
		0
		\end{matrix}} \right].
	\end{align}	
	Now  equations (\ref{VtildeVreldual}) and (\ref{WtildeWreldual}) yield that for $i\in \mathbb{N}$,  if $ \mathcal{V}_i(\Lambda^d)=(\mathcal{W}_i(\Lambda))^\bot$, then ${\mathscr V}_{i}(\Delta^d)=(H{\mathscr W}_{i}(\Delta))^\bot$. 
	Thus the proof of (\ref{relVdual}) is complete.
	
	Step 2\red{;} Step 2a: We prove that for $i\in \mathbb{N}$,
	\begin{align}\label{relWdual1}
	{\mathscr W}_{i+1}(\Delta^d)=(E{\mathscr V}_i(\Delta))^\bot\Leftrightarrow \mathbfcal{W}_i(\mathbf{\Lambda}^d)=(\mathbfcal{V}_i(\mathbf{\Lambda}))^\bot.
	\end{align}
	\red{Using} equation (\ref{tildeWdualwongform}) of Lemma \ref{tildeVWdualwongform}, we will prove by induction that for $i\in \mathbb{N}$,  
	\begin{align}\label{WtildeWdualsys}
	H^T{\mathscr W}_{i}(\Delta^d)=P^{T}\mathbfcal{W}_i(\mathbf{\Lambda}^d).
	\end{align}
	For $i=0$, $H^T{\mathscr W}_0(\Delta^d)=P^T\mathbfcal{W}_0(\mathbf{\Lambda}^d)=0$; If $H^T{\mathscr W}_{i}(\Delta^d)=P^T\mathbfcal{W}_i(\mathbf{\Lambda}^d)$, then $$
	H^T{\mathscr W}_{i+1}(\Delta^d)\mathop = \limits^{(\ref{Wrealized})}H^T(E^T)^{-1}H^T{\mathscr W}_{i}(\Delta^d)\!=\!H^T(E^T)^{-1}P^{T}\mathbfcal{W}_i(\mathbf{\Lambda}^d)\mathop = \limits^{(\ref{tildeWdualwongform})}P^{T}\mathbfcal{W}_{i+1}(\mathbf{\Lambda}^d).
	$$ 
	\red{By an induction argument,} (\ref{WtildeWdualsys}) holds for $i\in \mathbb N$.
	
	We now prove $``\Rightarrow"$ of (\ref{relWdual1}): Assume for $i\in \mathbb{N}$, ${\mathscr W}_{i+1}(\Delta^d)=(E{\mathscr V}_i(\Delta))^\bot$, it follows that
	\begin{align*}
	\mathbfcal{W}_{i+1}(\mathbf{\Lambda}^d)\mathop = \limits^{(\ref{WtildeWdualsys})}P^{-T}H^T{\mathscr W}_{i+1}(\Delta^d)=P^{-T}H^T(E{\mathscr V}_i(\Delta))^\bot=(PH^{-1}E{\mathscr V}_i(\Delta))^\bot\\\mathop = \limits^{(\ref{Vrealized})}(P{\mathscr V}_{i+1}(\Delta))=({\mathscr V}_{i+1}(\Delta^{Impl}))^\bot\mathop = \limits^{(\ref{rela VtildeV})}(\mathbfcal{V}_{i+1}(\mathbf{\Lambda}))^\bot.
	\end{align*}	
	We then prove $``\Leftarrow"$ of (\ref{relWdual1}):	Assume for $i\in \mathbb{N}$, $\mathbfcal{W}_i(\mathbf{\Lambda}^d)=(\mathbfcal{V}_i(\mathbf{\Lambda}))^\bot$, it follows that 
	\begin{align*}
	(E{\mathscr V}_i(\Delta))^\bot=E^{-T}({\mathscr V}_i(\Delta))^\bot=E^{-T}(P^{-1}{\mathscr V}_i(\Delta^{Impl}))^\bot\mathop = \limits^{(\ref{rela VtildeV})}E^{-T}(P^{-1}\mathbfcal{V}_i(\mathbf{\Lambda}))^\bot\\=E^{-T}P^{T}\mathbfcal{W}_i(\mathbf{\Lambda}^d)\mathop = \limits^{(\ref{WtildeWdualsys})}E^{-T}H^T{\mathscr W}_{i}(\Delta^d)\mathop = \limits^{(\ref{Wrealized})}{\mathscr W}_{i+1}(\Delta^d),
	\end{align*}
	\red{and the proof of (\ref{relWdual1}) is complete.}
	
	Step 2b: In this step, we show that for $i\in \mathbb{N}$, 
	\begin{align}\label{relVdual1}
	{\mathscr V}_{i}(\Delta^d)=(H{\mathscr W}_i(\Delta))^\bot\Leftrightarrow \mathbfcal{V}_i(\mathbf{\Lambda}^d)=(\mathbfcal{W}_i(\mathbf{\Lambda}))^\bot. 
	\end{align}
	\red{Using} equation (\ref{tildeVdualwongform}) of Lemma \ref{tildeVWdualwongform}, we will prove by induction that for $i\in \mathbb{N}$,     
	\begin{align}\label{VtildeVdualsys}
	{\mathscr V}_{i}(\Delta^d)=(H^T)^{-1}\left( P^T\mathbfcal{V}_i(\mathbf{\Lambda}^d)\right).
	\end{align}
	For $i=0$, ${\mathscr V}_0(\Delta^d)=\mathbb{R}^n=(H^T)^{-1}P^T\mathbfcal{V}_0(\mathbf{\Lambda}^d)$; If ${\mathscr V}_{i}(\Delta^d)=(H^T)^{-1}P^T\mathbfcal{V}_i(\mathbf{\Lambda}^d)$, then we get
	\begin{align*}
	{\mathscr V}_{i+1}(\Delta^d)&\mathop = \limits^{(\ref{Vrealized})}(H^T)^{-1}E^T{\mathscr V}_{i}(\Delta^d)=(H^T)^{-1}E^T(H^T)^{-1}P^T\mathbfcal{V}_i(\mathbf{\Lambda}^d)\\&\mathop = \limits^{(\ref{tildeVdualwongform})}(H^T)^{-1}P^T\mathbfcal{V}_{i+1}(\mathbf{\Lambda}^d).
	\end{align*}  
	\red{By an induction argument,} \red{(\ref{VtildeVdualsys}) holds for $i\in \mathbb N$}.
	
	We now prove $``\Rightarrow"$ of	(\ref{relVdual1}). Assume ${\mathscr V}_{i}(\Delta^d)=(H{\mathscr W}_i(\Delta))^\bot$, then
	\begin{align*}
	P^T\mathbfcal{V}_{i+1}(\mathbf{\Lambda}^d)\mathop = \limits^{(\ref{tildeVdualwongform})}E^TH^{-T}P^T{\mathbfcal{V} }_{i}(\mathbf{\Lambda}^d)\mathop = \limits^{(\ref{VtildeVdualsys})}E^T{\mathscr V}_i(\Delta^d)=E^T(H{\mathscr W}_i(\Delta))^\bot\\=(E^{-1}H{\mathscr W}_i(\Delta))^\bot\mathop = \limits^{(\ref{Wrealized})}({\mathscr W}_{i+1}(\Delta))^\bot=(P^{-1}{\mathscr W}_{i+1}(\Delta^{Impl}))^\bot\mathop=\limits^{(\ref{rela WtildeW})}P^T\mathbfcal W_{i+1}(\mathbf \Lambda),
	\end{align*} 
	We then prove $``\Leftarrow"$ of (\ref{relVdual1}): Assume $\mathbfcal{V}_i(\mathbf{\Lambda}^d)=(\mathbfcal{W}_i(\mathbf{\Lambda}))^\bot$, then for $i\in \mathbb{N}$,
	\begin{align*}
	(H{\mathscr W}_i(\Delta))^\bot&=(H^T)^{-1}({\mathscr W}_i(\Delta))^\bot=(H^T)^{-1}(P^{-1}{\mathscr W}_i(\Delta^{Impl}))^\bot\\&\mathop = \limits^{(\ref{rela WtildeW})}(H^T)^{-1}(P^{-1}\mathbfcal{W}_i(\mathbf{\Lambda}))^\bot=(H^T)^{-1}P^{T}\mathbfcal{V}_i(\mathbf{\Lambda}^d)\mathop= \limits^{(\ref{VtildeVdualsys})}{\mathscr V}_{i}(\Delta^d),
	\end{align*}
	\red{which completes the proof of (\ref{relVdual1})}.
	
	Step 3: Since $\mathscr V^*$, $\mathscr V^*$, $\mathcal{V}^*$, $\mathcal{W}^*$, $\mathbfcal{V}^*$, $\mathbfcal{W}^*$ are the limites of $\mathscr V_i$, $\mathscr V_i$, $\mathcal{V}_i$, $\mathcal{W}_i$, $\mathbfcal{V}_i$, $\mathbfcal{W}_i$, respectively, equations (\ref{relWdual}) and (\ref{relVdual}) prove \red{that} $(i)\Leftrightarrow (ii)$ \red{holds,} and equations (\ref{relWdual1}) and (\ref{relVdual1}) prove \red{that} $(i)\Leftrightarrow (iii)$ \red{holds}.
\end{proof}
\subsection{Proof of Proposition \ref{invariants relation}}\label{Pf:invariants relation}
\begin{proof}
	Note that the Kronecker invariants are \red{invariant} under ex-equivalence. By $\Delta\mathop \sim\limits^{ex}\Delta^{Impl}$, \red{in our proof we can work with} the Kronecker invariants of $\Delta^{Impl}$ instead of \red{those} of $\Delta$. In \red{what follows}, we will use the results of Lemma \ref{subs relationex} given in Section \ref{sec:propolem}.	
	
	(i) Recall Lemma \ref{invariants of DAEs}(i) for $\Delta^{Impl}$ and Lemma \ref{invariants of contrsys}(i) for $\Lambda$.  \red{For $i\in \mathbb{N}^+$, it holds that,}
	\begin{align}\label{Kirep}
	\mathscr{K}_i(\Delta^{Impl})&={\mathscr W}_i(\Delta^{Impl})\cap \mathscr V^*(\Delta^{Impl}) \nonumber\\
	&\mathop=\limits^{\rm Lemma \ \ref{subs relationex}}   \left( {\left[ {\begin{matrix}
			{\mathcal{W}_{i-1}(\Lambda)}\\
			*
			\end{matrix}} \right] + \left[ {\begin{matrix}
			0\\
			{\mathscr U}(\Lambda)
			\end{matrix}} \right]} \right) \cap \left( {\left[ {\begin{matrix}
			{\mathcal{V}^*(\Lambda)}\\
			F^*\mathcal{V}^*(\Lambda)
			\end{matrix}} \right] + \left[ {\begin{matrix}
			0\\
			\mathcal{U}^*(\Lambda)
			\end{matrix}} \right]} \right)  \nonumber \\ 
	&=  {\left[ {\begin{matrix}
			\mathcal{W}_{i-1}(\Lambda) \cap \mathcal{V}^*(\Lambda ) \\
			F^*\left( \mathcal{W}_{i-1}(\Lambda) \cap \mathcal{V}^*(\Lambda )\right) 
			\end{matrix}} \right] + \left[ {\begin{matrix}
			0\\
			\mathcal{U}^*(\Lambda)
			\end{matrix}} \right]}, 
	\end{align}
	\red{for a suitable} $F^*\in \mathbb{F}({{\mathcal{V}^*}(\Lambda)}) $. Then we have 
	\begin{align*}
	a\mathop=\limits^{\rm Lemma \ \ref{invariants of DAEs}(i)}\dim\,\left({\mathscr{K}}_1(\Delta^{Impl})\right)\mathop = \limits^{(\ref{Kirep})}\dim\,\left(  \left[ {\begin{matrix}
		0\\
		\mathcal{U}^*(\Lambda)
		\end{matrix}} \right]\right) =\dim\, (\mathcal{U}^*(\Lambda))\mathop=\limits^{\rm Lemma\ \ref{invariants of contrsys}(i)}a'.
	\end{align*}
	Moreover, it is seen that for $i\in \mathbb{N}$,
	\begin{align*}
	{\omega_i} &\mathop=\limits^{\rm Lemma \ \ref{invariants of DAEs}(i)} \dim\, \left( {{\mathscr{K}_{i + 2}}} (\Delta^{Impl})\right) - \dim\, \left( {{\mathscr{K}_{i+1}}}(\Delta^{Impl}) \right)\\
	&\mathop = \limits^{(\ref{Kirep})}\dim\, \left({\mathcal{W}_{i+1}}(\Lambda)\cap \mathcal{V}^*(\Lambda)\right)-\dim\, \left({\mathcal{W}_{i}}(\Lambda)\cap \mathcal{V}^*(\Lambda)\right)\\
	&=\dim\,(\mathcal R_{i+1}(\Lambda))-\dim\,(\mathcal R_{i}(\Lambda))\mathop=\limits^{\rm Lemma \ \ref{invariants of contrsys}(i)}\omega'_i.
	\end{align*}
	Now consider equations (\ref{varepsilon}) and (\ref{alpha}) \red{and} it is sufficient to show 
	\begin{align*}
	\left\{\begin{array}{lll}
	\begin{matrix}
	\varepsilon_j=\varepsilon'_j=0\\
	\varepsilon_j=\varepsilon'_j=i
	\end{matrix}&\begin{matrix}
	{\rm for}\\
	{\rm for}
	\end{matrix}&\begin{matrix}
	{1\le j \le a- \omega_0=a'- \omega'_0},\\
	a' - \omega'_{i - 1} + 1	=a - { \omega_{i - 1}} + 1\le j \le a- \omega_{i}=a'-{ \omega'}.
	\end{matrix}
	\end{array} \right.
	\end{align*}
	The statement that $d=d'$, $\eta_i=\eta'_i$ can be \red{proved in a similar way using dual objects}. It is not hard to see that for $i\in \mathbb{N}^+$,
	\begin{align*}
	\hat {\mathscr K}_i(\Delta^{Impl})&=\left(E\mathscr V_{i-1}(\Delta^{Impl})\right)^{\bot }\cap (H\mathscr W^*(\Delta^{Impl}))^{\bot }\\&\mathop=\limits^{\rm Prop. \ \ref{Subspaces of the dual system}(i)}\mathscr W_{i}((\Delta^{Impl})^d)\cap \mathscr V^*((\Delta^{Impl})^d)\\&\mathop=\limits^{\rm Lemma \ \ref{subs relationex}} {\left[ {\begin{matrix}
			\mathcal{W}_{i-1}(\Lambda ^d) \cap \mathcal{V}^*(\Lambda ^d) \\
			*
			\end{matrix}} \right] + \left[ {\begin{matrix}
			0\\
			\mathcal{U}^*(\Lambda ^d)
			\end{matrix}} \right]},
	\end{align*}
	where $(\Delta^{Impl})^d$ is the dual system of $\Delta^{Impl}$, which \red{coincides with} ${\rm Impl}(\Lambda^d)$.	It follows that
	\begin{align*}
	d\mathop=\limits^{\rm Lemma \ \ref{invariants of DAEs} (i)}\dim\,\left(\hat{\mathscr{K}}_1(\Delta^{Impl})\right)=\dim\,\left(  \left[ {\begin{matrix}
		0\\
		\mathcal{U}^*(\Lambda^d)
		\end{matrix}} \right]\right) =\dim\, ({\mathcal Y^*(\Lambda)})\mathop=\limits^{\rm Lemma \ \ref{invariants of contrsys}(i)}d'.
	\end{align*}
	We can also see that for $i\in \mathbb{N}$,
	\begin{align*}
	\hat {\omega}_i &= \dim\, \left( {\hat{\mathscr{K}_{i + 2}}}(\Delta^{Impl}) \right) - \dim\, \left( {\hat{\mathscr{K}_{i+1}}}(\Delta^{Impl}) \right)\\&=\dim\, \left({\mathcal{W}_{i+1}}(\Lambda ^d)\cap \mathcal{V}^*(\Lambda ^d)\right)-\dim\, \left({\mathcal{W}_{i}}(\Lambda ^d)\cap \mathcal{V}^*(\Lambda ^d)\right)\\
	&\mathop=\limits^{\rm Prop. \ref{Subspaces of the dual system}}\dim\,\left( (\mathcal{V}_{i+1})^{ \bot }\cap (\mathcal{W}^*)^{ \bot }\right) -\dim\,\left( (\mathcal{V}_{i})^{ \bot }\cap (\mathcal{W}^*)^{ \bot }\right) \\&=\dim\,(\hat {\mathcal R}_{i+1}({\Lambda}))-\dim\,(\hat {\mathcal R}_{i}({\Lambda}))=\hat \omega'_i.
	\end{align*}
	Now it is sufficient to show that 
	\begin{align*}
	\left\lbrace \begin{array}{lll}
	{\begin{matrix}
		{{\eta_j}}=\eta'_j=0\\
		{{\eta_j}}=\eta'_j=i
		\end{matrix}}&{\begin{matrix}
		{\rm for}\\
		{\rm for}
		\end{matrix}}&{\begin{matrix}
		{1\le j \le d- \hat \omega_0=h- \hat \omega'_0},\\
		{h - { \omega'_{i - 1}} + 1	={d - { \hat \omega_{i - 1}} + 1\le j \le d-{ \hat \omega_{i}}}=h-{ \hat \omega'}}.
		\end{matrix}}
	\end{array} \right.
	\end{align*}
	
	(ii) Recall Lemma \ref{invariants of DAEs}(ii) for $\Delta^{Impl}$ and Lemma \ref{invariants of contrsys}(ii) for $\Lambda$. We have for all $i\in\mathbb{N}^+$,
	\begin{align*}
	{\mathscr V^*}(\Delta^{Impl})+{\mathscr W}_i(\Delta^{Impl})&\mathop=\limits^{\rm Lemma \ref{subs relationex}}\left[ {\begin{matrix}
		{\mathcal{V}^*(\Lambda)}\\
		{F^* * \mathcal{V}^*(\Lambda)}
		\end{matrix}} \right] +  \left[ {\begin{matrix}
		0\\
		{{\mathcal{U}_i}(\Lambda)}
		\end{matrix}} \right]+ \left[ {\begin{matrix}
		{\mathcal{W}_{i-1}(\Lambda)}\\
		* 
		\end{matrix}} \right] + \left[ {\begin{matrix}
		0\\
		{\mathscr U}(\Lambda)
		\end{matrix}} \right]\\
	&= \left[ {\begin{matrix}
		{\mathcal{V}^*(\Lambda)}+{\mathcal{W}_{i-1}(\Lambda)}\\
		* 
		\end{matrix}} \right] + \left[ {\begin{matrix}
		0\\
		{\mathscr U}(\Lambda)
		\end{matrix}} \right].
	\end{align*}
	If $\nu=0$, then we have the following result \red{by} (\ref{nu}):
	\begin{align*}
	&{\mathscr V^*(\Delta^{Impl})}+{\mathscr W}_0(\Delta^{Impl})={\mathscr V^*}(\Delta^{Impl})+{\mathscr W}_{1}(\Delta^{Impl})\Rightarrow\\ &\left(\left[ {\begin{matrix}
		{\mathcal{V}^*(\Lambda)}\\
		F^*{\mathcal{V}^*(\Lambda)} 
		\end{matrix}} \right] + \left[ {\begin{matrix}
		0\\
		\mathcal{U}^*(\Lambda)
		\end{matrix}} \right]\right)=\left(\left[ {\begin{matrix}
		{\mathcal{V}^*(\Lambda)}\\
		* 
		\end{matrix}} \right] + \left[ {\begin{matrix}
		0\\
		{\mathscr U}(\Lambda)
		\end{matrix}} \right]\right)\Rightarrow {\mathscr U}(\Lambda)=\mathcal{U}^*(\Lambda).
	\end{align*}
	It follows that  $c'=\dim\, \left({\mathscr U}(\Lambda)\right) - \dim\, \left( \mathcal{U}^*(\Lambda)\right)=0$. Therefore, in this case, the  $MCF^3$-part of \textbf{MCF} is absent. \red{As a consequence}, if $N(s)$ of \textbf{KCF} is absent, then $MCF^3$ of \textbf{MCF} is absent \red{as well}. If $\nu>0$, from (\ref{nu}) we get
	\begin{align*}
	\nu&={\rm min}\left \{i\in \mathbb{N}^+\left|\left[ {\begin{matrix}
		{\mathcal{V}^*(\Lambda)}+{\mathcal{W}_{i-1}(\Lambda)}\\
		* 
		\end{matrix}} \right] + \left[ {\begin{matrix}
		0\\
		{\mathscr U}(\Lambda)
		\end{matrix}} \right]=\left[ {\begin{matrix}
		{\mathcal{V}^*(\Lambda)}+{{\mathcal{W}_{i}}(\Lambda)}\\
		* 
		\end{matrix}} \right] + \left[ {\begin{matrix}
		0\\
		{\mathscr U}(\Lambda)
		\end{matrix}} \right]\right.\right\}\\
	&={\rm min}\left \{i\in \mathbb{N}^+\left|{\mathcal{V}^*(\Lambda)}+{\mathcal{W}_{i-1}(\Lambda)}={\mathcal{V}^*(\Lambda)}+{{\mathcal{W}_{i}}(\Lambda)}\right.\right\}=\nu'+1.
	\end{align*}
	We have		 
	\begin{align*}
	c=\pi_0&=\dim\,\left ( {{\mathscr V}^*}(\Delta^{Impl})+{{ {\mathscr W}}_{ 1}} (\Delta^{Impl})\right) - \dim\, \left({ {\mathscr V}^*(\Delta^{Impl})}+{{ {\mathscr W}}_0(\Delta^{Impl})}  \right)\\
	&\mathop=\limits^{\rm Lemma \ \ref{subs relationex}}\dim\,\left(\left[ {\begin{matrix}
		{\mathcal{V}^*(\Lambda)}\\
		* 
		\end{matrix}} \right] + \left[ {\begin{matrix}
		0\\
		{\mathscr U}(\Lambda)
		\end{matrix}} \right]\right)-\dim\,\left(\left[ {\begin{matrix}
		{\mathcal{V}^*(\Lambda)}\\
		* 
		\end{matrix}} \right] + \left[ {\begin{matrix}
		0\\
		{\mathcal{U}(\Lambda)}
		\end{matrix}} \right]\right)\\
	&= \dim\, \left( {\mathscr U}(\Lambda)\right) - \dim\, \left(\mathcal{U}(\Lambda)\right)=c'.
	\end{align*}
	We also have for $i\in \mathbb{N}^+$,
	\begin{align*}
	{\pi _i} &=  \dim\,\left ({{\mathscr V}^*}(\Delta^{Impl})+{{ {\mathscr W}}_{i + 1}} (\Delta^{Impl})\right) - \dim\, \left({ {\mathscr V}^*(\Delta^{Impl})}+{{ {\mathscr W}}_i(\Delta^{Impl})} \right)\\
	&=\dim\,\left(\left[ {\begin{matrix}
		{\mathcal{V}^*(\Lambda)}+{{\mathcal{W}_{i}}(\Lambda)}\\
		* 
		\end{matrix}} \right] + \left[ {\begin{matrix}
		0\\
		{\mathscr U}(\Lambda)
		\end{matrix}} \right]\right)-\dim\,\left(\left[ {\begin{matrix}
		{\mathcal{V}^*(\Lambda)}+{\mathcal{W}_{i-1}(\Lambda)}\\
		* 
		\end{matrix}} \right] + \left[ {\begin{matrix}
		0\\
		{\mathscr U}(\Lambda)
		\end{matrix}} \right]\right)\\
	&= \dim\, \left({\mathcal{W}_{i }}(\Lambda) + {\mathcal{V}^*(\Lambda)}\right) - \dim\, \left(\mathcal{W}_{i-1}(\Lambda) + {\mathcal{V}^*}(\Lambda)\right)=\pi'_{i-1}.
	\end{align*}
	Now \red{substituting} $c=c'$, $\pi_i=\pi'_{i-1}$ and $\nu=\nu'+1$ into (\ref{sigma}),  we can rewrite equation (\ref{sigma}) as
	$$
	\left\lbrace 	\begin{array}{lll}
	{\sigma_j\!=\!0} \ \ {\rm for} \ \ 1\le j \le c-\pi_1= c'-\pi'_0=\delta,\\
	{\sigma_j\!=\!i} \ \ {\rm for} \ \ c'- \pi'_{i - 2} + 1\!=\!c- \pi _{i - 1} + 1\le j\le c- {\pi_i}\!=\!c'- {\pi'_{i-1}}, \  { i\!=\!2,...,\nu'+1}.
	\end{array}	\right. 
	$$
	\red{Replacing $i$ by $i-1$, we get}
	$$
	\begin{array}{*{20}{c}}
	{\sigma_j=i-1}&for&c'- \pi'_{i - 1} + 1\le j\le c'- {\pi'_{i}},&{ i=1,2,...,\nu'}.
	\end{array}	
	$$
	Finally, \red{compare} the above expression of $\sigma_j$ with \red{that for} $\sigma'_j$ of (\ref{sigma'}), it is not hard to see \red{that} $\sigma_j+1=\sigma'_j$ for $j= 1,\dots,c$.
	
	(iii) We only show that the invariant factors of $MCF^2$ of $\Lambda$ coincide with the invariant factors of the real Jordan pencil $J(s)$ of $\Delta^{Impl}$, then the \red{equalities}  $d=d'$, $\eta_1=\eta'_1,\cdots,\eta_d=\eta'_{d'}$ and $\lambda_{\rho_1}=\lambda_{\rho'_1},\dots,\lambda_{\rho_b}=\lambda_{\rho'_{b'}}$ are \red{immediately} satisfied. 
	\red{First}, let two subspaces $\mathscr X_2\subseteq\mathscr V^*(\Delta^{Impl})$ and $\mathscr Z_2\subseteq\mathcal V^*(\Lambda)$ be \red{such} that 
	$$
	\mathscr X_2\oplus\left( \mathscr{V^*}(\Delta^{Impl}) \cap \mathscr{W^*}(\Delta^{Impl})\right) =\mathscr{V^*}(\Delta^{Impl}) , \ \ \ \	\mathscr Z_2\oplus\left( \mathcal V^*(\Lambda) \cap \mathcal W^*(\Lambda)\right) =\mathcal V^*(\Lambda).
	$$
	The above construction \red{gives} $\Delta^{Impl}|\mathscr X_2\cong KCF^2$ and $\Lambda|\mathscr Z_2\cong MCF^2$, where $KCF^2$ corresponds to \red{the} Jordan pencil $J(s)$. Use Lemma \ref{subs relationex} \red{to conclude that} $$\mathscr X_2\oplus\left( \mathscr{V^*}(\Delta^{Impl}) \cap  \mathscr{W^*}(\Delta^{Impl})\right) =\mathscr{V^*}(\Delta^{Impl})$$ \red{implies}
	\begin{align*}
	&\mathscr X_2\oplus \left( \left({\left[ {\begin{matrix}
			{\mathcal{W}^*(\Lambda)}\\
			*
			\end{matrix}} \right] + \left[ {\begin{matrix}
			0\\
			{\mathscr U}(\Lambda)
			\end{matrix}} \right]}\right)  \cap \left(  {\left[ {\begin{matrix}
			{\mathcal{V}^*(\Lambda)}\\
			{{F^*}\mathcal{V}^*(\Lambda)}
			\end{matrix}} \right] + \left[ {\begin{matrix}
			0\\
			\mathcal{U}^*(\Lambda)
			\end{matrix}} \right]} \right) \right)\\&=\left({\left[ {\begin{matrix}
			{\mathcal{V}^*(\Lambda)}\\
			{{F^*}\mathcal{V}^*(\Lambda)}
			\end{matrix}} \right] + \left[ {\begin{matrix}
			0\\
			\mathcal{U}^*(\Lambda)
			\end{matrix}} \right]} \right)  \\
	&\Rightarrow\mathscr X_2\oplus\left({\left[ {\begin{matrix}
			{\mathcal{W}^*(\Lambda)}\cap{\mathcal{V}^*(\Lambda)}\\
			F'\left( {\mathcal{W}^*(\Lambda)}\cap{\mathcal{V}^*(\Lambda)}\right) 
			\end{matrix}} \right] +\left[ {\begin{matrix}
			0\\
			\mathcal{U}^*(\Lambda)
			\end{matrix}} \right]} \right)=\left({\left[ {\begin{matrix}
			{\mathcal{V}^*(\Lambda)}\\
			{{F^*}\mathcal{V}^*(\Lambda)}
			\end{matrix}} \right] + \left[ {\begin{matrix}
			0\\
			\mathcal{U}^*(\Lambda)
			\end{matrix}} \right]} \right),
	\end{align*}
	\red{where $F\in\mathbb{F}({\mathcal{V}^*(\Lambda)})$, $F'\in \mathbb{F}({\mathcal{W}^*(\Lambda)}\cap{\mathcal{V}^*(\Lambda)})$.}	Since $\mathscr Z_2\oplus\left( \mathcal V^*(\Lambda) \cap \mathcal W^*(\Lambda)\right) =\mathcal V^*(\Lambda)$, we have $\mathscr X_2=\left[ {\begin{matrix}
		\mathscr Z_2\\
		F''\mathscr Z_2
		\end{matrix}} \right] $, where $F''\in \mathbb F(\mathscr Z_2)$. Then, it follows that
	\begin{align*}
	\left. {\left[ {\begin{matrix}
			s I- A&-B\\
			-C&-D
			\end{matrix}} \right]} \right|_{\mathscr X_2}&=\left[ {\begin{matrix}
		{s I- {A}}&{ - {B}}\\
		{ - {C}}&{ - {D}}
		\end{matrix}} \right]\left[ {\begin{matrix}
		\mathscr Z_2\\
		F''\mathscr Z_2
		\end{matrix}} \right]=\left[ {\begin{matrix}
		\left( sI-(A+BF'')\right)\mathscr Z_2  \\
		(C+DF'')\mathscr Z_2  
		\end{matrix}} \right]
	\\&=\left[ {\begin{matrix}
		\left( sI-(A+BF'')\right)\mathscr Z_2  \\
		0  
		\end{matrix}} \right].
	\end{align*}
	Now	it is known from  Lemma 4.1 of \cite{morse1973structural} that $(A+BF'')|\mathscr Z_2 $ \red{does not dependent on the choice of} $F''$. Thus the invariant factors of $\left( sI-(A+BF'')\right)\mathscr Z_2$  coincide with the invariant factors of $MCF^2$ for $\Lambda$.  Finally, from \red{the} above equation, it is easy to see \red{that} the invariant factors of $J(s)$ in \textbf{KCF} of $\Delta$ coincide with \red{those} of $MCF^2$ of $\Lambda$.
\end{proof}
\subsection{Proof of Proposition \ref{Pro:M}}\label{Pf:Pro:M}
\begin{proof}
	(i) By Proposition \ref{Pro:inv-subspace}, $\mathscr M$ is an invariant subspace if and only if $H\mathscr M\subseteq E\mathscr M$. Therefore, $\mathscr M^*$ is the largest subspace such that $H\mathscr M^*\subseteq E\mathscr M^*$, then by Proposition \ref{Pro:VW}(ii), we have $\mathscr M^*={\mathscr V^*}$.
	
	(ii)  By Proposition \ref{Pro:inv-subspace}, for $\Delta|^{red}_{\mathscr M^*}=(E|^{red}_{\mathscr M^*},H|^{red}_{\mathscr M^*})$, the matrix $E|^{red}_{\mathscr M^*}$ is of full row rank. Thus from the explicitation procedure, it is straightforward to see that $\Lambda^*\in {\rm Expl}(\Delta|^{red}_{\mathscr M^*})$ is a control system without outputs.
	\red{ Note that, by the definitions of reduction and restriction, if two DAEs $\Delta\mathop\sim\limits^{ex}\tilde \Delta$, then $\Delta|^{red}_{\mathscr M^*}\mathop\sim\limits^{ex} \tilde \Delta|^{red}_{\tilde{\mathscr M}^*}$.}    Denote the four parts of the \textbf{KCF} of $\Delta$ as \red{$KCF^k$, $k=1,\dots,4$} and \red{the corresponding} matrix pencil of each part \red{is}:
	$$
	L(s) \ \  {\rm for}	\ \ KCF^1 , \ \ \  J(s) \ \  {\rm for}	\ \ KCF^2, \ \ \  N(s) \ \  {\rm for}	\ \ KCF^3, \ \ \ L^p(s) \ \  {\rm for}	\ \ KCF^4. \ \ \
	$$
\red{	By $\Delta \mathop\sim\limits^{ex} \mathbf{KCF}$, we have 
	\begin{align}\label{Eq:DeltaKCF}
	\Delta|^{red}_{\mathscr M^*}\mathop\sim\limits^{ex} \mathbf{KCF}|^{red}_{\tilde{\mathscr M}^*}=\left(KCF^1,KCF^2 \right).
	\end{align}	 
Moreover, it is clear that if two control systems $\Lambda \mathop\sim\limits^{M} \tilde \Lambda$, then $\Lambda|^{red}_{(\mathcal V^*,\mathcal U^*)}\mathop\sim\limits^{M}\tilde \Lambda|^{red}_{(\tilde {\mathcal V}^*,\tilde {\mathcal U}^*)}$. Since $\Lambda$ is always M-equivalent to its \textbf{MCF}, we have 
 \begin{align}\label{Eq:LambdaMCF}
 \Lambda|^{red}_{(\mathcal V^*,\mathcal U^*)}\mathop\sim\limits^{M}\mathbf{MCF}|^{red}_{(\tilde {\mathcal V}^*,\tilde {\mathcal U}^*)}=\left(MCF^1,MCF^2 \right).
 	\end{align}	
It is seen that $\Lambda|^{red}_{(\mathcal V^*,\mathcal U^*)}$ is a control system without outputs. From the one-to-one correspondence of the \textbf{KCF} and \textbf{MCF} discussed in Section \ref{Chap1sec5}, it is {straightforward} to see that $\left(MCF^1,MCF^2 \right)\in {\rm Expl}(KCF^1,KCF^2)$. Now combining the later result with the relations of (\ref{Eq:DeltaKCF}) and (\ref{Eq:LambdaMCF}), and using the results of Theorem \ref{main theorem}, we can deduce that $
\Lambda|^{red}_{(\mathcal V^*,\mathcal U^*)} \in {\rm Expl}(\Delta|^{red}_{\mathscr M^*})$.} Since $\Lambda^*\in {\rm Expl}(\Delta|^{red}_{\mathscr M^*})$, by Theorem \ref{main theorem}(ii) we have $\Lambda|^{red}_{(\mathcal V^*,\mathcal U^*)} \mathop  \sim \limits^{M} \Lambda^*$. Finally, since $\Lambda^*$ and $\Lambda|^{red}_{(\mathcal V^*,\mathcal U^*)}$ are two control systems without outputs,  their Morse equivalence reduces to their feedback equivalence (see Remark \ref{Rem:M-equi}).
\end{proof}
\subsection{Proof of Theorem \ref{in-equi}}\label{Pf:in-equi}
\begin{proof}
	$(i)\Leftrightarrow(ii)$: By Definition \ref{in-equivalent}, we have $\Delta\mathop  \sim \limits^{in} \tilde \Delta$ if and only if $\Delta|^{red}_{\mathscr M^*}\mathop  \sim \limits^{ex}\tilde \Delta|^{red}_{\mathscr M^*}$. Consider $\Lambda^*\in {\rm Expl}(\Delta|^{red}_{\mathscr M^*})$ and $\tilde \Lambda^*\in {\rm Expl}(\tilde \Delta|^{red}_{\tilde {\mathscr M}^*})$, \red{then by} Theorem \ref{main theorem}(ii), it follows that $\Delta|^{red}_{\mathscr M^*}\mathop  \sim \limits^{ex}\tilde \Delta|^{red}_{\mathscr M^*}$ if and only if $\Lambda^*\mathop  \sim \limits^{M} \tilde \Lambda^*$. By Proposition \ref{Pro:M}(ii), $\Lambda^*$ and $\tilde \Lambda^*$ are two control systems without outputs, 
	which implies that their Morse equivalence reduces to their feedback equivalence (see Remark \ref{Rem:M-equi}).
	
	$(ii)\Leftrightarrow(iii)$: We first prove that two DAEs $\Delta^*={\rm Impl}(\Lambda^*)$ and $\tilde \Delta^*={\rm Impl}(\tilde \Lambda^*)$ have isomorphic trajectories if and only if  $\Lambda^*$ and $\tilde {\Lambda}^*$ are feedback equivalent. \red{Let} $(z(t), u(t))$ and $(\tilde z(t), \tilde u(t))$ denote trajectories of $\Delta^*$ and $\tilde \Delta^*$, respectively.
	Suppose $\Lambda^*$ and $\tilde \Lambda^*$ are feedback equivalent, then there exist matrices $T_s\in Gl(n^*,\mathbb{R})$, $T_i\in Gl(m^*,\mathbb{R})$, $F\in \mathbb{R}^{{m^*}\times {n^*}}$ such that $\tilde A^*=T_s(A^*+B^*F){T_s^{-1}}$, $\tilde B^*=T_sB{T_i^{-1}}$. Since $\Lambda^*$ has no output, its implicitation (see Definition \ref{implicitation}) is 
	\begin{align*}
	\Delta^*:\left[ {\begin{matrix}
		I&0
		\end{matrix}} \right]\left[ {\begin{matrix}
		{\dot z}\\
		{\dot u}
		\end{matrix}} \right] = \left[ {\begin{matrix}
		{{A^ * }}&{{B^ * }}
		\end{matrix}} \right]\left[{\begin{matrix}
		{z}\\
		{u}
		\end{matrix}} \right].
	\end{align*}	
	For $\tilde \Lambda^*$,  its implicitation is
	\begin{align*}
	\tilde \Delta^*: \left[ {\begin{matrix}
		I&0
		\end{matrix}} \right]\left[ {\begin{matrix}
		{\dot {\tilde z}}\\
		{\dot {\tilde u}}
		\end{matrix}} \right] = \left[ {\begin{matrix}
		\tilde A^*&\tilde B^*
		\end{matrix}} \right]\left[ {\begin{matrix}
		{\tilde z}\\
		{\tilde u}
		\end{matrix}} \right]\Rightarrow	
	\left[ {\begin{matrix}
		I&0
		\end{matrix}} \right]\left[ {\begin{matrix}
		{\dot {\tilde z}}\\
		{\dot {\tilde u}}
		\end{matrix}} \right] = {T_s}\left[ {\begin{matrix}
		A^*&B^*
		\end{matrix}} \right]\left[ {\begin{matrix}
		{T_s^{ - 1}}&0\\
		{FT_s^{ - 1}}&{T_i^{ - 1}}
		\end{matrix}} \right]\left[{\begin{matrix}
		{\tilde z}\\
		{\tilde u}
		\end{matrix}} \right].
	\end{align*}
	It can be seen that any trajectory $(z(t), u(t))$ of $\Delta^*$ satisfying $z(0)=z^0$ and $u(0)=u^0$, is mapped via $T=\left[ {\begin{matrix}
		{T_s^{ - 1}}&0\\
		{FT_s^{ - 1}}&{T_i^{ - 1}}
		\end{matrix}} \right]^{-1}$ into a trajectory $(\tilde z(t), \tilde u(t))$ of 	$\tilde \Delta^*$ passing through $\left[ \begin{matrix}
	\tilde z^0\\
	\tilde u^0
	\end{matrix}\right]=T \left[ \begin{matrix}
	z^0\\
	u^0
	\end{matrix}\right]$. 
	
	Conversely, suppose that there exists an invertible matrix $T = \left[ {\begin{matrix}
		{{T_1}}&{{T_2}}\\
		{{T_3}}&{{T_4}}
		\end{matrix}} \right]$ such that $\left[ {\begin{matrix}
		{\tilde z\left( t \right)}\\
		{\tilde u\left( t \right)}
		\end{matrix}} \right] = \left[ {\begin{matrix}
		{{T_1}}&{{T_2}}\\
		{{T_3}}&{{T_4}}
		\end{matrix}} \right]\left[ {\begin{matrix}
		{z\left( t \right)}\\
		{u\left( t \right)}
		\end{matrix}} \right]$.  It follows that $(\tilde z(t),\tilde u(t))$, being a solution of $\tilde \Delta^*$, satisfies
	\begin{align*}
	\left[ {\begin{matrix}
		I&0
		\end{matrix}} \right]\left( {\begin{matrix}
		{\dot {\tilde z}\left( t \right)}\\
		{\dot {\tilde u}\left( t \right)}
		\end{matrix}} \right) = \left[ {\begin{matrix}
		\tilde A^*&\tilde B^*
		\end{matrix}} \right]\left( {\begin{matrix}
		{\tilde z\left( t \right)}\\
		{\tilde u\left( t \right)}
		\end{matrix}} \right),
	\end{align*}
	which implies 
	\begin{align*}
	\left[ {\begin{matrix}
		I&0
		\end{matrix}} \right]\left[ {\begin{matrix}
		{{T_1}}&{{T_2}}\\
		{{T_3}}&{{T_4}}
		\end{matrix}} \right]\left( {\begin{matrix}
		{\dot { z}\left( t \right)}\\
		{\dot { u}\left( t \right)}
		\end{matrix}} \right) = \left[ {\begin{matrix}
		\tilde A^*&\tilde B^*
		\end{matrix}} \right]\left[ {\begin{matrix}
		{{T_1}}&{{T_2}}\\
		{{T_3}}&{{T_4}}
		\end{matrix}} \right]\left( {\begin{matrix}
		{ z\left( t \right)}\\
		{ u\left( t \right)}
		\end{matrix}} \right).
	\end{align*}
	Since  $(z(t), u(t))$ satisfies $\dot z(t)=A^*z(t)+B^*u(t)$, it follows that 
	\begin{align}\label{Eq:prf_Thm2}
	{T_1}\dot z(t)+T_2\dot u(t)=(\tilde A^*T_1+\tilde B^*T_3)z(t)+(\tilde A^*T_2+\tilde B^*T_4)u(t)\Rightarrow \nonumber\\ 
	T_1(A^*z(t)+B^*u(t))+T_2\dot u(t)=(\tilde A^*T_1+\tilde B^*T_3)z(t)+(\tilde A^*T_2+\tilde B^*T_4)u(t).
	\end{align}
	Notice that equation (\ref{Eq:prf_Thm2}) is satisfied for any solution $(z(t),u(t))$ of $\Delta^*$. (a). Let $u(t)\equiv0$ and $(z(t,z^0),0)$ (where $z^0\ne0$) be a solution of $\Delta^*$ (obviously, such a solution always exists). By substituting this solution into (\ref{Eq:prf_Thm2}) and considering \red{it} for $t=0$, we have $T_1A^*z^0=(\tilde A^*T_1+\tilde B^*T_3)z^0$, where $z^0=z(0)$ \red{can be taken arbitrary}, which implies $A^*=T_1^{ - 1}(\tilde A^*+\tilde B^*(T_3{T_1^{-1}}))T_1$. (b).  \red{Fix $z(0)=z^0=0$ and set $u(t)=u^i(t)=\left[0,\dots,t,\dots,0 \right]^T $, where $t$ is in the $i$-th row. Evaluating at $t=0$, we have  $z(0)=0$, $u(0)=0$ and $\dot u^i(0)=\left[0,\dots,1,\dots,0 \right]^T$, and thus by (\ref{Eq:prf_Thm2}) we have $T_2\dot u^i(0)=0$. So taking controls, $u^1(t),\dots,u^{m^*}(t)$ of that form, we conclude that $T_2=0$.} Now it is easy to see from (\ref{Eq:prf_Thm2}) that $B^*=T_1^{ - 1}\tilde B^*T_4$. Thus $\Lambda^*$ and $\tilde {\Lambda}^*$ are feedback equivalent \red{(see Remark \ref{Rem:M-equi})} via $T_s=T_1$, $T_i=T^{-1}_4$ and $F=T_3{T_1^{-1}}$. Therefore, any trajectory of $\Delta^*$ is transformed via $T$ into a trajectory of $\tilde \Delta^*$ if and only if $ \Lambda^*$ and $\tilde \Lambda^*$ are feedback equivalent.
	
	Then by Theorem \ref{main theorem}(i), we have 
	$$
	\Delta|^{red}_{\mathscr M^*}\mathop  \sim \limits^{ex}\Delta^*={\rm Impl}(\Lambda^*) \ \  {\rm and} \ \ \tilde \Delta|^{red}_{\tilde {\mathscr M^*}}\mathop  \sim \limits^{ex}\tilde\Delta^*={\rm Impl}(\tilde \Lambda^*)
	$$ 
	(since $\Lambda^*\in {\rm Expl}(\Delta|^{red}_{\mathscr M^*})$ and $\tilde \Lambda^*\in {\rm Expl}(\tilde \Delta|^{red}_{\tilde {\mathscr M}^*})$). Moreover, by Remark \ref{ex-quitra}, there exist matrices $P\in Gl(n^*,\mathbb {R})$ and $\tilde P\in Gl(n^*,\mathbb {R})$ such that 
	any trajectory  of $\Delta|^{red}_{\mathscr M^*}$ is mapped via $P$ into the corresponding trajectory of $\Delta^*$ and any trajectory  of $\tilde \Delta|^{red}_{\tilde {\mathscr M^*}}$ is mapped via $\tilde P$ into the corresponding trajectory of $\tilde\Delta^*$. Now we can conclude that the linear and invertible map $S=PT\tilde P^{-1}$ \red{sends} any trajectory of $\Delta|^{red}_{\mathscr M^*}$ into the corresponding trajectory of $\tilde \Delta|^{red}_{\tilde {\mathscr M^*}}$ if and only if $\Lambda^*$ and $\tilde {\Lambda}^*$ are feedback equivalent. 
\end{proof}
\subsection{Proof of Proposition \ref{Pro:in-regular}}\label{Pf:Pro:in-regular}
\begin{proof}
	$(i)\Leftrightarrow(ii)$:  Consider a DAE $\Delta^*={\rm Impl}(\Lambda^*)$. We have $\Delta|^{red}_{\mathscr M^*}\mathop  \sim \limits^{ex} \Delta^* $ (implied by  $\Lambda^*\in {\rm Expl}(\Delta|^{red}_{\mathscr M^*})$ and Theorem \ref{main theorem}(i)), we get  $\Delta|^{red}_{\mathscr M^*}\mathop  \sim \limits^{ex} \Delta^*$. Actually, since $\Lambda^*$ is defined on $\mathscr M^*$, it follows from Definition \ref{in-equivalent} that $\Delta|^{red}_{\mathscr M^*}\mathop  \sim \limits^{in} \Delta^*={\rm Impl}(\Lambda^*) $. Thus by the equivalence of item (i) and (iii) of Theorem \ref{in-equi},  the solutions of $\Delta$ passing through $x^0\in \mathscr M^*$ \red{are mapped, via a certain linear isomorphism S,} into the solutions of $\Delta^*$, 
	which means \red{that} $\Delta$ is internally regular if and only if \red{$\Delta^*$} has only one solution \red{passing through} any initial point in $\mathscr M^*$. This is true if and only if the input of $\Lambda^*$ is absent, i.e., $\Delta^*$ is an ODE without free variables. Therefore, $\Delta$ is internally regular if and only if $\Lambda^*$ has no inputs.
	
	$(ii)\Leftrightarrow(iii)\Leftrightarrow(vi)$:
	From the proof of Proposition \ref{Pro:M}(ii), we can see that the input is absent in $\Lambda^*$ if and only if  $\Lambda^* =MCF^2$ of $\Lambda$, that is, $MCF^1$  is absent in the \textbf{MCF} of $\Lambda$. 
	
	$(i)\Leftrightarrow(iv)\Leftrightarrow(v)$: \red{Using} $\mathscr V^*=\mathscr M^*$ and the \textbf{KCF} of $\Delta$, it is straightforward to see this equivalence.
\end{proof}
\section{Conclusion}\label{sec:conclusions}
In this paper, we propose a procedure named explicitation for DAEs. The explicitation of a DAE is, simply speaking, attaching to the DAE a class of linear control systems defined up to a coordinates change, a feedback and an output injection. We prove that the invariant subspaces of the attached control systems have direct relations with the limits of the Wong sequences of the DAE. We  show that the Kronecker invariants of the DAE have direct relations with the Morse invariants of the attached control systems, and as a consequence, the Kronecker canonical form \textbf{KCF} of the DAE and the Morse canonical from \textbf{MCF} of control systems have a perfect correspondence. We also propose a notion named internal equivalence for DAEs and show that the internal equivalence is useful when analyzing the existence and uniqueness of solutions (internal regularity). 

\section*{Appendix} 
\textbf{Kronecker Canonical Form (KCF)} \cite{kronecker1890algebraische},\cite{gantmacher1959matrix}: For any matrix pencil $sE-H\in \mathbb{R}^{l\times n}[s]$, there exist matrices $Q\in Gl(l,\mathbb{R})$, $P\in Gl(n,\mathbb{R})$ and integers $\varepsilon_1,...,\varepsilon_a\in \mathbb{N},\rho_1,...,\rho_b\in \mathbb{N},\sigma_1,...,\sigma_c\in \mathbb{N}^{+},\eta_1,...,\eta_d \in \mathbb{N}$ with $a,b,c,d\in \mathbb{N}$ such that
\begin{align*}
&Q(sE-H)P^{-1}\\&={\rm diag}\left( 
L_{\varepsilon_1}(s),...,L_{\varepsilon_a}(s),J_{\rho_1}(s),...,J_{\rho_b}(s),N_{\sigma_1}(s),...,N_{\sigma_c}(s),L^p_{\eta_1}(s),...,L^p_{\eta_d}(s)
\right) ,
\end{align*}
where \red{(omitting, for simplicity, the index $i$ of $\varepsilon_i,\rho_i,\sigma_i,\eta_i$)} the bidiagonal pencils $L_{\varepsilon}(s)\in {\mathbb{R}^{\varepsilon  \times \left( {\varepsilon  + 1} \right)}}[s]$, the real Jordan pencils $J_{\rho}(s)\in {\mathbb{R}^{\rho  \times \rho }}[s]$, the nilpotent pencils $N_{\sigma}(s)\in {\mathbb{R}^{\sigma  \times \sigma }}[s]$ and the ``per-transpose'' pencils $L^p_{\eta}(s)\in {\mathbb{R}^{\eta  \times (\eta+1 )}}[s]$ have the following form: 
$$
\begin{small}
\begin{array}{lll}
{L_\varepsilon }\left( s \right) = \left[ {\begin{smallmatrix}
	s&{ - 1}&{}&{}\\
	{}& \ddots & \ddots &{}\\
	{}&{}&s&{ - 1}
	\end{smallmatrix}} \right],   \ \ \ {N_\sigma}\left( s \right) = \left[ \begin{smallmatrix}
{ - 1}&s&{}&{}\\
{}& \ddots & \ddots &{}\\
{}&{}& \ddots &s\\
{}&{}&{}&{ - 1}
\end{smallmatrix} \right], \ \ \ \  {L^p_\eta }\left( s \right) = \left[ \begin{smallmatrix}
{ - 1}&{}&{}\\
s& \ddots &{}\\ 
{}& \ddots &{ - 1}\\
{}&{}&s
\end{smallmatrix} \right],
\\ 	{J_\rho }\left( s \right)\! =\! \left[ \begin{smallmatrix}
{s - \lambda_\rho }&{ - 1}&{}&{}\\
{}& \ddots & \ddots &{}\\
{}&{}& \ddots &{ - 1}\\
{}&{}&{}&{s - \lambda_\rho }
\end{smallmatrix} \right] \ {\rm or} \ 
{J_\rho }(s)\! =\! \left[ {\begin{smallmatrix}
	{S - {\Lambda _\rho }}&{ - I}&{}&{}\\
	{}& \ddots & \ddots &{}\\
	{}&{}& \ddots &{ - I}\\
	{}&{}&{}&{S - {\Lambda _\rho }}
	\end{smallmatrix}} \right],  \ S - \Lambda _\rho\!  =\! \left[ \begin{smallmatrix}
{s-{\phi _\rho } }&{{-\varphi _\rho }}\\
{  {\varphi _\rho }}&{s-{\phi _\rho } }
\end{smallmatrix} \right],
\end{array}
\end{small}
$$
where $\lambda_{\rho}$, $\varphi _\rho$, $\phi _\rho\in\mathbb R$. The integers $\varepsilon_i$, $\rho_i$, $\sigma_i$, $\eta_i$ are \red{called, respectively,}  Kronecker column (minimal) indices, the degrees of the finite elementary divisors,  the degrees of the infinite elementary divisors, and Kronecker row (minimal) indices. In addition, $\lambda_{\rho}$ and $\varphi_{\rho}+i\phi_{\rho}$ are the corresponding  eigenvalues of $J(s)$. These indices  and eigenvalues are invariant under external equivalence of Definition \ref{ex-equivalence}.   

\textbf{Morse Canonical Form MCF} \cite{morse1973structural},\cite{molinari1978structural}: Any control system $\Lambda=(A,B,C,D)$ is Morse equivalent to the Morse canonical form  \textbf{MCF} shown below:
\begin{equation*}
\textbf{MCF}: \left\lbrace  
\begin{array}{lll}
MCF^1: \ \dot z^1 = A^1z^1+B^1u^1&\\
MCF^2: \ \dot z^2 = A^2z^2&\\
MCF^3: \ \dot z^3 = A^3z^3+B^3u^3,&y^3 = C^3z^3+D^3u^3\\
MCF^4: \ \dot z^4 = A^4z^4,&y^4 = C^4z^4.
\end{array}\right. 
\end{equation*}
If a control system $\Lambda =(A,B,C,D)$ is in the \textbf{MCF}, then the matrices $A,B,C,D$, together with all invariants are thus given by
$$
\left[ {\begin{matrix}
	{ A}&{ B}\\
	{ C}&{ D}
	\end{matrix}} \right] = \left[ {\begin{matrix}
	{{ A^1}}&0&0&0&\vline& {{ B^1}}&0\\
	0&{{ A^2}}&0&0&\vline& 0&0\\
	0&0&{{ A^3}}&0&\vline& 0&{{ B^3}}\\
	0&0&0&{{ A^4}}&\vline& 0&0\\
	\hline
	0&0&{{ C^3}}&0&\vline& 0&{{ D^3}}\\
	0&0&0&{{ C^4}}&\vline& 0&0
	\end{matrix}} \right],
$$
(i) with $A^1={\rm diag}\{ A^1_{\varepsilon'_1 },..., A^1_{\varepsilon'_{a'} }\}$,  $B^1={\rm diag}\{ B^1_{\varepsilon'_1 },..., B^1_{\varepsilon'_{a'} }\}$,  where (\mage{throughout we omit, for simplicity, the index $i$ of $\varepsilon'_i, \rho'_i, \sigma'_i, \eta'_i$})
\begin{align*}
A_{{\varepsilon'}}^1 = \left[ {\begin{matrix}
	0&I_{{\varepsilon'}-1}\\
	0&0
	\end{matrix}} \right] \in {\mathbb{R}^{{\varepsilon'}  \times {\varepsilon'} }},\ \ \  B_{\varepsilon'} ^1 = \left[ {\begin{matrix}
	0\\
	1
	\end{matrix}} \right]\in \mathbb{R}^{{\varepsilon'} \times 1 }, 
\end{align*}
The integers $\varepsilon'_1,...,\varepsilon'_{a'}\in \mathbb{N}$ are the controllability indices of $( A^1,B^1)$.

(ii) $A^{2}={\rm diag}\{ A^2_{\rho'_1 },..., A^1_{\rho'_{b'}}\}$, where $A^2_{\rho'}$ is given by 
$$
\begin{array}{lll}
A_{\rho'}^2  = \left[ \begin{smallmatrix}
{{\lambda_{\rho'}}}&{ 1}&{}&{}\\
{}& \ddots & \ddots &{}\\
{}&{}& \ddots &{ 1}\\
{}&{}&{}&\lambda_{\rho'}
\end{smallmatrix} \right] & {\rm or} \ \
A_{\rho'}^2  = \left[ \begin{smallmatrix}
{{\Lambda_{\rho'}}}&{ I}&{}&{}\\
{}& \ddots & \ddots &{}\\
{}&{}& \ddots &{ I}\\
{}&{}&{}&\Lambda_{\rho'}
\end{smallmatrix} \right],&
\Lambda _{\rho'} = \left[ \begin{smallmatrix}
{s-{\phi _{\rho'}} }&{{-\varphi _{\rho'} }}\\
{  {\varphi _{\rho'} }}&{s-{\phi _{\rho'}} }
\end{smallmatrix} \right],
\end{array}
$$
where $\lambda_{\rho'}, \varphi _{\rho'},  \phi _{\rho'}\in\mathbb R$.

(iii) The  4-tuple $(A^3, B^3, C^3, D^3)$ is controllable and observable (prime). That is,
\begin{align}\label{Eq:prime}
\left[ {\begin{matrix}
	A^3&B^3\\
	C^3&D^3
	\end{matrix}} \right] = \left[ {\begin{matrix}
	{{\hat  A^3}}&\vline& {\hat  B^3}&0\\
	\hline
	{\hat  C^3}&\vline& 0&0\\
	0&\vline& 0&{{I_\delta }}
	\end{matrix}} \right],
\end{align}
where $\left[ {\begin{matrix}
	{{\hat  A^3}}&{{\hat  B^3}}\\
	{{\hat  C^3}}&0
	\end{matrix}} \right]$ is square and invertible and $\delta={\rm rank\,}   D^3\in \mathbb N $, and the matrices 
$$
\hat  A^3={\rm diag}\{\hat  A^3_{\sigma'_{\delta+1} },...,\hat  A^3_{\sigma'_{c'} }\},  \ \hat  B^3={\rm diag}\{\hat  B^3_{{\sigma'_{\delta+1} }},...,\hat  B^3_{{\sigma'_{c'} }}\},  \  \hat  C^3={\rm diag}\{\hat  C^3_{{{\sigma'_{\delta+1} } }},...,\hat  C^3_{{{\sigma'_{c'} } }}\},  
$$ 
where
\begin{align*}
\begin{array}{lll}
\hat  A^3_{{{\sigma' } }} = \left[ {\begin{matrix}
	0&I_{\sigma'-1}\\
	0&0
	\end{matrix}} \right] \in {\mathbb{R}^{\sigma'  \times \sigma' }},& \hat  B^3_{{{\sigma' } }}= \left[ {\begin{matrix}
	0\\
	1
	\end{matrix}} \right]\in \mathbb{R}^{\sigma' \times 1 },& \hat  C^3_{{{\sigma' } }} = \left[ {\begin{matrix}
	1&0
	\end{matrix}} \right] \in {\mathbb{R}^{1 \times\sigma'}} .
\end{array}
\end{align*}
The integers $\sigma'_1=\dots=\sigma'_{\delta}=0$,  and $\sigma'_{\delta+1},...,\sigma'_{c'}\in \mathbb{N^+} $ are the controllability indices of the pair $(\hat  A^3,\hat  B^3)$ and they are equal to the observability indices of the pair $(\hat  C^3, \hat  A^3) $.

(iv)	$ A^4={\rm diag}\{ A^4_{\eta'_1 },...,A^4_{\eta'_{d'} }\}$, $ C^4={\rm diag}\{ C^4_{\eta'_1 },..., C^4_{\eta'_{d'} }\}$, 
where 
$$
\begin{matrix}
A^4_{\eta' } = \left[ {\begin{matrix}
	0&I_{\eta'-1}\\
	0&0
	\end{matrix}} \right] \in {\mathbb{R}^{\eta'  \times \eta'}}, \ \ \ \   C^4_{\eta' }= \left[
{\begin{matrix}
	1&0
	\end{matrix}} \right] \in \mathbb{R}^{1\times \eta' }.
\end{matrix}
$$
The integers $\eta'_1,...,\eta'_{d'}\in \mathbb{N}$ are the observability indices of the pair $( C^4, A^4)$.

 Clearly, the subsystem $MCF^2$ is in the real Jordan canonical form. For the remaining subsystems $MCF^k$, denote $\mu_i=\epsilon'_i$ if $k=1$, $\mu_i=\sigma'_i$ if $k=3$, and $\mu_i=\eta'_i$ if $k=4$. Then for $k=1,3,4$, the subsystem $MCF^k$ consists of $a', c', d'$, subsystems (indexed by $i$) for which either $\mu_i\ge 1$ and then they are given by
	$$
	\dot z^{k,j}_i=\left\lbrace \begin{array}{lll}
	z^{k,j+1}_i, &1\le j\le \mu_i-1,& {\rm for} \ \ k=1,3,4,\ \ \ \  y^k_i=z^{k,1}_i, \ {\rm for} \ \ k=3,4,\\
	u^k_i, &j= \mu_i,& {\rm for} \ \ k=1,3,\\
	0, &j=\mu_i,& {\rm for} \ \ k=4,
	\end{array}\right. 
	$$
	or $\mu_i=0$ (notice that we allow for the Morse indices to be equal to zero) in which case the input $u^1$ contains components $u^1_i$ that do not affect the system at all (if $\epsilon'_i=0$), the output $y^4$ contains trivial components $y^4_i=0$ (if $\eta'_i=0$) and the output $y^3$ contains $\delta={\rm rank\,} D^3$ static relations $y^3_i=u^3_i$ (if $\sigma'_i=0$).

We call the integers $\varepsilon'_i$, $\rho'_i$, $\sigma'_i$, $\eta'_i$  the Morse indices of control systems, together with $a',b',c',d'$, $\delta$ and $\lambda_{\rho'}\in \mathbb R$ or \mage{$\lambda_{\rho'}=\varphi_{\rho'}+j\phi_{\rho'}\in \mathbb C$, with $\rho'$ taking  all values  $\rho'_i$,  where $j=\sqrt{-1}$}, they are all invariant under Morse equivalence. 

\bibliographystyle{siamplain}
\bibliography{siamref}
\end{document}

%% file: format.tex
\usepackage{lipsum}
\usepackage{amsfonts}
\usepackage{graphicx}
\usepackage{epstopdf}
\usepackage{algorithmic}
\usepackage{amsmath}
\usepackage{graphics}
\usepackage{epsfig}
\usepackage{verbatim}
\usepackage{stmaryrd}
\usepackage{epstopdf}
\usepackage{xcolor}
\usepackage{color}
\usepackage{bbm}
\usepackage{amssymb}
\usepackage{mathrsfs}
\usepackage{caption}
\usepackage{nomencl}
\usepackage{multirow}
\usepackage{chngpage}
\usepackage{array}
\usepackage{comment}
\usepackage{enumerate}  
\usepackage{graphicx}
\usepackage{transparent}
\usepackage{mathtools}
\usepackage{tikz}
\usepackage{longtable}
\usetikzlibrary{calc}
\usetikzlibrary{shapes.geometric, arrows}
\usetikzlibrary{matrix}
\usepackage[amsmath,thmmarks,hyperref]{ntheorem}
\makenomenclature
\allowdisplaybreaks

\DeclareMathAlphabet\mathbfcal{OMS}{cmsy}{b}{n}
\ifpdf
\DeclareGraphicsExtensions{.eps,.pdf,.png,.jpg}
\else
\DeclareGraphicsExtensions{.eps}

\fi

\allowdisplaybreaks
\newcommand{\red}[1]{\textcolor{black}{#1}}
\newcommand{\blue}[1]{\textcolor{black}{#1}}
\newcommand{\mage}[1]{\textcolor{black}{#1}}

\newsiamthm{lem}{\bf Lemma}
\newsiamthm{thm}{\bf Theorem}
\newsiamthm{cor}{\bf Corollary}
\newsiamthm{pro}{\bf Proposition}
\newsiamthm{alg}{\bf Algorithm}
\newsiamthm{clm}{\bf Claim}
\newsiamthm{prc}{\bf Procedure}

\newsiamremark{rem}{\bf Remark}
\newsiamremark{defn}{\bf Definition}
\newsiamremark{exa}{\bf Example}

\headers{Geometric Analysis of DAES via Linear Control Theory}{Y. Chen  and W. Respondek}

\title{Geometric Analysis of Differential-Algebraic Equations via Linear Control Theory\thanks{Submitted to the editors SIMAX.
}}
\author{Yahao Chen\thanks{Bernoulli Institute for Mathematics, Computer Science, and
		Artificial Intelligence, University of Groningen, The Netherlands
		(\email{yahao.chen@rug.nl}).}
	\and Witold Respondek\thanks{Normandie Université, INSA-Rouen, LMI, 76801 Saint-Etienne-du-Rouvray, France
		(\email{witold.respondek@insa-rouen.fr}).}}

\usepackage{amsopn}
